%% file: main.tex
\documentclass[twoside,11pt]{article}

\include{preamble}

\ShortHeadings{Interacting Particle Markov Chain Monte Carlo}{Rainforth, Naesseth, Lindsten, Paige, van de Meent, Doucet and Wood}
\firstpageno{1}

\begin{document} 

\input{coverArXiv.tex}

\title{Interacting Particle Markov Chain Monte Carlo}

\author{\name Tom Rainforth* \email twgr@robots.ox.ac.uk \\
        \addr Department of Engineering Science, University of Oxford
        \AND
        \name Christian A. Naesseth* \email christian.a.naesseth@liu.se \\
        \addr Department of Electrical Engineering, Link\"oping University
        \AND
        \name Fredrik Lindsten \email fredrik.lindsten@it.uu.se \\
        \addr Department of Information Technology, Uppsala University
        \AND
        \name Brooks Paige \email brooks@robots.ox.ac.uk \\
        \addr Department of Engineering Science, University of Oxford
        \AND
        \name Jan-Willem van de Meent \email jwvdm@robots.ox.ac.uk \\
        \addr Department of Engineering Science, University of Oxford
        \AND
        \name Arnaud Doucet \email doucet@stats.ox.ac.uk \\
        \addr Department of Statistics, University of Oxford
        \AND
        \name Frank Wood \email fwood@robots.ox.ac.uk \\
        \addr Department of Engineering Science, University of Oxford
        \\
        \emph{* equal contribution}
        }
\editor{}

\maketitle

\begin{abstract}
We introduce \emph{interacting particle Markov chain Monte Carlo} (\ipmcmc), a \pmcmc method based on an interacting pool of standard and conditional sequential Monte Carlo samplers. Like related methods, \ipmcmc is a Markov chain Monte Carlo sampler on an extended space. We present empirical results that show significant improvements in mixing rates relative to both non-interacting \pmcmc samplers, and a single \pmcmc sampler with an equivalent memory and computational budget. 
An additional advantage of the \ipmcmc method is that it is suitable for distributed and multi-core architectures.  
\end{abstract}


\setcounter{page}{1}

% ===================================
%            INTRODUCTION
% ===================================
\input{introduction.tex}

% ===================================
%               MODEL
% ===================================
\input{background.tex}
\input{method.tex}

% ===================================
%           EXPERIMENTS
% ===================================
\input{experiments.tex}

% ===================================
%           Discussion
% ===================================
\input{discussion.tex}

\section*{Acknowledgments}
Tom Rainforth is supported by a BP industrial grant. Christian A. Naesseth is supported by CADICS, a Linnaeus Center, funded by the Swedish Research Council (VR). Fredrik Lindsten is supported by the project \emph{Learning of complex dynamical systems} (Contract number: 637-2014-466) also funded by the Swedish Research Council. Frank Wood is supported under DARPA PPAML through the U.S. AFRL under Cooperative Agreement number FA8750-14-2-0006, Sub Award number 61160290-111668.

% ===================================
%            APPENDIX
% ===================================
\appendix
\input{appendix.tex}

\input{supp.tex}

\clearpage
\bibliography{refs}

\end{document}

%% file: preamble.tex
\usepackage{abbreviations}

\usepackage{microtype}

\usepackage{times}
\usepackage{graphicx} % more modern
\usepackage{subcaption} 
\graphicspath{{../figures/}}
\usepackage{tikz}
\usetikzlibrary{fit}					% fitting shapes to coordinates
\usetikzlibrary{backgrounds}	% drawing the background after the foreground
\usepackage{setspace}

\usepackage{natbib}

\usepackage{algorithm,algorithmic}

\usepackage{dsfont}
\usepackage{amssymb,amsmath}
\usepackage{naesseth}
\usepackage{bbm}

\usepackage{hyperref}
\usepackage{jmlr2e}

\usepackage{bibentry}
\nobibliography*

\newcommand\nw{\bar w} % Normalised weights
\newcommand\nz{\zeta} % Normalised conditional sampling weights

\newcommand{\angurl}{\scriptsize \url{http://www.robots.ox.ac.uk/~fwood/anglican}}
\newcommand{\myurl}{\scriptsize \url{https://bitbucket.org/twgr/ipmcmc}}

%% file: coverArXiv.tex
\newcommand{\coverTitle}{Interacting Particle Markov Chain Monte Carlo}
\newcommand{\coverAuthors}{Tom Rainforth*, Christian A. Naesseth*, Fredrik Lindsten, Brooks Paige, Jan-Willem van de Meent, Arnaud Doucet and Frank Wood}
\newcommand{\coverYear}{2016}

\begin{titlepage}
\begin{center}
{\large \em Technical report}

\vspace*{2.5cm}
%
%% TITLE
{\Huge \bfseries \coverTitle  \\[0.4cm]}

%
%% AUTHORS
{\Large \coverAuthors\\ {\normalsize * equal contribution} \\[2cm]}

\renewcommand\labelitemi{\color{red}\large$\bullet$}
\begin{itemize}
\item {\Large \textbf{Please cite this version:}} \\[0.4cm]
\large
\bibentry{rainforth-icml-2016}
%\coverAuthors. (\coverYear). \emph{\coverTitle}. JMLR: W\&CP (Vol. 48). (* equal contribution)

\end{itemize}
%{\em \coverStatus}
\vfill

\setlength{\parfillskip}{0pt plus 1fil}

\begin{abstract}
We introduce \emph{interacting particle Markov chain Monte Carlo} (\ipmcmc), a \pmcmc method based on an interacting pool of standard and conditional sequential Monte Carlo samplers. Like related methods, \ipmcmc is a Markov chain Monte Carlo sampler on an extended space. We present empirical results that show significant improvements in mixing rates relative to both non-interacting \pmcmc samplers, and a single \pmcmc sampler with an equivalent memory and computational budget. 
An additional advantage of the \ipmcmc method is that it is suitable for distributed and multi-core architectures.  
\end{abstract}

\begin{keywords}sequential Monte Carlo, Markov chain Monte Carlo, particle Markov chain Monte Carlo, parallelisation\end{keywords}

\vfill
\end{center}
\end{titlepage}

%% file: introduction.tex
% Intro, contribution statement and related work
\section{Introduction}
\label{sec:intro}
%In this paper we propose the interacting particle Markov chain Monte Carlo (\ipmcmc), a particle Markov chain Monte Carlo (\pmcmc) method that is suitable to execution on distributed and multi-core computing architectures and present empirical results suggesting significant improved mixing rates relative to both non-interacting \pg nodes and a single \pg node with an identical computational budget. 

%In this paper we present a Markov chain Monte Carlo (\mcmc) method that employs a pool of particle Markov chain Monte Carlo (\pmcmc) and sequential Monte Carlo (\smc) samplers. 
\mcmc methods are a fundamental tool for generating samples from a posterior density in Bayesian data analysis (see \eg, \citet{robert2013monte}). Particle Markov chain Monte Carlo (\pmcmc) methods, introduced by \citet{andrieuDH2010}, make use of sequential Monte Carlo (\smc) algorithms \citep{gordon1993novel,doucet2001sequential} to construct efficient proposals for the \mcmc sampler. 

One particularly widely used \pmcmc algorithm is particle Gibbs (\pg). The \pg algorithm modifies the \smc step in the \pmcmc algorithm to sample the latent variables conditioned on an existing particle trajectory, resulting in what is called a conditional sequential Monte Carlo (\csmc) step. The \pg method was first introduced as an efficient Gibbs sampler for latent variable models with static parameters \citep{andrieuDH2010}. Since then, the \pg algorithm and the extension by \citet{lindstenJS2014} have found numerous applications in \eg Bayesian non-parametrics \citep{ValeraFSPC2015,tripuraneni2015}, probabilistic programming \citep{wood2014new,vandemeent_aistats_2015} and graphical models \citep{everitt2012,naessethLS2014,naessethLS2015nested}.  

A drawback of PG is that it can be particularly adversely affected by \emph{path degeneracy} in the CSMC step.  Conditioning on an existing trajectory means that whenever resampling of the trajectories results in a common ancestor, this ancestor must correspond to this trajectory.  Consequently, the mixing of the Markov chain for the early steps in the state sequence can become very slow when the particle set typically coalesces to a single ancestor during the CSMC sweep.

In this paper we propose the interacting particle Markov chain Monte Carlo (\ipmcmc) sampler. In \ipmcmc we run a pool of \csmc and unconditional \smc algorithms as parallel processes that we refer to as nodes. After each run of this pool, we apply successive Gibbs updates to the indexes of the \csmc nodes, such that the indices of the \csmc nodes changes. Hence, the nodes from which retained particles are sampled can change from one MCMC iteration to the next. This lets us trade off exploration (\smc) and exploitation (\csmc) to achieve improved mixing of the Markov chains. Crucially, the pool provides numerous candidate indices at each Gibbs update, giving a significantly higher probability that an entirely new retained particle will be ``switched in'' than in non-interacting alternatives.

This interaction requires only minimal communication; each node must report an estimate of the marginal likelihood and receive a new role (\smc or \csmc) for the next sweep. This means that \ipmcmc is embarrassingly parallel and can be run in a distributed manner on multiple computers.

%We introduce the interacting particle Markov chain Monte Carlo (\ipmcmc) method, a new type of 
%high-dimensional \tom{The phrase high-dimensional seems a little out of place here to me.  I understand what you are saying but it is maybe a bit premature to make the point}
%Monte Carlo integration method enabling efficient 

We prove that \ipmcmc is a partially collapsed Gibbs sampler on the extended space containing the particle sets for all nodes. In the special case where \ipmcmc uses only \emph{one} \csmc node, it can in fact be seen as a non-trivial and unstudied instance of the $\alpha$-\smc-based \citep{whiteley2016} \pmcmc method introduced by \citet{huggins2015}. However, with \ipmcmc we extend this further to allow for an arbitrary number of \csmc and standard \smc algorithms with interaction. Our experimental evaluation shows that \ipmcmc outperforms both independent \pg samplers as well as a single \pg sampler with the same number of particles run longer to give a matching computational budget.

An implementation of iPMCMC is provided in the probabilistic programming system \emph{Anglican}\footnote{\angurl} \citep{wood2014new}, whilst illustrative MATLAB code, similar to that used for the experiments, is also provided\footnote{\myurl}.

%Our key contribution is the \ipmcmc method, %with a highly efficient off-the-shelf \brooks{what do you mean by ``off the shelf''?} \mcmc method, the interacting particle Markov chain Monte Carlo (\ipmcmc), that also can make systematic use of distributed or multi-core computing architectures. 

%\tom{Should we maybe add the longer run chains to results to that we can make this point clearly?}
%\fredrik{Is the last sentence ok?}

%% file: background.tex
% !TEX root = main.tex

% Here we consider model and problem formulation
\section{Background}
\label{sec:background}
We start by briefly reviewing sequential Monte Carlo \citep{gordon1993novel,doucet2001sequential} and the particle Gibbs algorithm \citep{andrieuDH2010}. Let us consider a non-Markovian latent variable model of the following form
\begin{subequations}
	\label{eq:ssm}
	\begin{alignat}{2}
	x_t | x_{1:t-1} &\sim f_t(x_t | x_{1:t-1}), \\
	y_t | x_{1:t} &\sim g_t(y_t|x_{1:t}),
	\end{alignat}
\end{subequations}
where $x_t \in \setX$ is the latent variable and $y_t \in \setY$ the observation at time step $t$, respectively,
with transition densities $f_t$ and observation densities $g_t$; $x_1$ is drawn from some initial distribution $\mu(\cdot)$. The method we propose is not restricted to the above model, it can in fact be applied to an arbitrary sequence of targets.

%We focus on the non-Markovian latent variable model because it has been shown to be useful within \eg probabilistic programming \citep{wood2014new}.%However, we restrict the exposition to the latent variable model for clarity and to see the potential usefulness of it to \eg probabilistic programming \citep{wood2014new}. \tom{I am not sure I agree with this statement - PP should actually allow the most general possible use of iPMCMC.  I think we would be better making the point that its ability to operate on arbitrary models make it a suitable candidate for PP inference engines.  By the final draft I would expect us to have it running and publically availible in Anglican}

We are interested in calculating expectations with respect to the posterior distribution $p(x_{1:T}|y_{1:T})$ on latent variables $x_{1:T} \eqdef (x_1,\ldots,x_T)$ conditioned on observations $y_{1:T} \eqdef (y_1,\ldots,y_T)$, which is proportional to the joint distribution $p(x_{1:T}, y_{1:T})$,
\begin{align}
\label{eq:jointdistribution}
p(x_{1:T} | y_{1:T}) \propto  \mu(x_1) \prod_{t=2}^T f_t(x_t | x_{1:t-1}) \prod_{t=1}^T g_t(y_t|x_{1:t}).\nonumber
\end{align}
In general, computing the posterior $p(x_{1:T}|y_{1:T})$ is intractable and we have to resort to approximations. We will in this paper focus on, and extend, the family of particle Markov chain Monte Carlo algorithms originally proposed by \citet{andrieuDH2010}. The key idea in \pmcmc is to use \smc to construct efficient proposals of the latent variables $x_{1:T}$ for an \mcmc sampler.

%
%   SMC
%

\begin{algorithm}[tb]
	\caption{Sequential Monte Carlo \hfill {\small (all for $i=1,\ldots,N$)}}
	\label{alg:smc}
	\begin{spacing}{1.2}
	\begin{algorithmic}[1]
		\STATE {\bfseries Input:} data $y_{1:T}$, number of particles $N$, proposals $q_t$
		\STATE $x_1^i \sim q_1(x_1)$
		\STATE $w_1^i = \frac{g_1(y_1|x_1^i) \mu(x_1^i)}{q_1(x_1^i)}$
		\FOR{$t = 2$ {\bfseries to} $T$}
		\STATE $a_{t-1}^i \sim \Discrete\left(\left\{\nw_{t-1}^{\ell}\right\}_{\ell=1}^N\right)$% \tom{We can maybe more general that categorical here?}
		\STATE $x_t^i \sim q_t(x_t | x_{1:t-1}^{a_{t-1}^i})$ 
		\STATE Set $x_{1:t}^i = (x_{1:t-1}^{a_{t-1}^i},x_t^i)$
		\STATE $w_t^i = \frac{g_t(y_t|x_{1:t}^i) f_t(x_t^i | x_{1:t-1}^{a_{t-1}^i})}{q_t(x_t^i|x_{1:t-1}^{a_{t-1}^i})}$
		\ENDFOR
	\end{algorithmic}
\end{spacing}
\end{algorithm}

\subsection{Sequential Monte Carlo}
\label{sec:smc}
The \smc method is a widely used technique for approximating a sequence of target distributions: in our case $p(x_{1:t}|y_{1:t}) = p(y_{1:t})^{-1} p(x_{1:t},y_{1:t}), ~t=1,\ldots,T$. 
At each time step $t$ we 
%assume we have access to 
generate a \emph{particle system}
$\{(x_{1:t}^i,w_{t}^i)\}_{i=1}^N$ which provides a weighted approximation  to $p(x_{1:t}|y_{1:t})$. Given such a weighted particle system at time $t-1$, this 
%The particle system is then
is propagated forward in time to $t$ by first drawing an ancestor variable $a_{t-1}^i$ for each particle from its corresponding distribution:
\begin{align}
\Prb(a_{t-1}^i = \ell) &= \nw_{t-1}^\ell.
&
\ell&=1,\ldots,N,
\end{align}
where $\nw_{t-1}^\ell = w_{t-1}^\ell / \sum_i w_{t-1}^i$. This is commonly known as the resampling step in the literature. We introduce the ancestor variables $\{a_{t-1}^i\}_{i=1}^N$ explicitly to simplify the exposition of the theoretical justification given in Section \ref{sec:theory}.

We continue by simulating from some given proposal density $x_t^i \sim q_t(x_t | x_{1:t-1}^{a_{t-1}^i})$ and re-weight the system of particles as follows:
\begin{align}
\label{eq:smcweights}
w_t^i = \frac{g_t(y_t|x_{1:t}^i) f_t(x_t^i | x_{1:t-1}^{a_{t-1}^i})}{q_t(x_t^i|x_{1:t-1}^{a_{t-1}^i})},
\end{align}
where $x_{1:t}^i = (x_{1:t-1}^{a_{t-1}^i},x_t^i)$. This results in a new particle system $\{(x_{1:t}^i,w_t^i)\}_{i=1}^N$ that approximates $p(x_{1:t}|y_{1:t})$. A summary is given in Algorithm~\ref{alg:smc}.

%Let $q_{\text{SMC}}(\xb^{1:N},\ab^{1:N})$ denote the joint probability distribution over $\xb^{1:N}=x_{1:T}^{1:N}, \ab^{1:N} = a_{1:t-1}^{1:N}$ induced by running Algorithm~\ref{alg:smc}. We can write the complete distribution, which will be useful for the correctness proof, as follows
%\begin{align}
%q_{\text{SMC}}(\xb^{1:N},\ab^{1:N}) = \prod_{i=1}^N q_1(x_1^i) \prod_{t=2}^T \frac{W_{t-1}^{a_{t-1}^i}}{\sum_\ell W_{t-1}^\ell} q_t(x_t^i|x_{1:t-1}^{a_{t-1}^i}).
%\end{align}

% 
%   PG
% 
\subsection{Particle Gibbs}
\label{sec:pg}
The \pg algorithm \citep{andrieuDH2010} is a Gibbs sampler on the extended space composed of all random variables generated at one iteration, which still retains the original target distribution as a marginal. Though \pg allows for inference over both latent variables and static parameters, we will in this paper focus on sampling of the former.  The core idea of \pg is to iteratively run \emph{conditional} sequential Monte Carlo (\csmc) sweeps as shown in Algorithm~\ref{alg:csmc}, whereby each conditional trajectory is sampled from the surviving trajectories of the previous sweep.  This \emph{retained particle} index, $b$, is sampled with probability proportional to the final particle weights $\bar{w}^i_T$.

\begin{algorithm}[tb]
	\caption{Conditional sequential Monte Carlo}
	\label{alg:csmc}
	\begin{spacing}{1.2}
	\begin{algorithmic}[1]
		\STATE {\bfseries Input:} data $y_{1:T}$, number of particles $N$, proposals $q_t$, conditional trajectory $x_{1:T}'$
		\STATE $x_1^i \sim q_1(x_1), ~i=1,\ldots,N-1$ and set $x_1^N = x_1'$
		\STATE $w_1^i = \frac{g_1(y_1|x_1^i) \mu(x_1^i)}{q_1(x_1^i)}, ~i=1,\ldots,N$
		\FOR{$t = 2$ {\bfseries to} $T$}
		\STATE $a_{t-1}^i \sim \Discrete\left(\left\{\nw_{t-1}^\ell\right\}_{\ell=1}^N\right), ~i=1,\ldots,N-1$
		\STATE $x_t^i \sim q_t(x_t | x_{1:t-1}^{a_{t-1}^i}), ~i=1,\ldots,N-1$
		\STATE Set $a_{t-1}^N = N$ and $x_t^N = x_t'$
		\STATE Set $x_{1:t}^i = (x_{1:t-1}^{a_{t-1}^i},x_t^i), ~i=1,\ldots,N$
		\STATE $w_t^i = \frac{g_t(y_t|x_{1:t}^i) f_t(x_t^i | x_{1:t-1}^{a_{t-1}^i})}{q_t(x_t^i|x_{1:t-1}^{a_{t-1}^i})}, ~i=1,\ldots,N$
		\ENDFOR
	\end{algorithmic}
\end{spacing}
\end{algorithm}

%% file: method.tex
\section{Interacting Particle Markov Chain Monte Carlo}
\label{sec:method}
%Our main goal with the interacting particle Markov chain Monte Carlo method is to increase the efficiency of particle \mcmc, particle Gibbs especially, by coupling independent \smc methods with \csmc algorithms, perhaps running on different nodes or threads. The method we propose, interacting particle Markov chain Monte Carlo (\ipmc), makes efficient use of multi-core and distributed computing architectures to increase accuracy of the Monte Carlo sampler.

The main goal of \ipmcmc is to increase the efficiency of \pmcmc, in particular particle Gibbs. The basic \pg algorithm is especially susceptible to the \emph{path degeneracy} effect of \smc samplers, \ie sample impoverishment due to frequent resampling.  Whenever the ancestral lineage collapses at the early stages of the state sequence, the common ancestor is, by construction, guaranteed to be equal to the retained particle.  This results in high correlation between the samples, and poor mixing of the Markov chain. %Since we force one trajectory, the conditional part, to survive to the end it means that for early time steps we will almost always pick the corresponding sample from last iteration. 
To counteract this we might need a very high number of particles to get good mixing for all latent variables $x_{1:T}$, which can be infeasible due to e.g.~limited available memory. \ipmc can alleviate this issue by, from time to time, switching out a \csmc particle system with a completely independent \smc one, resulting in improved mixing.

\ipmcmc, summarized in Algorithm~\ref{alg:ipmc}, consists of $M$ interacting separate \csmc and \smc algorithms, exchanging only very limited information at each iteration to draw new \mcmc samples. We will refer to these internal \csmc and \smc algorithms as nodes, and assign an index $m=1,\ldots,M$. 
At every iteration, we have $P$ nodes running local \csmc algorithms, with the remaining
$M-P$ nodes running independent \smc.
The \csmc nodes are given an identifier $c_j \in \{1,\ldots,M\}, ~j=1,\ldots,P$ with $c_j \neq c_k,~k \neq j$ and we write $c_{1:P} = \{c_1,\ldots,c_P\}$. Let $\xb_m^i = x_{1:T,m}^i$ be the internal particle trajectories of node $m$.

\begin{algorithm}[tb]
	\caption{\ipmcmc sampler}
	\label{alg:ipmc}
	\begin{spacing}{1.2}
	\begin{algorithmic}[1]
		\STATE {\bfseries Input:} number of nodes $M$, conditional nodes $P$ and \mcmc steps $R$, initial $\xb_{1:P}'[0]$
		\FOR{$r = 1$ {\bfseries to} $R$}
		\STATE Workers $1:M \backslash c_{1:P}$ run Algorithm~\ref{alg:smc} (\smc)
		\STATE Workers $c_{1:P}$ run Algorithm~\ref{alg:csmc} (\csmc), conditional on $\xb_{1:P}'[r-1]$ respectively.
		\FOR{$j = 1$ {\bfseries to} $P$}
		\STATE Select a new conditional node by simulating $c_j$ according to \eqref{eq:simConditional}. %with probability 
		\STATE Set new \mcmc sample $\xb_j'[r] = \xb_{c_j}^{b_j}$ by simulating $b_j$ according to~\eqref{eq:simTrajectory}
		\ENDFOR
		\ENDFOR
	\end{algorithmic}
\end{spacing}
\end{algorithm}

Suppose we have access to $P$ trajectories ${\xb_{1:P}'[0]=(\xb_1'[0],\ldots,\xb_P'[0])}$ corresponding to the initial retained particles, where the index $[\cdot]$ denotes \mcmc iteration. At each iteration $r$, the nodes $c_{1:P}$ run \csmc (Algorithm~\ref{alg:csmc}) with the previous \mcmc sample $\xb_j'[r-1]$ as the retained particle. The remaining $M-P$ nodes run standard (unconditional) \smc, \ie Algorithm~\ref{alg:smc}.  Each node $m$ returns an estimate of the marginal likelihood for the internal particle system defined as
\begin{align}
\label{eq:ML}
\hat Z_{m} = \prod_{t=1}^T \frac{1}{N} \sum_{i=1}^N w_{t,m}^{i}.
\end{align}
%along with a single trajectory $\xb_{m}^{s_m}$ sampled according to 
%\begin{align}
%\label{eq:simTrajectory}
%\Prb(s_m = i) &= \nw_{T,m}^i
%\end{align}
%where $s_m$ is the local index of the sampled particle.  

The new conditional nodes are then set using a single loop $j=1:P$ of Gibbs updates, sampling new indices $c_j$ where
\begin{align}
\label{eq:simConditional}
&\Prb(c_j = m|c_{1:P\backslash j}) = \hat\nz_{m}^j \\
\label{eq:zeta_def}
\mathrm{and}  \quad &{\hat\nz_{m}^j = \frac{\hat Z_{m} \iden_{m \notin c_{1:P \backslash j}}}{ \sum_{n=1}^M \hat Z_{n} \iden_{n \notin c_{1:P \backslash j}}}},
%\Prb(c_j = m|c_{1:P\backslash j}) = \frac{\hat Z_{\pi_T,m} \iden_{m \notin c_{1:P \backslash j}}}{\sum_{n=1}^M \hat Z_{\pi_T,n} \iden_{n \notin c_{1:P \backslash j}}},
\end{align}
defining ${c_{1:P\backslash j} = \{c_1,\ldots,c_{j-1},c_{j+1},\ldots,c_P\}}$.  We thus loop once through the conditional node indices and resample them from the union of the current node index and the unconditional node indices\footnote{Unconditional node indices here refers to all $m \notin c_{1:P}$ at that point in the loop. It may thus include nodes who just ran a CSMC sweep, but have been ``switched out'' earlier in the loop.}, in proportion to their marginal likelihood estimates.  This is the key step that lets us switch completely the nodes from which the retained particles are drawn.

One \mcmc iteration $r$ is concluded by setting the new samples $\xb_{1:P}'[r]$ by simulating from the corresponding conditional node's, $c_j$, internal particle system
\begin{align}
\label{eq:simTrajectory}
\Prb(b_j = i | c_j) &= \nw_{T,c_j}^i, \nonumber\\%\frac{w_{T,c_j}^{i}}{\sum_\ell w_{T,c_j}^{\ell}}, \nonumber\\
\xb_j'[r] &= \xb_{c_j}^{b_j}.
\end{align}
The potential to pick from updated nodes $c_j$, having run independent \smc algorithms, decreases correlation and improves mixing of the  \mcmc sampler. Furthermore, as each Gibbs update corresponds to a one-to-many comparison for maintaining the same conditional index, the probability of switching is much higher than in an analogous non-interacting system.
%A summary of the \ipmcmc algorithm can be found in Algorithm~\ref{alg:ipmc}. % and a high-level overview of the algorithm flow can be found in Figure~\ref{fig:algorithm}.

The theoretical justification for iPMCMC is independent of how the initial trajectories $\xb_{1:P}'[0]$ are generated.  One simple and effective method (that we use in our experiments) is to run standard SMC sweeps for the ``conditional'' nodes at the first iteration.

The \ipmc samples $\xb_{1:P}'[r]$ can be used to estimate expectations for test functions $f: \setX^T \mapsto \reals$ in the standard Monte Carlo sense, with
\begin{align}
\label{eq:mcestimate}
\E[f(\xb)] \approx \frac{1}{RP}\sum_{r=1}^R\sum_{j=1}^P f(\xb_j'[r]).
\end{align}
However, we can improve upon this if we have access to all particles generated by the algorithm, see Section~\ref{sec:allparticles}.

We note that \ipmcmc is suited to distributed and multi-core architectures. In practise, the particle to be retained, should the node be a conditional node at the next iteration, can be sampled upfront and discarded if unused.  Therefore, at each iteration, only a single particle trajectory and normalisation constant estimate need be communicated between the nodes, whilst the time taken for calculation of the updates of $c_{1:P}$ is negligible.  Further, iPMCMC should be amenable to an asynchronous adaptation under the assumption of a random execution time, independent of $\xb_j'[r-1]$ in Algorithm~\ref{alg:ipmc}. We leave this asynchronous variant to future work.
%\brooks{maybe mention right away that we can improve this?}

%The method we propose is to couple a conditional \smc method with several independent standard \smc{s}. Assume that we have one \csmc and $M-1$ standard \smc processes running, each supplying an estimate of the normalisation constant $p_m^N(y_{1:T}), m=1,\ldots,M$ based on its $N$ particles. 

%The method proceeds by, at each \mcmc iteration, drawing the next process $m^*$ to become the \csmc according to the estimates $\{p_m^N(y_{1:T})\}$. Then, the conditional trajectory $x_{1:t}'$ is drawn from the corresponding process' internal particle representation either by simulating from the final weights or by doing backward simulation \citet{TODO}.

%\begin{algorithm}[tb]
%\caption{Parallel iterated \csmc}
%\label{alg:picsmc}
%\begin{enumerate}
%\item Initialize $x_{1:T}'[0]$ and set $m^* = 1$
%\item \textbf{for $r=1$ to $R$}
%\begin{enumerate}
%\item In process $m \neq m^*$ run $\operatorname{SMC}$
%\item In process $m = m^*$ run $\operatorname{CSMC}(x_{1:T}'[r-1])$
%\item Simulate $m^* \sim \cat \left( \frac{p_m^N(y_{1:T})}{\sum_{\ell=1}^M p_\ell^N(y_{1:T})} \right)$
%\item Extract $x_{1:T}'[r]$ from process $m^*$
%\end{enumerate}
%\end{enumerate}
%\end{algorithm}
%We do this by picking the next conditional trajectory $x_{1:T}'$ in the following way. First, draw a categorical random variable according to

% Theoretical foundation
\subsection{Theoretical Justification}
\label{sec:theory}
In this section we will give some crucial results to justify the proposed \ipmc sampler. This section is due to space constraints fairly brief and it is helpful to be familiar with the proof of \pg in \citet{andrieuDH2010}.
We start by defining some additional notation.
% Let $\nw_t^i \eqdef w_t^i/\sum_\ell w_t^\ell$ denote the normalised importance weights,
% and let 
Let $\xib \eqdef \{x_{t}^i\}_{\substack{i=1:N\\t=1:T}} \bigcup \{a_{t}^i\}_{\substack{i=1:N\\t=1:T-1}}$
denote all generated particles and ancestor variables of a (C)\smc sampler.
We write $\xib_m$ when referring to the variables of the sampler local to node $m$.
%
% We start by denoting the internal particle system, \ie all generated particles and ancestor variables, of node $m$ by $\xib_m \eqdef \{x_{t,m}^i\}_{\substack{i=1:N\\t=1:T}} \bigcup \{a_{t,m}^i\}_{\substack{i=1:N\\t=1:T-1}}$.
Let the conditional particle trajectory and corresponding ancestor variables for node $c_j$ be denoted by $\{\xb_{c_j}^{b_j},\bb_{c_j}\}$, with $\bb_{c_j} = (\beta_{1,c_j},\ldots,\beta_{T,c_j})$,
%\fredrik{I changed to lower case $b_j$ here to agree with (7), but the notation might be a bit confusing. Use a different symbol in (7) or here?}
$\beta_{T,c_j} = b_j$ and $\beta_{t,c_j} = a_{t,c_j}^{\beta_{t+1,c_j}}$. %\tom{Feel that mabye this could do with an extra line of explanation for people not already familiar with this notation from the PMCMC paper}. 
Let the posterior distribution of the latent variables be denoted by $\pi_T(\xb) \eqdef p(x_{1:T}|y_{1:T})$ with normalisation constant $Z \eqdef p(y_{1:T})$. 
Finally we % the induced distributions of the \smc and \csmc algorithms
note that the \smc and \csmc algorithms induce the respective distributions over the random variables
generated by the procedures:
\vspace{-2mm}
\begin{align*}
q_{\text{SMC}}(\xib) &= \prod_{i=1}^N q_1(x_{1}^i) \cdot \prod_{t=2}^T \prod_{i=1}^N \left[ 
\nw_{t-1}^{a_{t-1}^i}
%\frac{w_{t-1}^{a_{t-1}^i}}{\sum_\ell w_{t-1}^{\ell}}
q_t(x_{t}^i|x_{1:t-1}^{a_{t-1}^i}) \right], \\
q_{\text{CSMC}}\left(\xib \backslash \{\xb', \bb\} \mid \xb', \bb \right) &= \prod_{\substack{i=1\\i\neq b_{1}}}^N  q_1(x_{1}^i) \cdot \prod_{t=2}^T \prod_{\substack{i=1\\i\neq b_{t}}}^N \left[
\nw_{t-1}^{a_{t-1}^i}
%\frac{w_{t-1}^{a_{t-1}^i}}{\sum_\ell w_{t-1}^\ell}
q_t(x_{t}^i|x_{1:t-1}^{a_{t-1}^i})\right].
\end{align*}
Note that running Algorithm~\ref{alg:csmc} corresponds to simulating from $q_\text{CSMC}$ using a fixed
choice for the index variables $\bb = (N\,\ldots,N)$. While these indices are used to facilitate the
proof of validity of the proposed method, they have no practical relevance and can thus be set to arbitrary
values, as is done in Algorithm~\ref{alg:csmc}, in a practical implementation.

% \begin{align*}
% &q_{\text{SMC}}(\xib_m) = \nonumber \\
% &\prod_{i=1}^N q_1(x_{1,m}^i) \cdot \prod_{t=2}^T \prod_{i=1}^N \left[ \frac{w_{t-1,m}^{a_{t-1,m}^i}}{\sum_\ell w_{t-1,m}^\ell} q_t(x_{t,m}^i|x_{1:t-1,m}^{a_{t-1,m}^i}) \right],\\
% &q_{\text{CSMC}}\left(\xib_m \backslash \{\xb_m^j, \bb_m\} \mid \xb_m^j, \bb_m, m, j \right) = \nonumber \\
% &\prod_{\substack{i=1\\i\neq b_{1,m}}}^N  q_1(x_{1,m}^i) \cdot \prod_{t=2}^T \prod_{\substack{i=1\\i\neq b_{t,m}}}^N \left[\frac{w_{t-1,m}^{a_{t-1,m}^i}}{\sum_\ell w_{t-1,m}^\ell} q_t(x_{t,m}^i|x_{1:t-1,m}^{a_{t-1,m}^i})\right].
% \end{align*}
Now we are ready to state the main theoretical result.
\begin{theorem}
	\label{thm:one}
	The interacting particle Markov chain Monte Carlo sampler of Algorithm~\ref{alg:ipmc} is a partially collapsed Gibbs sampler \citep{van2008partially} for the target distribution
	\begin{align}
	\label{eq:targetdistribution}
	&\tilde \pi(\xib_{1:M}, c_{1:P}, b_{1:P}) =  \nonumber \\
	&\frac{1}{N^{PT} \binom{M}{P}} \prod_{\substack{m=1\\m\notin c_{1:P}}}^M q_{\text{SMC}}\left(\xib_m\right) \cdot \prod_{j = 1}^P \left[ \pi_T\left(\xb_{c_j}^{b_j}\right) \iden_{c_j \notin c_{1:j-1}} 
	q_{\text{CSMC}}\left(\xib_{c_j} \backslash \{\xb_{c_j}^{b_j}, \bb_{c_j}\} \mid \xb_{c_j}^{b_j}, \bb_{c_j}\right) \right].
	\end{align}
\end{theorem}
\begin{proof}
	See Appendix~\ref{sec:proof} at the end of the paper. %\tom{We should maybe make it clear that this is at the end of the paper rather than in the supplementary material}
\end{proof}
\newtheorem{rem}{Remark}
\begin{rem}
	The marginal distribution of $(\xb_{c_{1:P}}^{b_{1:P}},c_{1:P},b_{1:P})$, with $\xb_{c_{1:P}}^{b_{1:P}} = (\xb_{c_1}^{b_1},\ldots,\xb_{c_P}^{b_P})$, under \eqref{eq:targetdistribution} is given by
	\begin{align}
	\label{eq:marginaldistribution}
	&\tilde \pi\left(\xb_{c_{1:P}}^{b_{1:P}},c_{1:P},b_{1:P}\right) = \frac{\prod_{j = 1}^P \pi_T\left(\xb_{c_j}^{b_j}\right) \iden_{c_j \notin c_{1:j-1}}}{N^{PT} \binom{M}{P}} .
	\end{align}
	This means that each trajectory $\xb_{c_j}^{b_j}$ is marginally distributed according to the
	posterior distribution of interest, $\pi_T$. Indeed, the $P$ retained trajectories of \ipmcmc
	will in the limit $R \rightarrow \infty$ %(or $N \rightarrow \infty$) 
	be independent draws from $\pi_T$.
	% , by targeting \eqref{eq:targetdistribution} with an \mcmc sampler we will, in the limit, equivivalently draw $P$ exact samples from $\pi_T$ as a byproduct at each iteration.
\end{rem}
Note that adding a backward or ancestor simulation step can drastically increase mixing when sampling the conditional trajectories $\xb_j'[r]$ \citep{lindsten2013backward}. In the \ipmcmc sampler we can replace simulating from the final weights on line~7 by a backward simulation step. Another option for the \csmc nodes is to replace this step by internal ancestor sampling \citep{lindstenJS2014} steps and simulate from the final weights as normal.
%According to the above remark if we can generate samples from \eqref{eq:targetdistribution} we will equivivalently draw exact samples from $\pi_T$ as a byproduct. We formalize this and show a basic convergence result of the \mcmc samples $\xb^{1:P}[r]$ generated by running Algorithm~\ref{alg:ipmc}.
%\begin{theorem}[Convergence]
%\begin{align}
%\label{eq:convergence}
%\|\mathcal{L}\{\xb^j[r] \in \cdot \}-\pi_T(\cdot) \| \rightarrow 0, \quad \text{as}~~r\to\infty,~\forall j=1,\ldots,P
%\end{align}
%\end{theorem}
%\begin{proof}
%See Section~\ref{} in the supplementary material.
%\end{proof}

% Parameters, Rao-Blackwellization, BS etc.
%\subsection{Extensions and Improvements}
%We will here discuss some potential extensions and improvements on the \ipmcmc. The first one, making use of all particles, we use in the experiments and the following two, introducing backward simulation and asynchronous implementation, we leave for future work.

% How to use all particles when estimating expectations
\subsection{Using All Particles}
\label{sec:allparticles}
At each \mcmc iteration $r$, we generate $MN$ full particle trajectories. Using only $P$ of these as in \eqref{eq:mcestimate} might seem a bit wasteful. We can however make use of all particles to estimate expectations of interest by, for each Gibbs update $j$, averaging over the possible new values for the conditional node index $c_j$ and corresponding particle index $b_j$. We can do this by replacing $f(\xb_j'[r])$ in \eqref{eq:mcestimate} by
\begin{align*}
\E_{c_j|c_{1:P\backslash j}}\left[\E_{b_j | c_j}\left[f(\xb_j'[r])\right]\right] 
= \sum_{m=1}^M 
\hat\nz_{m}^j
%\frac{\hat Z_{\pi_T,m} \iden_{m \notin c_{1:P \backslash j}}}{\sum_n \hat Z_{\pi_T,n} \iden_{n \notin c_{1:P \backslash j}}} 
\sum_{i=1}^N
\nw_{T,m}^i f(\xb_{m}^i).
%\frac{w_{T,m}^{i}}{\sum_\ell w_{T,m}^{\ell}} f(\xb_{m}^i),
\end{align*}
This procedure is referred to as a Rao-Blackwellization of a statistical estimator and is (in terms of variance) never worse than the original one.  We highlight that each $\hat\nz_{m}^j$, as defined in~\eqref{eq:zeta_def}, depends on which indices are sampled earlier in the index reassignment loop.  Further details, along with a derivation, are provided in Appendix~\ref{sec:suppUseAll}.

%\tom{I feel it might be better to move the asynchronous and backward simulation to the discussion section given we are leaving them for future work?}

% Introduce \theta and sample that as well
%\subsubsection{Static Parameters}
%We can also use the \ipmcmc algorithm to sample static parameters $\theta$ in a Gibbs-type method by alternating between sampling latent variables $x_{1:T}$ and $\theta$. The basic version of this would add two extra steps in Algorithm~\ref{alg:ipmc} by simulating new parameters from the full conditional distributions for the \csmc nodes and the parameter prior for the standard \smc nodes.

\subsection{Choosing P}
\label{sec:choosingP}
Before jumping into the full details of our experimentation, we quickly consider the choice of $P$. Intuitively we can think of the independent \smc's as particularly useful if they are selected as the next conditional node. The probability of the event %(denoted by $\mathcal{S}$) 
that at least one conditional node switches with an unconditional, is given by
\begin{align}
\label{eq:switchingprob}
%\Prb(\mathcal{S}) 
\Prb(\{\text{switch}\}) 
= 1 - \E\Big[\prod_{j=1}^P \frac{\hat Z_{c_j}}{\hat Z_{c_j} +\sum_{m \notin c_{1:P}}^M \hat Z_{m}}\Big].
\end{align}
There exist theoretical and experimental results \citep{pitt2012some,berard2014lognormal,doucet2015efficient} that show that the distributions of the normalisation constants are well-approximated by their log-Normal limiting distributions. Now, with $\sigma^2$ ($\propto \frac{1}{N}$) being the variance of the (C)\smc estimate, it means we have $\log \left(Z^{-1} \hat Z_{c_j} \right) \sim \N(\frac{\sigma^2}{2},\sigma^2)$ and $\log \left(Z^{-1} \hat Z_{m} \right) \sim \N(-\frac{\sigma^2}{2},\sigma^2)$, $m\notin c_{1:P}$ at stationarity, where $Z$ is the true normalization constant. Under this assumption, we can accurately estimate the probability \eqref{eq:switchingprob} for different choices of $P$ an example of which is shown in Figure~\ref{fig:theSwitchingProb} along with additional analysis in Appendix~\ref{sec:supp-choosep}. These provide strong empirical evidence that the switching probability is maximised for $P=M/2$.

\begin{figure}[t]
	\centering
	\begin{subfigure}[t]{0.4\textwidth}
		\includegraphics[height=0.75\textwidth]{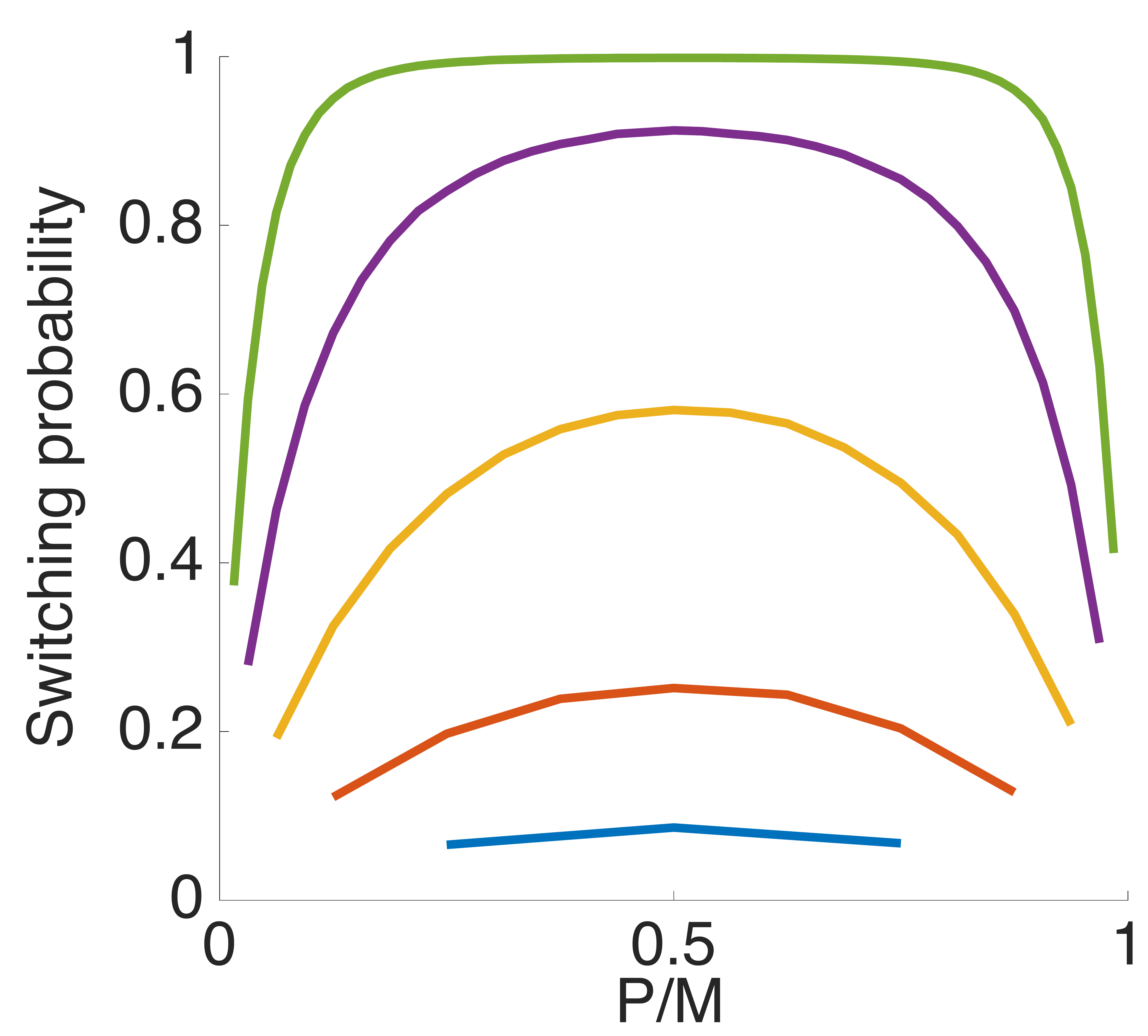}
		\caption{Limiting log-Normal\label{fig:theSwitchingProb}}
	\end{subfigure}
	~~~~~~~~~~~
	\begin{subfigure}[t]{0.4\textwidth}
		\includegraphics[height=0.75\textwidth]{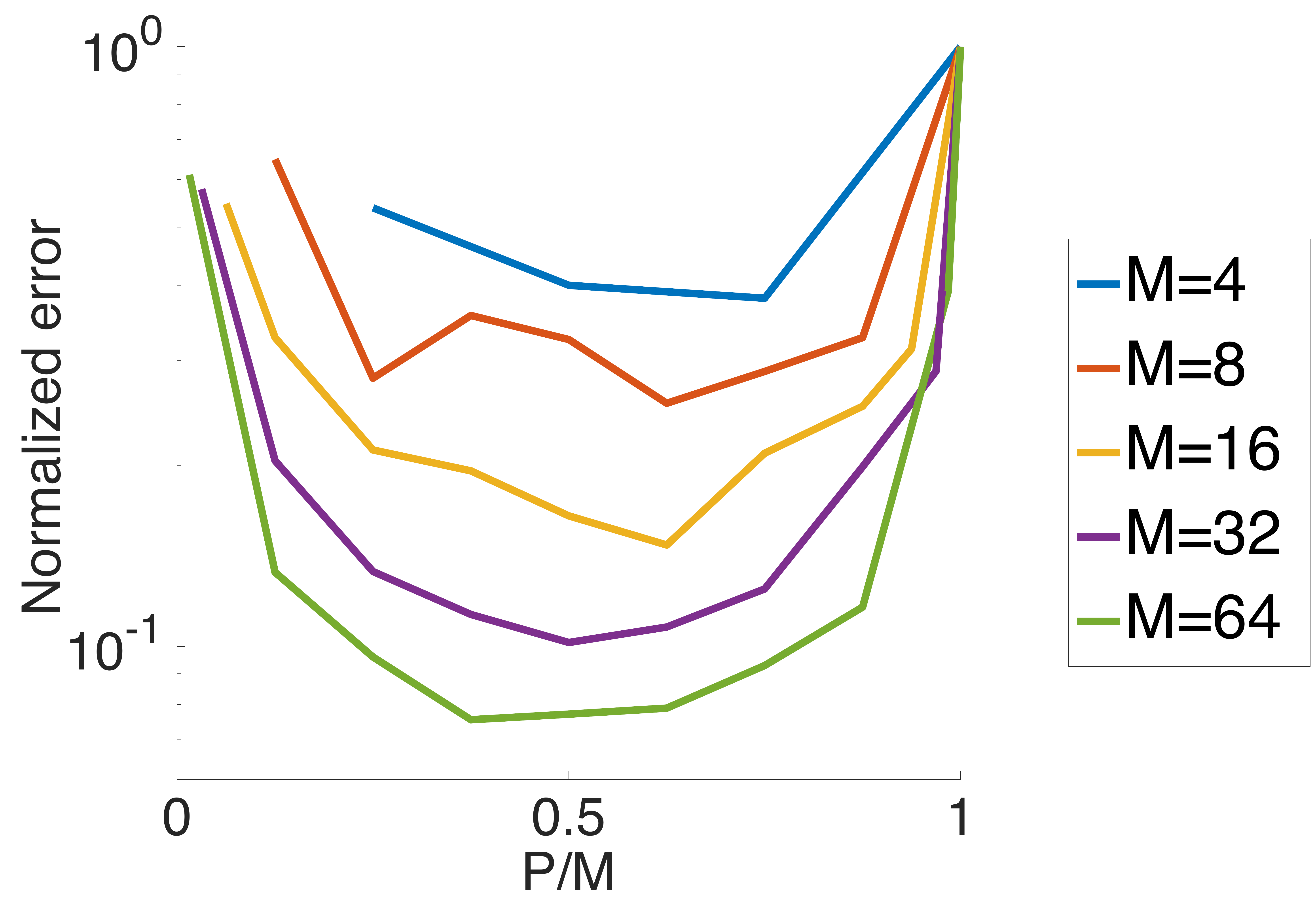}
		\caption{Gaussian state space model\label{fig:Psweep}}
	\end{subfigure}
	\caption{a) Estimation of switching probability for different choices of P and M assuming the log-Normal limiting distribution for $\hat{Z}_m$ with $\sigma=3$. b) Median error in mean estimate for different choices of P and M over 10 different synthetic datasets of the linear Gaussian state space model given in~\eqref{eq:LGSS} after 1000 MCMC iterations. Here errors are normalized by the error of a multi-start PG sampler which is a special case of iPMCMC for which $P=M$ (see Section \ref{sec:experiments}).
	}
\end{figure}

%
%\begin{figure}[h]
%
%\end{figure}
%~ %add desired spacing 

In practice we also see that best results are achieved when $P$ makes up roughly half of the nodes, see Figure~\ref{fig:Psweep} for performance on the state space model introduced in~\eqref{eq:LGSS}. Note also that the accuracy seems to be fairly robust with respect to the choice of $P$. % At least for this model, three things are clear from this sweep - firstly the optimal choice for the ratio of P/M is roughly 1/2, secondly that the performance of iPMCMC is relatively robust to changes in P around this optimum and thirdly that as $M$ increases, to relatively more preferable iPMCMC is to the trivial distribution of PG given by $P=M$ (this reason this occurs is discussed in more detail in Section \ref{sec:discussion}). 
Based on these results, we set the value of $P=\frac{M}{2}$ for the rest of our experiments.

%% file: experiments.tex
% Numerical experiments section
\section{Experiments}
\label{sec:experiments}

To demonstrate the %validity and
empirical performance of iPMCMC we report experiments on two state space models.  
Although both the models considered are Markovian, we emphasise that iPMCMC goes far beyond this and can be applied to arbitrary graphical models. 
%For exposition we will focus our comparison to the trivial distribution, whereby $M$ independent PMCMC samplers are run in parallel, of PG, particle independent Metropolis-Hastings (PIMH) \cite{andrieuDH2010} and the alternate move PG sampler (APG) \cite{holenstein2009particle}. 
We will focus our comparison on the trivially distributed alternatives, whereby $M$ independent PMCMC samplers are run in parallel--these are PG, particle independent Metropolis-Hastings (PIMH) \cite{andrieuDH2010} and the alternate move PG sampler (APG) \cite{holenstein2009particle}. Comparisons to other alternatives, including independent SMC, serialized implementations of PG and PIMH, and running a mixture of independent PG and PIMH samplers, are provided in Appendix \ref{sec:supp-additionalFigures}.  None outperformed the methods considered here, with the exception of running a serialized PG implementation with an increased number of particles, requiring significant additional memory ($O(MN)$ as opposed to $O(M+N)$).

In PIMH a new particle set is proposed at each \mcmc step using an independent \smc sweep, which is then either accepted or rejected using the standard Metropolis-Hastings acceptance ratio. % \cite{andrieuDH2010}. %\cite{hastings1970monte}. 
APG interleaves PG steps with PIMH steps
%, alternating between \csmc updates and a Metropolis-Hastings step with an independent \smc proposal,
in an attempt to overcome the issues caused by path degeneracy in PG.  We refer to the trivially distributed versions of these algorithms as multi-start PG, PIMH and APG respectively (mPG, mPIMH and mAPG). 
We use Rao-Blackwellization, as described in \ref{sec:allparticles}, to average over all the generated particles for all methods, weighting the independent Markov chains equally for mPG, mPIMH and mAPG. We note that mPG is a special case of iPMCMC for which $P=M$.  For simplicity, multinomial resampling was used in the experiments, with the prior transition distribution of the latent variables taken for the proposal.  $M=32$ nodes and $N=100$ particles were used unless otherwise stated.  Initialization of the retained particles for iPMCMC and mPG was done by using standard SMC sweeps.

\subsection{Linear Gaussian State Space Model}
\label{sec:LGSS}
We first consider a linear Gaussian state space model (LGSSM) with 3 dimensional latent states $x_{1:T}$, 20 dimensional observations $y_{1:T}$ and dynamics given by %\cn{It seems you use deterministic initial conditions for $x_0$, reformulate the model so you have a prior $\mu(x_1)$ instead? (also follows notation above)}
\begin{subequations}
	\label{eq:LGSS}
	\begin{align}
	x_1 & \sim \mathcal{N} \left(\mu, V\right) \label{eq:LGSSa}\\
	x_t & = \alpha x_{t-1} + \delta_{t-1} \quad & \delta_{t-1} \sim \mathcal{N} \left(0, \Omega\right) \label{eq:LGSSb}\\
	y_t & = \beta x_{t} + \varepsilon_{t} \quad & \varepsilon_{t} \sim \mathcal{N} \left(0, \Sigma\right).
	\label{eq:LGSSc}
	\end{align}
\end{subequations}
We set $\mu = [0, 1, 1]^T$, $V = 0.1 \; \mathbf{I}$, $\Omega = \mathbf{I}$ and $\Sigma = 0.1 \; \mathbf{I}$ where $\mathbf{I}$ represents the identity matrix.  The constant transition matrix, $\alpha$, corresponds to successively applying rotations of $\frac{7\pi}{10}$, $\frac{3\pi}{10}$ and $\frac{\pi}{20}$ about the first, second and third dimensions of $x_{t-1}$ respectively followed by a scaling of $0.99$ to ensure that the dynamics remain stable.  A total of 10 different synthetic datasets of length $T=50$ were generated by simulating from~\eqref{eq:LGSSa}--\eqref{eq:LGSSc}, each with a different emission matrix $\beta$ generated by sampling each column independently from a symmetric Dirichlet distribution with concentration parameter 0.2.

Figure \ref{fig:meanConv} shows convergence in the estimate of the latent variable means to the ground-truth solution for iPMCMC and the benchmark algorithms as a function of MCMC iterations.  It shows that iPMCMC comfortably outperforms the alternatives from around 200 iterations onwards, with only iPMCMC and mAPG demonstrating behaviour consistent with the Monte Carlo convergence rate, suggesting that mPG and mPIMH are still far from the ergodic regime.  Figure \ref{fig:meanPos} shows the same errors after $10^4$ MCMC iterations as a function of position in state sequence.  This demonstrates that iPMCMC outperformed all the other algorithms for the early stages of the state sequence, for which mPG performed particularly poorly. Toward the end of state sequence, iPMCMC, mPG and mAPG all gave similar performance, whilst that of mPIMH was significantly worse.

\begin{figure*}[t]
	\centering
	\begin{subfigure}[t]{0.49\textwidth}
		\includegraphics[width=\textwidth]{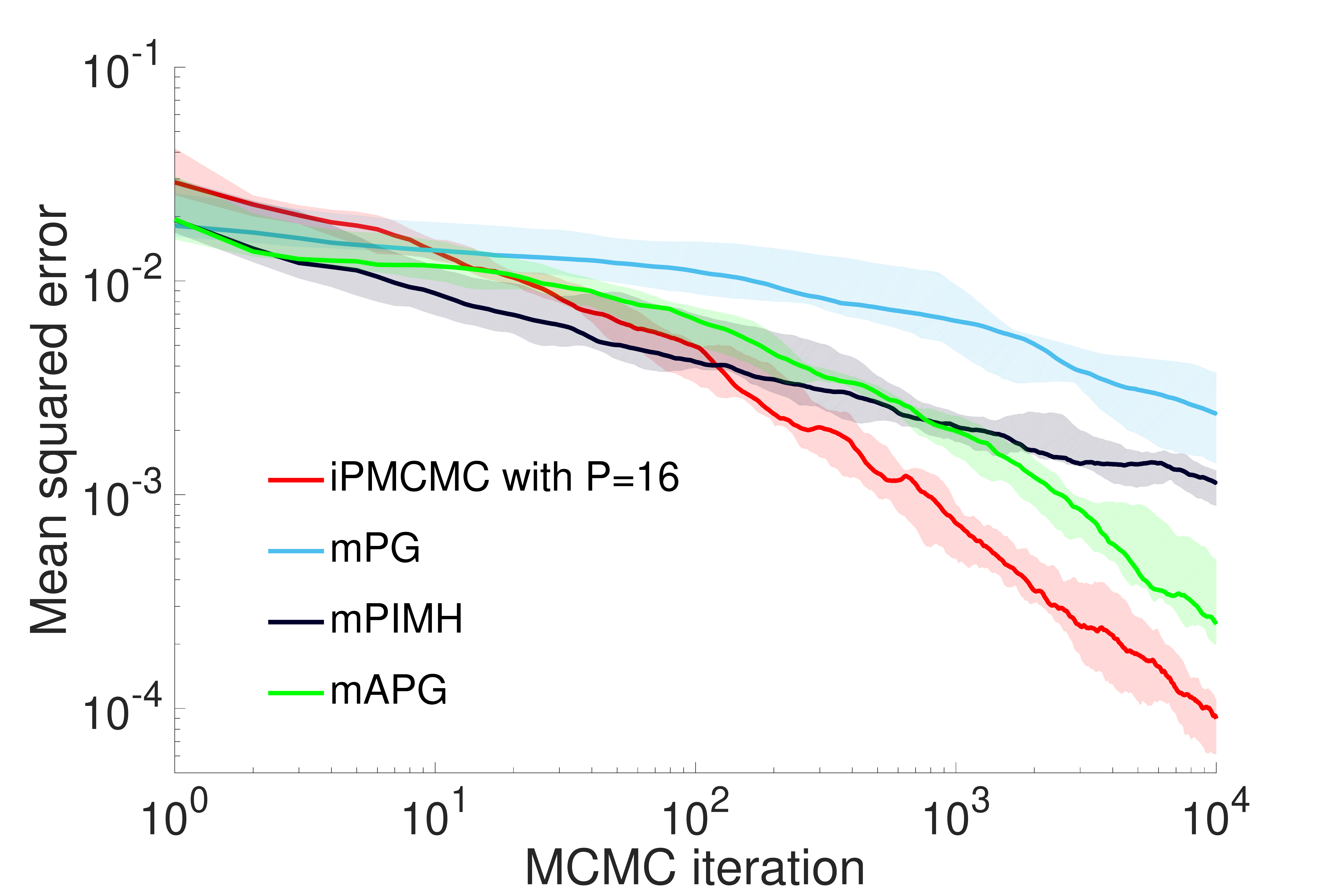}
		\caption{Convergence in mean for full sequence}
		\label{fig:meanConv}
	\end{subfigure}
	~  %add desired spacing between images, e. g. ~, \quad, \qquad, \hfill etc. 
	%(or a blank line to force the subfigure onto a new line)
	\begin{subfigure}[t]{0.49\textwidth}
		\includegraphics[width=\textwidth]{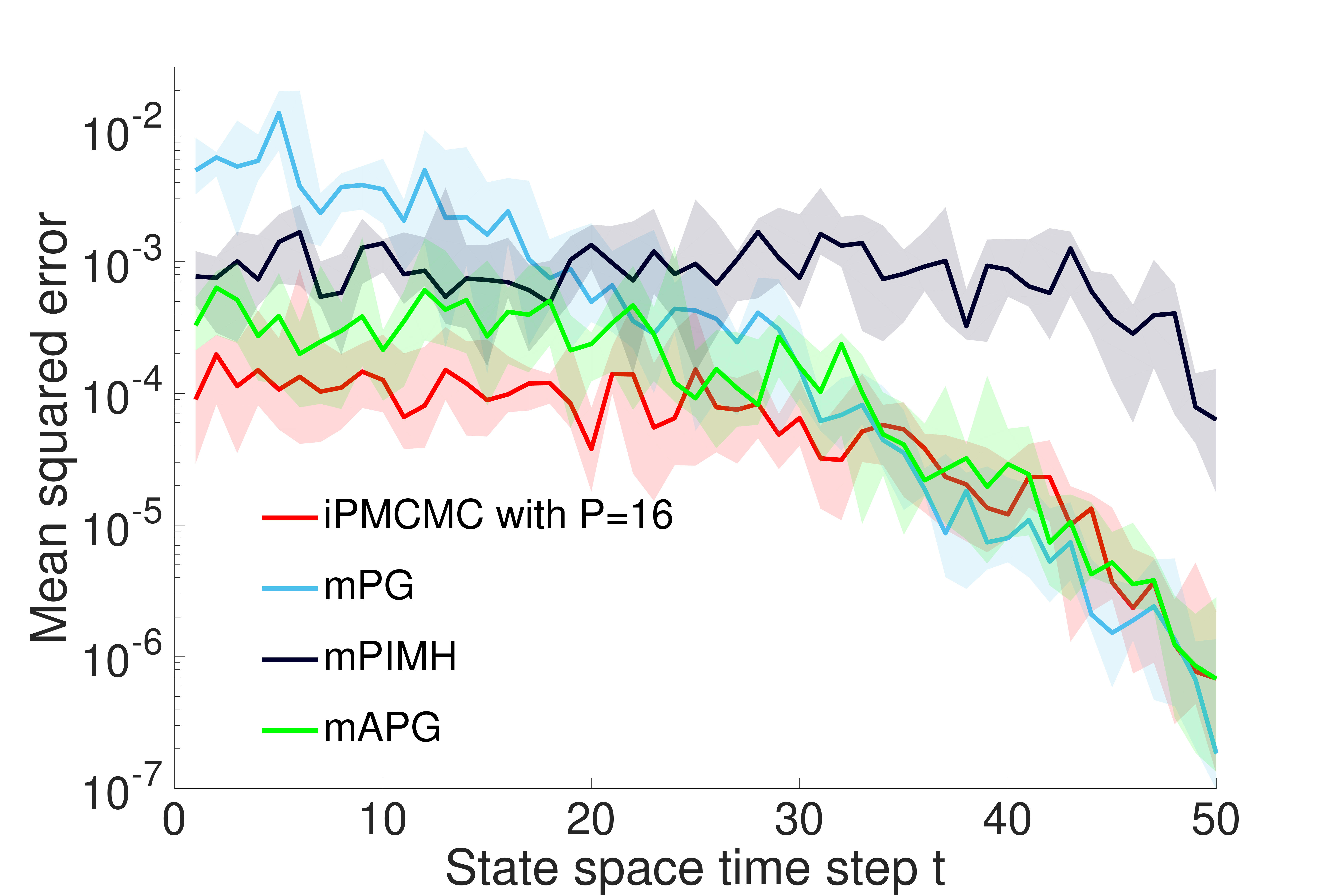}
		\caption{Final error in mean for latent marginals}
		\label{fig:meanPos}
	\end{subfigure}
	
	%	\begin{subfigure}[t]{0.49\textwidth}
	%		\includegraphics[width=\textwidth]{std_conv_lss}
	%		\caption{Convergence in standard deviation for full sequence}
	%		\label{fig:stdConv}
	%	\end{subfigure}
	%	~ %add desired spacing between images, e. g. ~, \quad, \qquad, \hfill etc. 
	%	%(or a blank line to force the subfigure onto a new line)
	%	\begin{subfigure}[t]{0.49\textwidth}
	%		\includegraphics[width=\textwidth]{std_pos_lss}
	%		\caption{Final error in standard deviation for latent marginals}
	%		\label{fig:stdPos}
	%	\end{subfigure}	
	\caption{Mean squared error averaged over all dimensions and steps in the state sequence as a function of MCMC iterations (left) and mean squared error after $10^4$ iterations averaged over dimensions as function of position in the state sequence (right) for \eqref{eq:LGSS} with 50 time sequences.  The solid line shows the median error across the 10 tested synthetic datasets, while the shading shows the upper and lower quartiles.  Ground truth was calculated using the Rauch--Tung--Striebel smoother algorithm \cite{rauch1965maximum}. 
		\label{fig:groundTruth}}
\end{figure*}

\begin{figure*}[t]
	\centering
	\begin{subfigure}[t]{0.49\textwidth}
		\includegraphics[width=\textwidth]{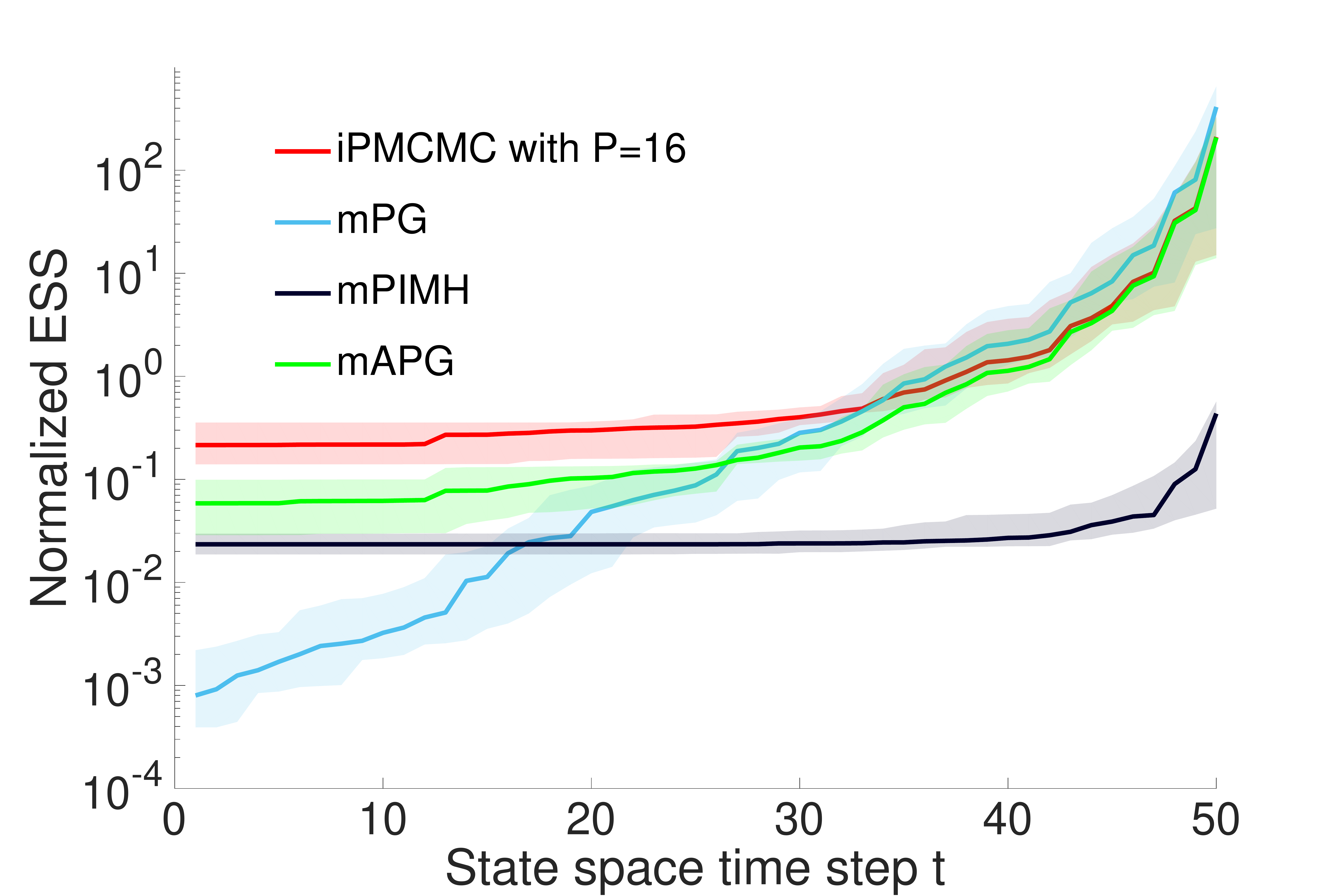}
		\caption{LGSSM}
	\end{subfigure}
	~ %add desired spacing between images, e. g. ~, \quad, \qquad, \hfill etc. 
	%(or a blank line to force the subfigure onto a new line)
	\begin{subfigure}[t]{0.49\textwidth}
		\includegraphics[width=\textwidth]{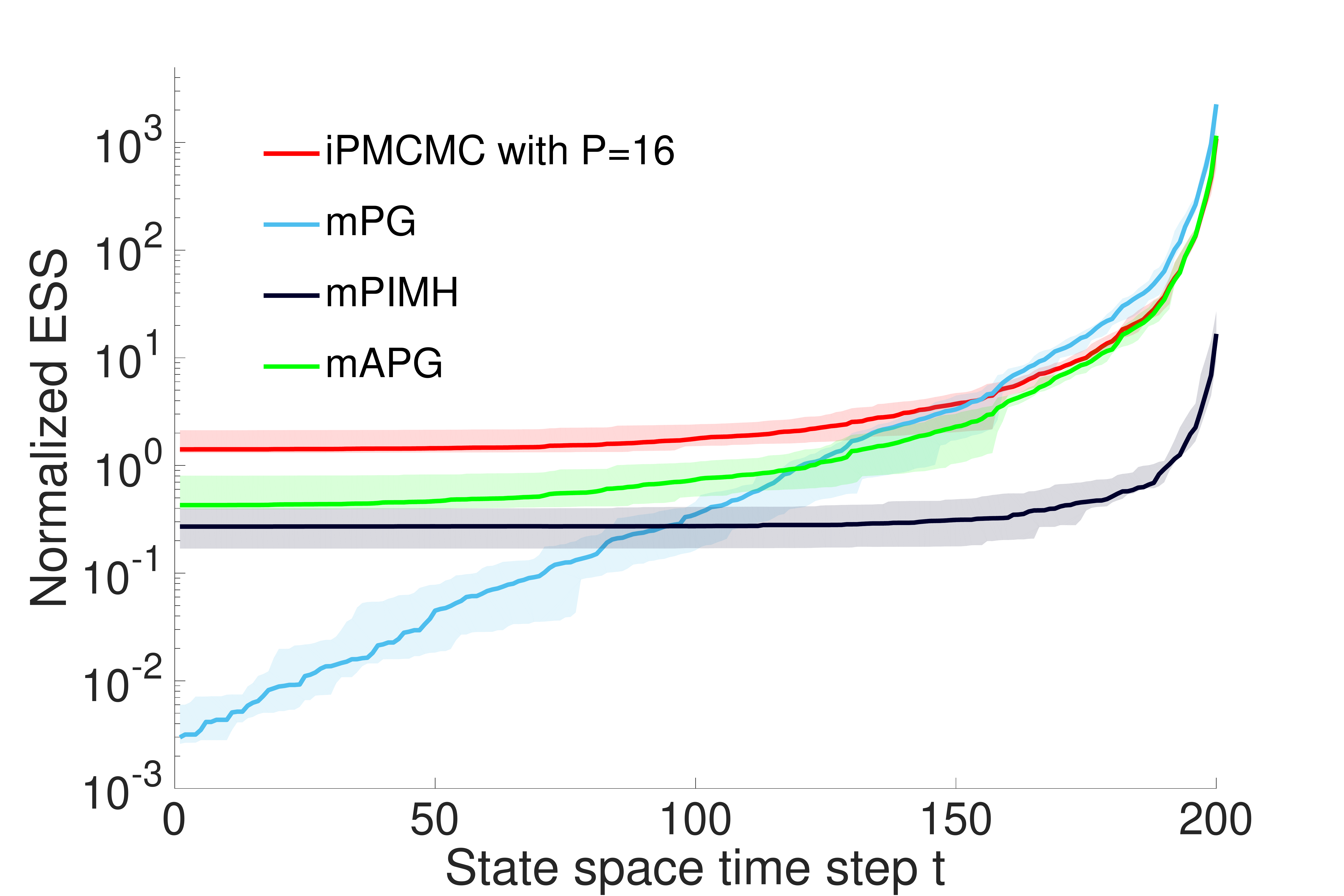}
		\caption{NLSSM}
	\end{subfigure}
	
	\caption{Normalized effective sample size  (NESS) for LGSSM (left) and NLSSM (right).
		\label{fig:ESS}}
\end{figure*}

\subsection{Nonlinear State Space Model}
\label{sec:nlss}

We next consider the one dimensional nonlinear state space model (NLSSM) considered by, among others, \citet{gordon1993novel,andrieuDH2010}
\begin{subequations}
	\label{eq:NLSS}
	\begin{align}
	x_1 & \sim \mathcal{N} \left(\mu, v^2\right) \label{eq:NLSSa}\\
	x_t & = \frac{x_{t-1}}{2} + 25 \frac{x_{t-1}}{1+x_{t-1}^2} + 8 \cos \left(1.2t\right) + \delta_{t-1} \label{eq:NLSSb} \\
	y_t & = \frac{{x_{t}}^2}{20} + \varepsilon_{t} \label{eq:NLSSc}
	\end{align}
\end{subequations}
where $\delta_{t-1} \sim \mathcal{N} \left(0, \omega^2\right)$ and $\varepsilon_{t} \sim \mathcal{N} \left(0, \sigma^2\right)$.  We set the parameters as $\mu = 0$, $v=\sqrt{5}$, $\omega = \sqrt{10}$ and $\sigma = \sqrt{10}$.  Unlike the LGSSM, this model does not have an analytic solution and therefore one must resort to approximate inference methods. 
% such as sampling.
Further, the multi-modal nature of the latent space makes full posterior inference over $x_{1:T}$ challenging for long state sequences. 

\begin{figure*}[t]
	\centering
	%\begin{subfigure}[t]{0.99\textwidth}
	\includegraphics[width=1\textwidth]{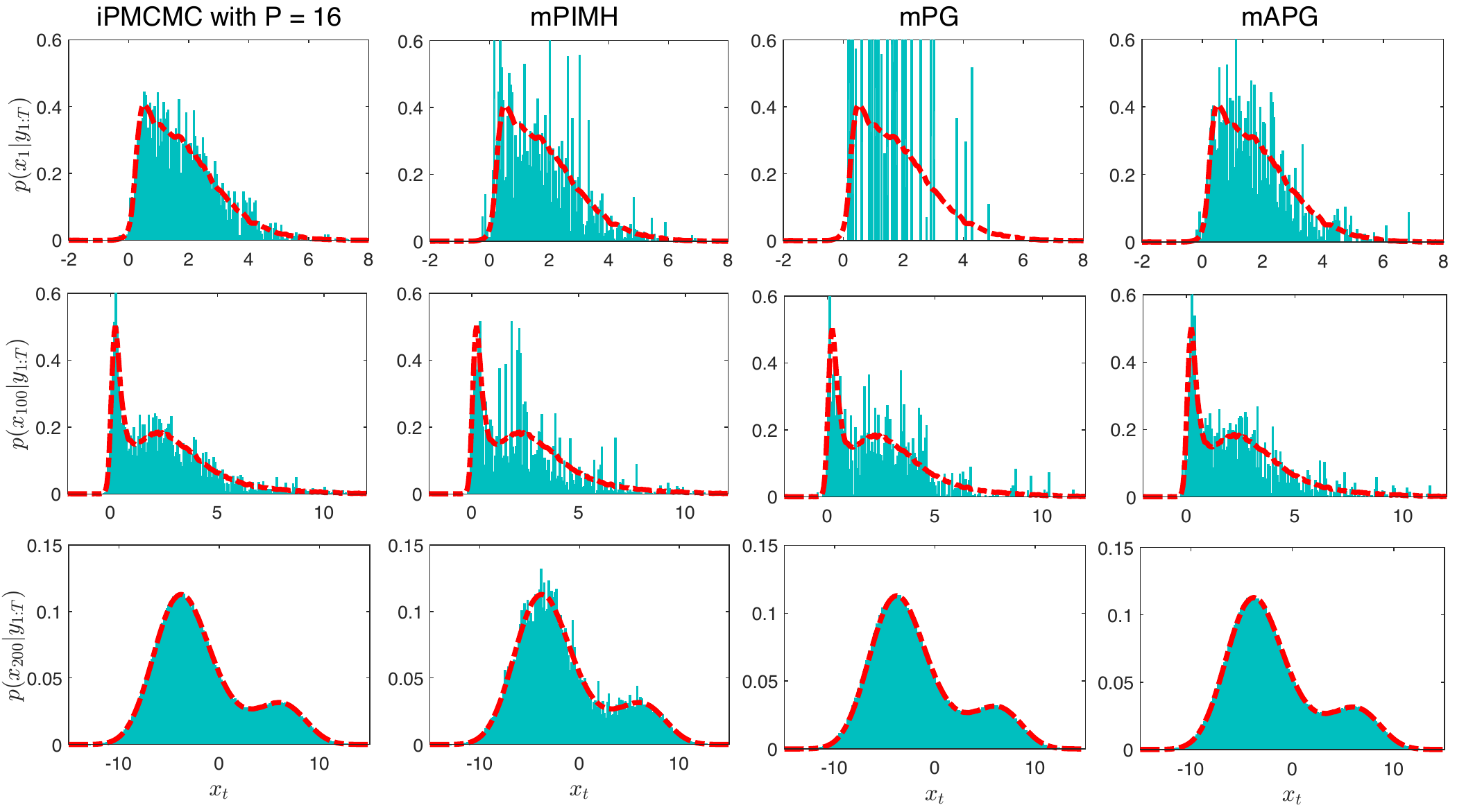}
	%\end{subfigure}
	\caption{Histograms of generated samples at $t=1, 100, \text{ and } 200$ for a single dataset generated from \eqref{eq:NLSS} with $T=200$.  Dashed red line shows an approximate estimate of the ground truth, found by running a kernel density estimator on the combined samples from a small number of independent SMC sweeps, each with $10^7$ particles. \label{fig:nlssHists}}
\end{figure*}

To examine the relative mixing of iPMCMC we calculate an effective sample size (ESS) for different steps in the state sequence.  In order to calculate the ESS, we condensed identical samples as done in for example \cite{vandemeent_aistats_2015}.  Let 
\begin{align*}
u_{t}^k \in \{x_{t,m}^{i}[r]\}^{i=1:N,r=1:R}_{m=1:M}, \quad \forall k \in 1 \dots K, \; t \in 1 \dots T
\end{align*} 
denote the unique samples of $x_t$ generated by all the nodes and sweeps of particular algorithm after $R$ iterations, where $K$ is the total number of unique samples generated.  The weight assigned to these unique samples, $v_t^{k}$, is given by the combined weights of all particles for which $x_t$ takes the value $u_{t}^k$:
\begin{align}
v_t^{k} = \sum_{r=1}^{R} \sum_{m=1}^{M} \sum_{i=1}^{N} \bar{w}_{t,m}^{i,r} \eta_{m}^{r} \delta_{x_{t,m}^{i}[r]}(u_{t}^{k})
\end{align}
where $\delta_{x_{t,m}^{i}[r]}(u_{t}^{k})$ is the Kronecker delta function and $\eta_{m}^{r}$ is a node weight.  For iPMCMC the node weight is given by as per the Rao-Blackwellized estimator described in Section~\ref{sec:allparticles}. For mPG and mPIMH, $\eta_{m}^{r}$ is simply $\frac{1}{RM}$,
as samples from the different nodes are weighted equally in the absence of interaction. 
Finally we define the effective sample size as $\text{ESS}_t = \left(\textstyle\sum_{k=1}^K \left(v_t^{k}\right)^2\right)^{-1}$.
%\begin{align}
%\label{eq:ESS}
%\text{ESS}_t = \left(\textstyle\sum_{k=1}^K \left(v_t^{k}\right)^2\right)^{-1}.
%\end{align}

Figure \ref{fig:ESS} shows the ESS for the LGSSM and NLSSM as a function of position in the state sequence.  For this, we omit the samples generated by the initialization step as this SMC sweep is common to all the tested algorithms.  We further normalize by the number of MCMC iterations so as to give an idea of the rate at which unique samples are generated.  These show that for both models the ESS of iPMCMC, mPG and mAPG is similar towards the end of the space sequence, but that iPMCMC outperforms all the other methods at the early stages. The ESS of mPG was particularly poor at early iterations.  PIMH performed poorly throughout, reflecting the very low observed acceptance ratio of around $7.3\%$ on average. 
%The lack of Monte Carlo convergence rate appearing in Figure \ref{fig:meanConv} also suggests this acceptance ratio is yet to converge, with the value in the ergodic regime likely to be even lower.   

It should be noted that the ESS is not a direct measure of performance for these models.  For example, the equal weighting of nodes is likely to make the ESS artificially high for mPG, mPIMH and mAPG, when compared with methods such as iPMCMC that assign a weighting to the nodes at each iteration.  To acknowledge this, we also plot histograms for the marginal distributions of a number of different position in the state sequence as shown in Figure \ref{fig:nlssHists}.  These confirm that iPMCMC and mPG have similar performance at the latter state sequence steps, whilst iPMCMC is superior at the earlier stages, with mPG producing almost no more new samples than those from the initialization sweep due to the degeneracy.  The performance of PIMH was consistently worse than iPMCMC throughout the state sequence, with even the final step exhibiting noticeable noise.

%An important feature of the ESS is that it is equal to the total number of samples ($NRM$) if all the samples are unique and have the same weight, and is equal to $1$ if there is only single unique sample with non-zero weight. It should be noted that ESS is not a direct measure of performance in the context of distributed PMCMC methods. In particular it takes no account of the suitability of the weights assigned to samples, for example the equal weighting assigned to different chains for mPG and mPIMH mean they will always have a higher ESS for a particle sweep \brooks{clarify ``particle sweep''?} than a iPMCMC sweep generating the same samples, regardless of whether this equal weighting gives a better approximation to the true posterior. \fredrik{I did not quite understand the last sentence.}

%% file: discussion.tex
% !TEX root = main.tex

% Discussion for things such as choice of P
\section{Discussion and Future Work}
\label{sec:discussion}

The \ipmcmc sampler overcomes degeneracy issues in \pg by allowing the newly sampled particles from \smc nodes to replace the retained particles in \csmc nodes. Our experimental results demonstrate that, for the models considered, this switching in rate is far higher than the rate at which PG generates fully independent samples. Moreover, the results in Figure~\ref{fig:Psweep} suggest that the degree of improvement over an m\pg sampler with the same total number of nodes increases with the total number of nodes in the pool. 

The mAPG sampler performs an accept reject step that compares the marginal likelihood estimate of a single \csmc sweep to that of a single \smc sweep. In the \ipmcmc sampler the \csmc estimate of the marginal likelihood is compared to a population sample of \smc estimates, resulting in a higher probability that at least one of the \smc nodes will become a \csmc node.

Since the original \pmcmc paper in 2010 there have been several papers studying \citep{chopin2015,lindsten2015uniform} %\fredrik{Should we give some references to theoretical work on PG? for instance, I know of a paper in SJOS which is very good ;)}
and improving upon the basic \pg algorithm. Key contributions to combat the path degeneracy effect are backward simulation \citep{whiteley2010efficient,lindsten2013backward}
and ancestor sampling \citep{lindstenJS2014}. These can also be used to improve the \ipmcmc method ever further.  

%Noting that the Gibbs update in \eqref{eq:simConditional} requires no interaction between the \csmc nodes, iPMCMC should be amenable to an asynchronous adaptation under the assumption of a random execution time, independent of $\xb'$, in Algorithm~\ref{alg:csmc}.

%For this one would run a number of \smc nodes independently, communicating only their normalisation constant estimate and a sampled trajectory to a pool of possible retained particles with associated weights.  Whenever a \csmc sweep finishes its execution, it samples a new retained particle from either the pool or itself, replenishing the pool with a new sample if the new retained particle was not from the original \csmc sweep.  It should be noted that as this pool could be increased indefinitely, this adaptation has the potential to increase the scope and potential performance of iPMCMC even further.

% MOVED BACK TO METHODS
% Use backward sampling and/or ancestor sampling
%\subsection{Backward Simulation and Ancestor Sampling}
%Adding a backward or ancestor simulation step can drastically increase mixing when sampling the conditional trajectories $\xb_j'[r]$ \citep{lindsten2013backward}. In the \ipmcmc sampler we can replace simulating from the final weights on line~7 by a backward simulation step. Another option for the \csmc nodes is to replace this step by internal ancestor sampling \citep{lindstenJS2014} steps and simulate from the final weights as normal.

%% file: appendix.tex
% !TEX root = main.tex

%\section{Notation}
%We will consider a target $\pi_T(x_{1:T}) = Z_{\pi_T}^{-1}\gamma_T(x_{1:T})$ that can be split up in to a sequence of unnormalised distributions $\gamma_t(x_{1:t}), ~t=1,\ldots,T$. For this we run an \smc algorithm that generates samples $\xb_m^i = (x_m^{i,1},\ldots,x_m^{i,T}) \in \setX^T$ and ancestry paths $\ab_m^i = (a_m^{i,1},\ldots,a_m^{i,T-1})\in \{1,\ldots,N\}^{\otimes T-1} \eqdef [N]^{\otimes T-1}$. Running the \smc algorithm induces a distribution denoted by $q_{\text{SMC}}\left( \xb_m^i, \ab_m^i \right)$ over all generated random variables for $m=1,\ldots,M$ and $i=1,\ldots,N$. For simplicity we only consider the basic sequential Monte Carlo algorithm where we resample multinomially at each step, propagate from $q_t(x_t|x_{1:t-1})$ and weight according to
%\begin{align}
%W_t(x_{1:t}) = \frac{\gamma_t(x_{1:t})}{\gamma_{t-1}(x_{1:t-1}) q_t(x_t|x_{1:t-1})}.\nonumber
%\end{align}
%Note that this formulation is slightly more general than the one in the main manuscript. For a more thorough description of the \smc sampler we refer the reader to \eg \citet{doucet2001sequential,doucet2009tutorial}.

% ======================================================================
%
%         Distributed Multiple Conditional Sequential Monte Carlo          
%  
% ======================================================================

\section{Proof of Theorem~\ref{thm:one}}
\label{sec:proof}
The proof follows similar ideas as \citet{andrieuDH2010}. We prove that the interacting particle Markov chain Monte Carlo sampler is in fact a standard partially collapsed Gibbs sampler \citep{van2008partially} on an extended space ${\Upsilon \eqdef \setX^{\otimes MTN} \times [N]^{\otimes M(T-1)N} \times [M]^{\otimes P} \times [N]^{\otimes P}}$.
\vspace{10pt}
%\[
%{\Upsilon \eqdef \setX^{\otimes MTN} \times [N]^{\otimes M(T-1)N} \times [M]^{\otimes P} \times [N]^{\otimes P}}.
%\]

\begin{proof}
	Assume the setup of Section~\ref{sec:method}. %For convenience, we repeat the target distribution $\tilde\pi(\cdot)$, \ie \eqref{eq:targetdistribution}, on $\Upsilon$ here 
	With $\tilde\pi(\cdot)$ with as per \eqref{eq:targetdistribution}, we will show that the Gibbs sampler on the extended space, $\Upsilon$, defined as follows	
	\begin{subequations}
		\label{eq:gibbs}
		\begin{align}
		\xib_{1:M} \backslash\{\xb_{c_{1:P}}^{b_{1:P}}, \bb_{c_{1:P}} \} ~&\sim \tilde \pi(~\cdot~|\xb_{c_{1:P}}^{b_{1:P}}, \bb_{c_{1:P}},c_{1:P}, b_{1:P})\label{eq:particles},\\
		c_j &\sim \tilde \pi(~\cdot~|\xib_{1:M}, c_{1:P\backslash j}), ~~j=1,\ldots,P,\label{eq:worker}\\
		b_j &\sim \tilde \pi(~\cdot~|\xib_{1:M}, c_{1:P}), ~~j=1,\ldots,P,\label{eq:index}
		\end{align}
	\end{subequations}
	is equivalent to the \ipmcmc method in Algorithm~\ref{alg:ipmc}.
	%where $\tilde \pi( \cdot )$ is given by \eqref{eq:supptargetdist}.
	
	First, the initial step \eqref{eq:particles} corresponds to sampling from
	\begin{align*}
	\tilde\pi(\xib_{1:M} &\backslash\{\xb_{c_{1:P}}^{b_{1:P}}, \bb_{c_{1:P}} \} | \xb_{c_{1:P}}^{b_{1:P}}, \bb_{c_{1:P}},c_{1:P}, b_{1:P}) = \\ 
	& \prod_{\substack{m=1\\m\notin c_{1:P}}}^M q_{\text{SMC}}\left(\xib_m\right) \prod_{j = 1}^P q_{\text{CSMC}}\left(\xib_{c_j} \backslash \{\xb_{c_j}^{b_j}, \bb_{c_j}\} \mid \xb_{c_j}^{b_j}, \bb_{c_j}, c_j, b_j \right).
	\end{align*}
	This, excluding the conditional trajectories, just corresponds to steps 3--4 in Algorithm~\ref{alg:ipmc}, \ie running $P$ \csmc and $M-P$ \smc algorithms independently.
	
	We continue with a reformulation of \eqref{eq:targetdistribution} which will be useful to prove correctness for the other two steps
	\begin{align}
	\label{eq:reformtargetdist}
	\tilde \pi & (\xib_{1:M},  c_{1:P}, b_{1:P}) \nonumber\\ &=\frac{1}{\binom{M}{P}} \prod_{m=1}^M q_{\text{SMC}}\left(\xib_m\right) \cdot
	\prod_{j = 1}^P \left[
	\iden_{c_j \notin c_{1:j-1}} \nw_{T,c_j}^{b_j}
	\pi_T\left(\xb_{c_j}^{b_j}\right) \frac{q_{\text{CSMC}}\left(\xib_{c_j} \backslash \{\xb_{c_j}^{b_j}, \bb_{c_j}\} \mid \xb_{c_j}^{b_j}, \bb_{c_j}, c_j, b_j \right)}{N^{T} 
		\nw_{T,c_j}^{b_j}
		%\frac{w_{T,c_j}^{b_j}}{\sum_{i=1}^N w_{T,c_j}^i} 
		q_{\text{SMC}}\left(\xib_{c_j}\right)}\right] \nonumber\\
	&= \frac{1}{\binom{M}{P}} \prod_{m=1}^M q_{\text{SMC}}\left(\xib_m\right) \cdot \prod_{j = 1}^P \frac{\hat Z_{c_j}}{Z}\iden_{c_j \notin c_{1:j-1}} 
	\nw_{T,c_j}^{b_j}
	%\frac{w_{T,c_j}^{b_j}}{\sum_{i=1}^N w_{T,c_j}^i}
	.
	\end{align}
	
	Furthermore, we note that by marginalising (collapsing) the above reformulation, \ie \eqref{eq:reformtargetdist}, over $b_{1:P}$ we get
	\begin{align*}
	\tilde\pi(\xib_{1:M}, c_{1:P}) 
	%&= \sum_{b_{1:P}}\tilde\pi(\xib_{1:M}, c_{1:P}, b_{1:P}) \nonumber \\
	&= \frac{1}{\binom{M}{P}} \prod_{m=1}^M q_{\text{SMC}}\left(\xib_m\right) \prod_{j = 1}^P \frac{\hat Z_{c_j}}{Z}\iden_{c_j \notin c_{1:j-1}} .\nonumber
	\end{align*}
	From this it is easy to see that $\tilde\pi(c_j | \xib_{1:M}, c_{1:P\backslash j}) = \hat\nz_{c_j}^j$, which 
	%\begin{align*}
	%\tilde\pi(c_j | \xib_{1:M}, c_{1:P\backslash j}) = \hat\nz_{\pi_T,c_j}^j
	%\tilde\pi(c_j | \xib_{1:M}, c_{1:P\backslash j}) = \frac{\hat Z_{\pi_T,c_j} \iden_{c_j \notin c_{1:P\backslash j}}}{\sum_{m=1}^M \hat Z_{\pi_T,m} \iden_{m \notin c_{1:P\backslash j}}}
	%\end{align*}
	corresponds to sampling the conditional node indices, \ie step 6 in Algorithm~\ref{alg:ipmc}. Finally, from \eqref{eq:reformtargetdist} we can see that simulating $b_{1:P}$ can be done independently as follows
	\begin{align*}
	&\tilde\pi(b_{1:P} | \xib_{1:M}, c_{1:P}) = \frac{\tilde\pi(b_{1:P} ,\xib_{1:M}, c_{1:P})}{\tilde\pi(\xib_{1:M}, c_{1:P})} =  \prod_{j = 1}^P 
	\nw_{T,c_j}^{b_j}
	%\frac{w_{T,c_j}^{b_j}}{\sum_{i=1}^N w_{T,c_j}^i}
	.
	\end{align*}
	This corresponds to step 7 in the \ipmcmc sampler, Algorithm~\ref{alg:ipmc}. So the procedure defined by \eqref{eq:gibbs} is a partially collapsed Gibbs sampler, derived from \eqref{eq:targetdistribution}, and we have shown that it is exactly equal to the \ipmcmc sampler described in Algorithm~\ref{alg:ipmc}.
\end{proof}

%% file: supp.tex
\section{Using All Particles}
\label{sec:suppUseAll}
The Monte Carlo estimator is given by
\begin{align}
\label{eq:mcestimate-supp}
\E[f(\xb)] & \approx \frac{1}{RP}\sum_{r=1}^R\sum_{j=1}^P f(\xb_j'[r]) \nonumber \\ &= \frac{1}{R}\sum_{r=1}^R \frac{1}{P}\sum_{j=1}^P f(\xb_j'[r]),
\end{align}
where we can note that $\xb_j'[r] = \xb_{c_j}^{b_j}$ from the internal particle system at iteration $r$. We can however make use of all particles to estimate expectations of interest by, for each \mcmc iteration $r$, averaging over the sampled conditional node indices $c_{1:P}$ and corresponding particle indices $b_{1:P}$. This procedure is referred to as a Rao-Blackwellization of a statistical estimator and is (in terms of variance) never worse than the original one, and often much better. For iteration $r$ we need to calculate the following
\begin{align*}
\frac{1}{P}\sum_{j=1}^P f(\xb_j'[r]) = \frac{1}{P}\sum_{j=1}^P f(\xb_{c_j}^{b_j}),
\end{align*}
where we can Rao-Blackwellize the selection of the retained particle along with each individual Gibbs update as following
\begin{align*}
\frac{1}{P}\sum_{j=1}^P \E_{c_j, b_j | \xi_{1:M}, c_{1:P \backslash j}}\left[f(\xb_{c_j}^{b_j}) \right] =& \frac{1}{P}\sum_{j=1}^P \E_{c_j | \xi_{1:M}, c_{1:P \backslash j}}\left[\sum_{i=1}^N \bar w_{T,c_j}^i  f(\xb_{c_j}^{i}) \right] \\ =& \frac{1}{P}\sum_{j=1}^P \sum_{i=1}^N \E_{c_j | \xi_{1:M}, c_{1:P \backslash j}}\left[\bar w_{T,c_j}^i  f(\xb_{c_j}^{i}) \right] \\
=& \frac{1}{P}\sum_{j=1}^P \sum_{i=1}^N \sum_{m=1}^M \hat \zeta_m^j \bar w_{T,m}^i  f(\xb_{m}^{i}) \\=& \frac{1}{P}\sum_{j=1}^P \sum_{m=1}^M \hat \zeta_m^j \sum_{i=1}^N \bar w_{T,m}^i  f(\xb_{m}^{i})
\\=& \frac{1}{P} \sum_{m=1}^M  \left[\left(\sum_{j=1}^P \hat \zeta_m^j \right) \cdot \left(\sum_{i=1}^N \bar w_{T,m}^i  f(\xb_{m}^{i}) \right)\right]
%&\E_{c_j|c_{1:P\backslash j}}\left[\E_{b_{1:P}}\left[\frac{1}{P}\sum_{j=1}^P f(\xb_{c_j}^{b_j})\right]\right] = \frac{1}{P}\sum_{j=1}^P \E_{c_j|c_{1:P\backslash j}}\left[\E_{b_{1:P}}\left[ f(\xb_{c_j}^{b_j})\right]\right] = \frac{1}{P}\sum_{j=1}^P \E_{c_j|c_{1:P\backslash j}}\left[ \sum_{i=1}^N \bar w_{T,c_j}^i  f(\xb_{c_j}^{i})\right] \\
%&= \frac{1}{P}\sum_{j=1}^P \sum_{i=1}^N \E_{c_j|c_{1:P\backslash j}}\left[ \bar w_{T,c_j}^i  f(\xb_{c_j}^{i})\right] = \frac{1}{P}\sum_{j=1}^P \sum_{i=1}^N \sum_{m=1}^M \hat \zeta_m^j \bar w_{T,m}^i  f(\xb_{m}^{i}) = \frac{1}{P}\sum_{j=1}^P \sum_{m=1}^M \hat \zeta_m^j \sum_{i=1}^N \bar w_{T,m}^i  f(\xb_{m}^{i}),
\end{align*}
where we have made use of the knowledge that the internal particle system $\{(\xb_{m}^{i},\bar w_{T,m}^i)\}$ does not change between Gibbs updates of the $c_j$'s, whereas the $\hat \zeta_m^j$ do.  We emphasise that this is a separate Rao-Blackwellization of each Gibbs update of the conditional node indices, such that each is conditioned upon the actual update made at $j-1$, rather than a simultaneous Rao-Blackwellization of the full batch of $P$ updates.  Though the latter also has analytic form and should theoretically be lower variance, it suffers from inherent numerical instability and so is difficult to calculate in practise.  We found that empirically there was not a noticeable difference between the performance of the two procedures.  Furthermore, one can always run additional Gibbs updates on the $c_j$'s and obtain an improve estimate on the relative sample weightings if desired.

\section{Choosing $P$}
\label{sec:supp-choosep}
For the purposes of this study we assume, without loss of generality, that the indices for the conditional nodes are always $c_{1:P} = \{1,\ldots,P\}$. Then we can show that the probability of the event that at least one conditional nodes switches with an unconditional is given by
\begin{align}
\label{eq:switchingprob-supp}
\Prb(\{\text{switch}\}) = 1 - \E\left[\prod_{j=1}^P \frac{\hat Z_{j}}{\hat Z_{j} +\sum_{m=P+1}^M \hat Z_{m}}\right].
\end{align}
%Since the normalisation constant estimates $\hat Z_{\pi_T}^i$'s are random variables, this probability is of course also a random.

Now, there are some asymptotic (and experimental) results \citep{pitt2012some,berard2014lognormal,doucet2015efficient} that indicate that a decent approximation for the distribution of the log of the normalisation constant estimates is Gaussian. This would mean the distributions of the conditional and unconditional normalisation constant estimates with variance $\sigma^2$ can be well-approximated as follows
\begin{align}
\log \left(\frac{\hat Z_{j}}{Z}\right) &\sim \N(\frac{\sigma^2}{2},\sigma^2),\quad j = 1,\ldots,P,\\
\log \left(\frac{\hat Z_{m}}{Z}\right) &\sim \N(-\frac{\sigma^2}{2},\sigma^2),\quad m = P+1,\ldots,M.
\end{align}
A straight-forward Monte Carlo estimation of the switching probability, \ie $\Prb(\{\text{switch}\})$, can be seen in Figure~\ref{fig:supp-expectedswitch} for various settings of $\sigma$ and $M$. These results seem to indicate that letting $P \approx M/2$ maximises the probability of switching.

\begin{figure}
	\centering
	\begin{subfigure}[b]{0.48\textwidth}
		\includegraphics[width=\textwidth,trim={0.3cm 0 0.9cm 0},clip]{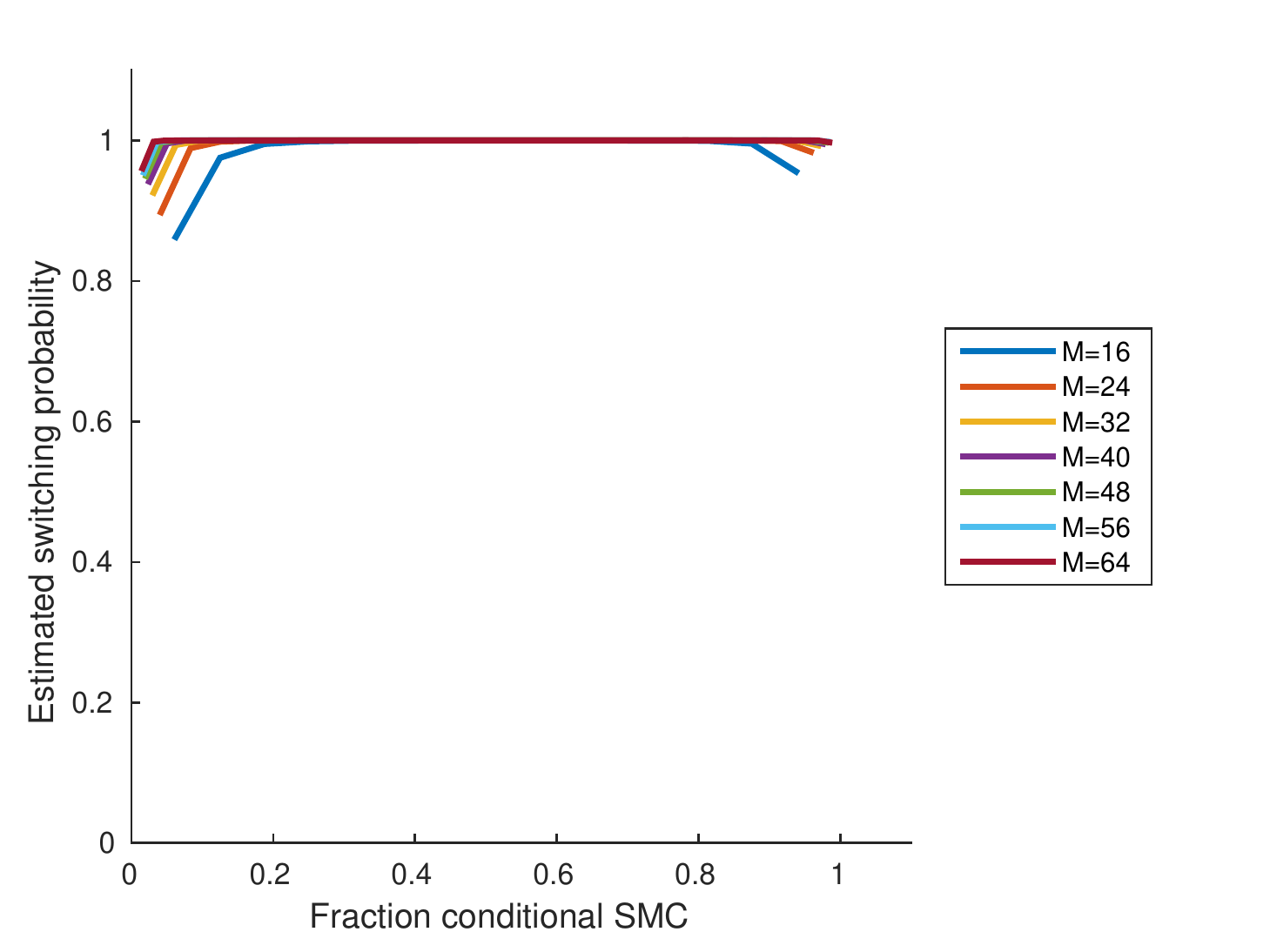}
		\caption{$\sigma=1$}
	\end{subfigure}
	~~ %add desired spacing between images, e. g. ~, \quad, \qquad, \hfill etc. 
	%(or a blank line to force the subfigure onto a new line)
	\begin{subfigure}[b]{0.48\textwidth}
		\includegraphics[width=\textwidth,trim={0.3cm 0 0.9cm 0},clip]{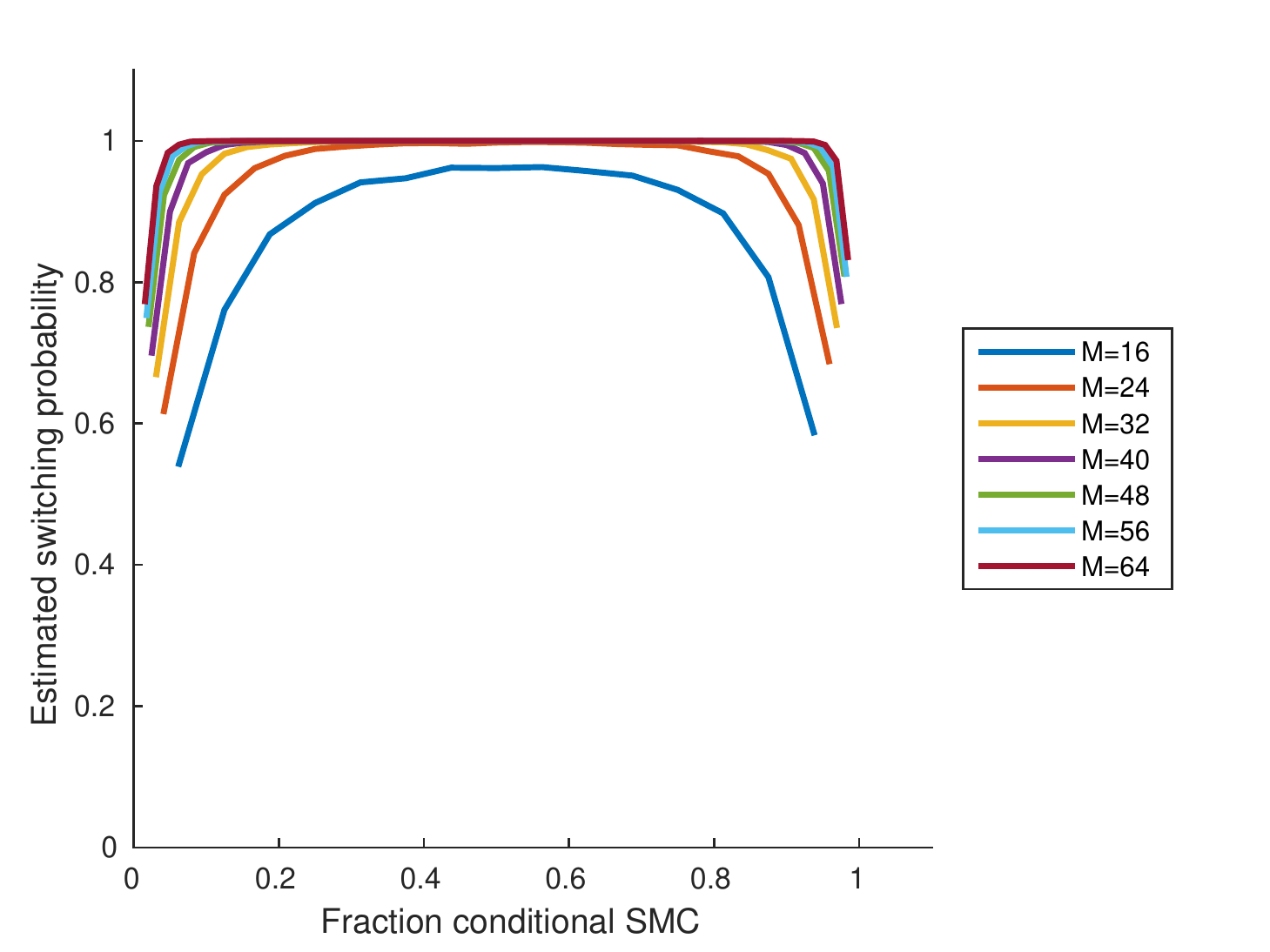}
		\caption{$\sigma=2$}
	\end{subfigure}
	
	\begin{subfigure}[b]{0.48\textwidth}
		\includegraphics[width=\textwidth,trim={0.3cm 0 0.9cm 0},clip]{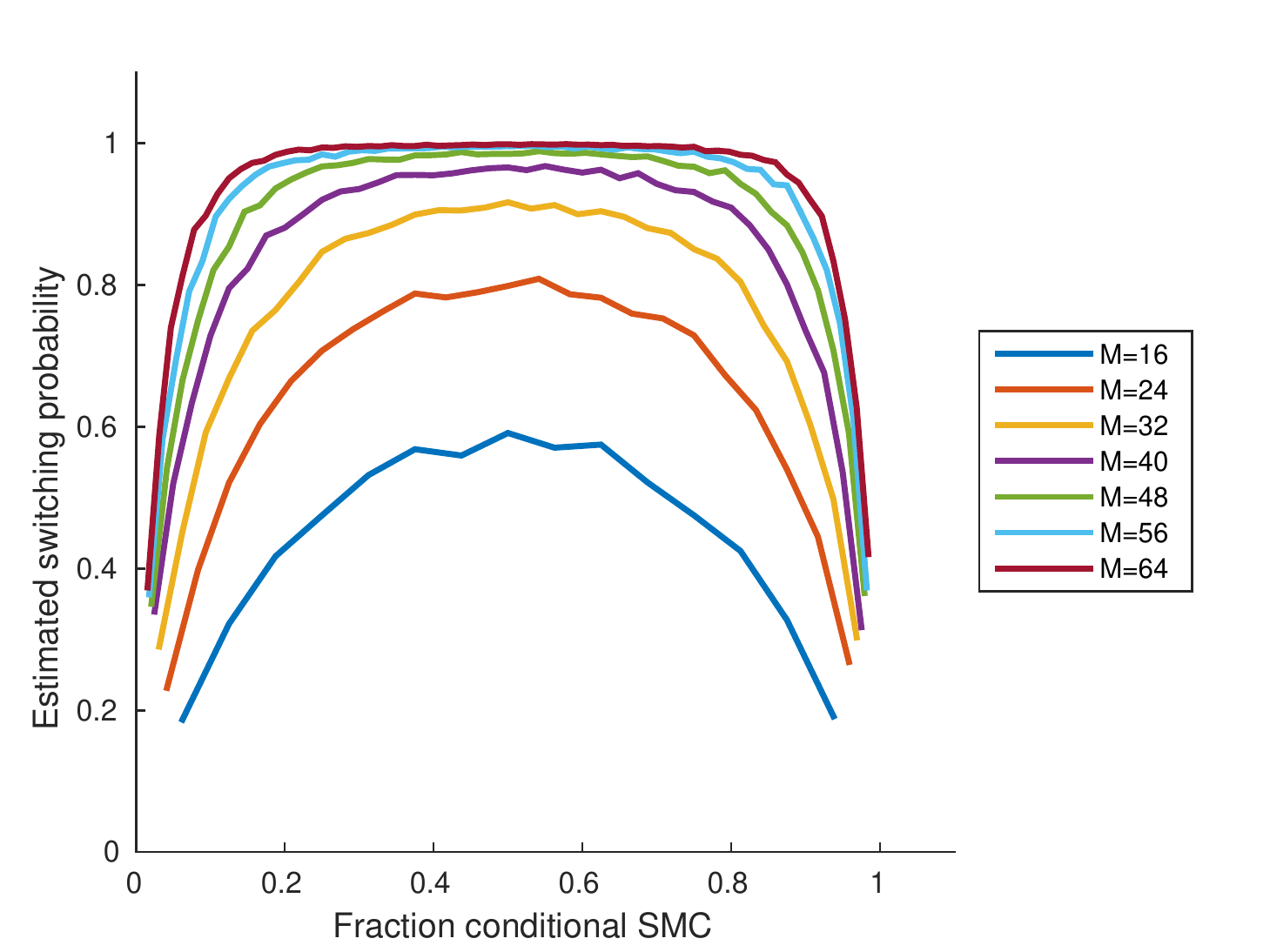}
		\caption{$\sigma=3$}
	\end{subfigure}
	~~ %add desired spacing between images, e. g. ~, \quad, \qquad, \hfill etc. 
	%(or a blank line to force the subfigure onto a new line)
	\begin{subfigure}[b]{0.48\textwidth}
		\includegraphics[width=\textwidth,trim={0.3cm 0 0.9cm 0},clip]{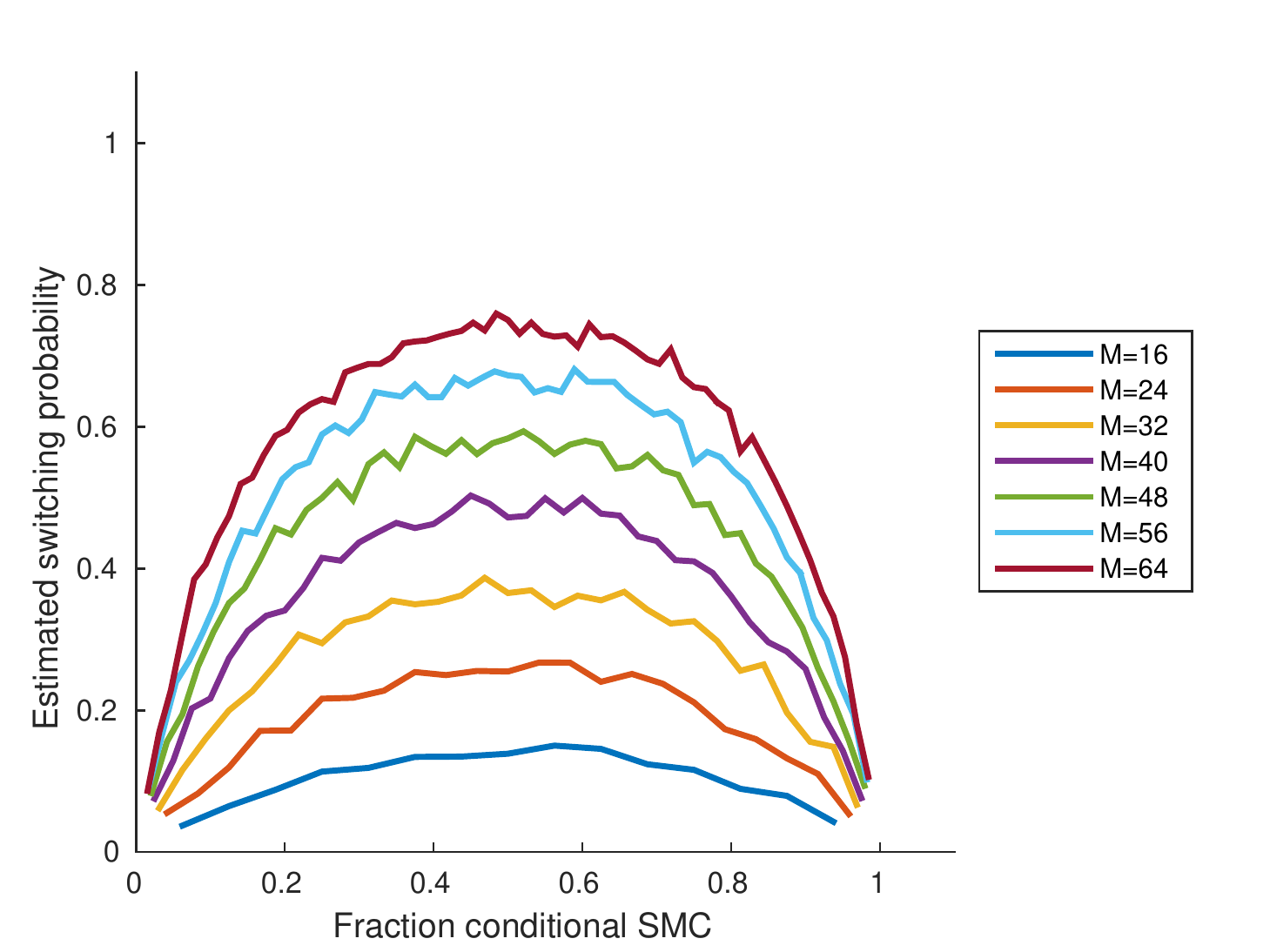}
		\caption{$\sigma=4$}
	\end{subfigure}
	
	\begin{subfigure}[b]{0.48\textwidth}
		\includegraphics[width=\textwidth,trim={0.3cm 0 0.9cm 0},clip]{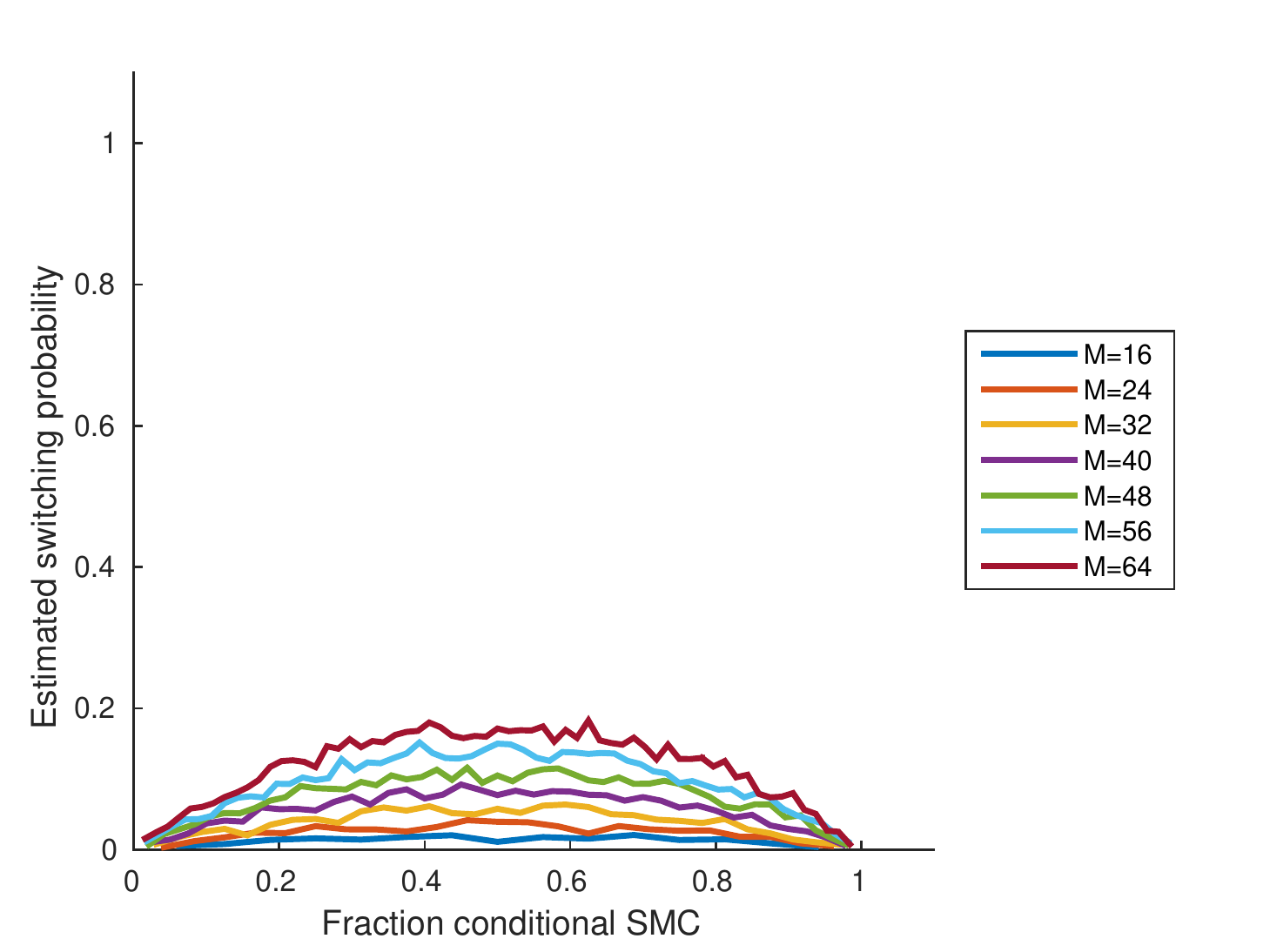}
		\caption{$\sigma=5$}
	\end{subfigure}
	~~ %add desired spacing between images, e. g. ~, \quad, \qquad, \hfill etc. 
	%(or a blank line to force the subfigure onto a new line)
	\begin{subfigure}[b]{0.48\textwidth}
		\includegraphics[width=\textwidth,trim={0.3cm 0 0.9cm 0},clip]{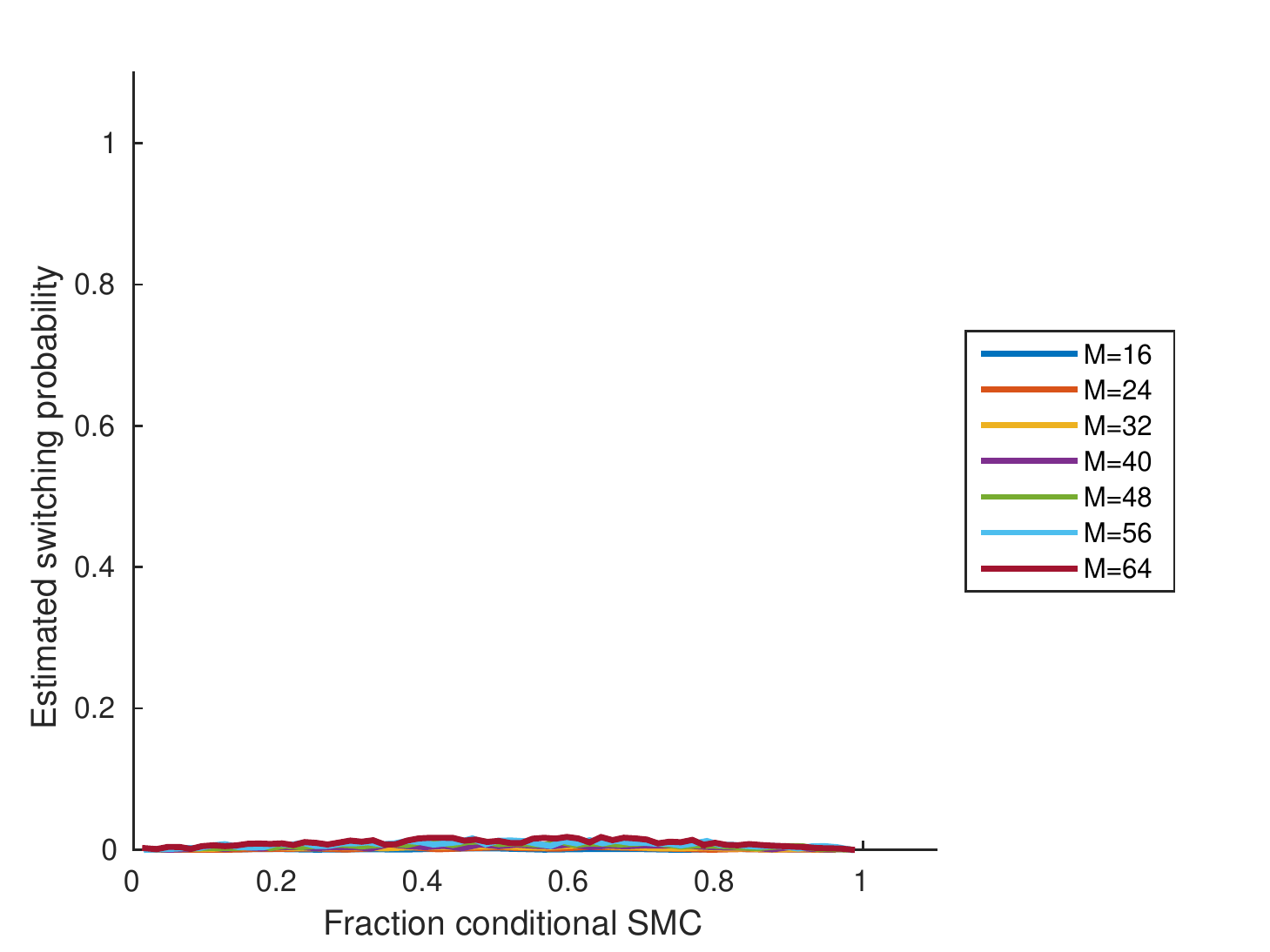}
		\caption{$\sigma=6$}
	\end{subfigure}
	\caption{Estimation of switching probability for various settings of $\sigma$ and $M$.}\label{fig:supp-expectedswitch}
\end{figure}

\newpage
\section{Additional Results Figures}
\label{sec:supp-additionalFigures}

%In this section we provide figures to support those in the main paper.
\vspace{-10pt}
\begin{figure*}[h!]
	\centering
	\begin{subfigure}[t]{0.48\textwidth}
		\includegraphics[width=\textwidth,trim={4cm 0 6cm 0},clip]{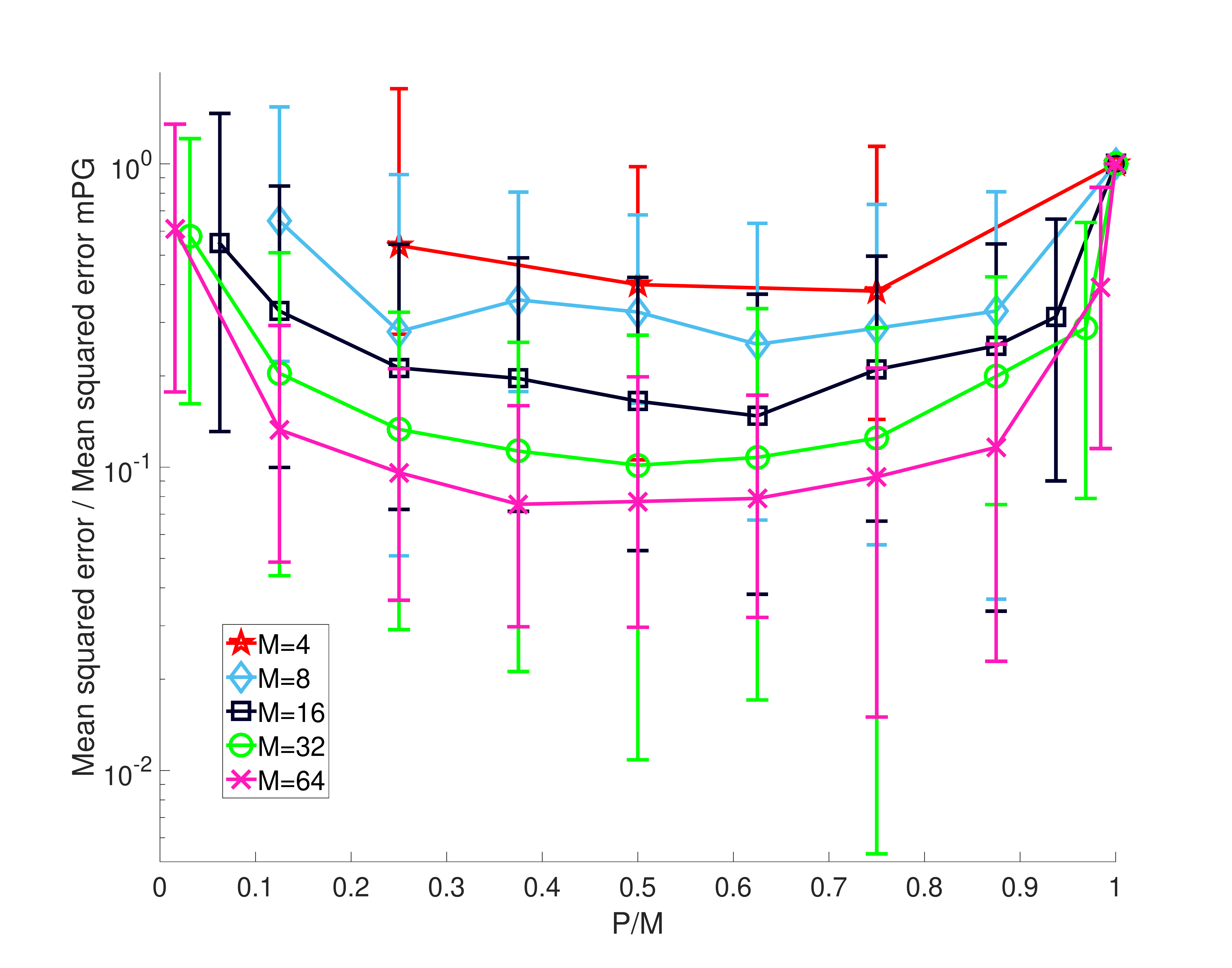}
		\caption{Mean}
		\label{fig:supp-meanConv}
	\end{subfigure}
	~~~ %add desired spacing between images, e. g. ~, \quad, \qquad, \hfill etc. 
	%(or a blank line to force the subfigure onto a new line)
	\begin{subfigure}[t]{0.48\textwidth}
		\includegraphics[width=\textwidth,trim={4cm 0 6cm 0},clip]{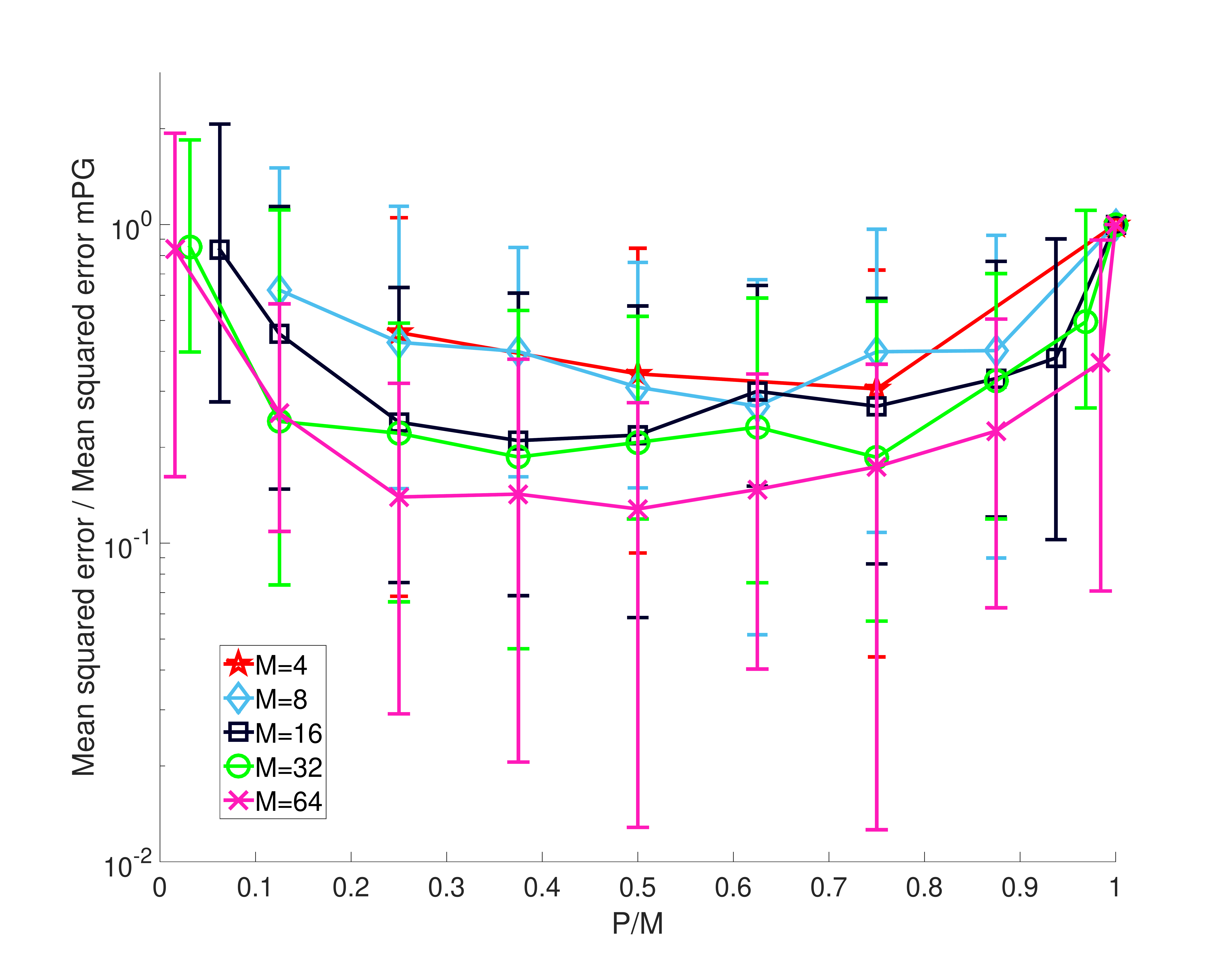}
		\caption{Standard Deviation}
		\label{fig:supp-meanPos}
	\end{subfigure}
	
	\begin{subfigure}[t]{0.48\textwidth}
		\includegraphics[width=\textwidth,trim={4cm 0 6cm 0},clip]{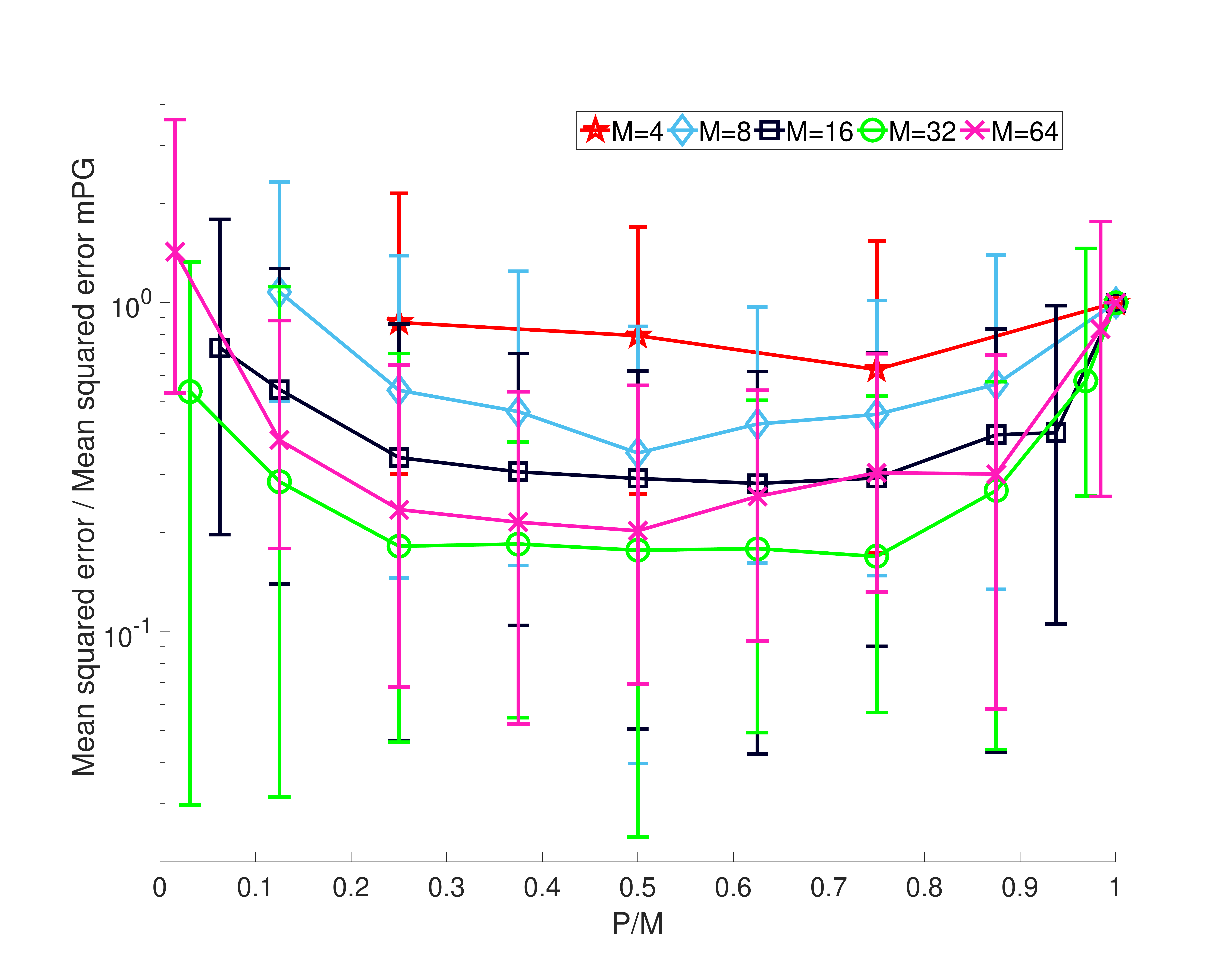}
		\caption{Skewness}
		\label{fig:supp-stdConv}
	\end{subfigure}
	~~~ %add desired spacing between images, e. g. ~, \quad, \qquad, \hfill etc. 
	%(or a blank line to force the subfigure onto a new line)
	\begin{subfigure}[t]{0.48\textwidth}
		\includegraphics[width=\textwidth,trim={4cm 0 6cm 0},clip]{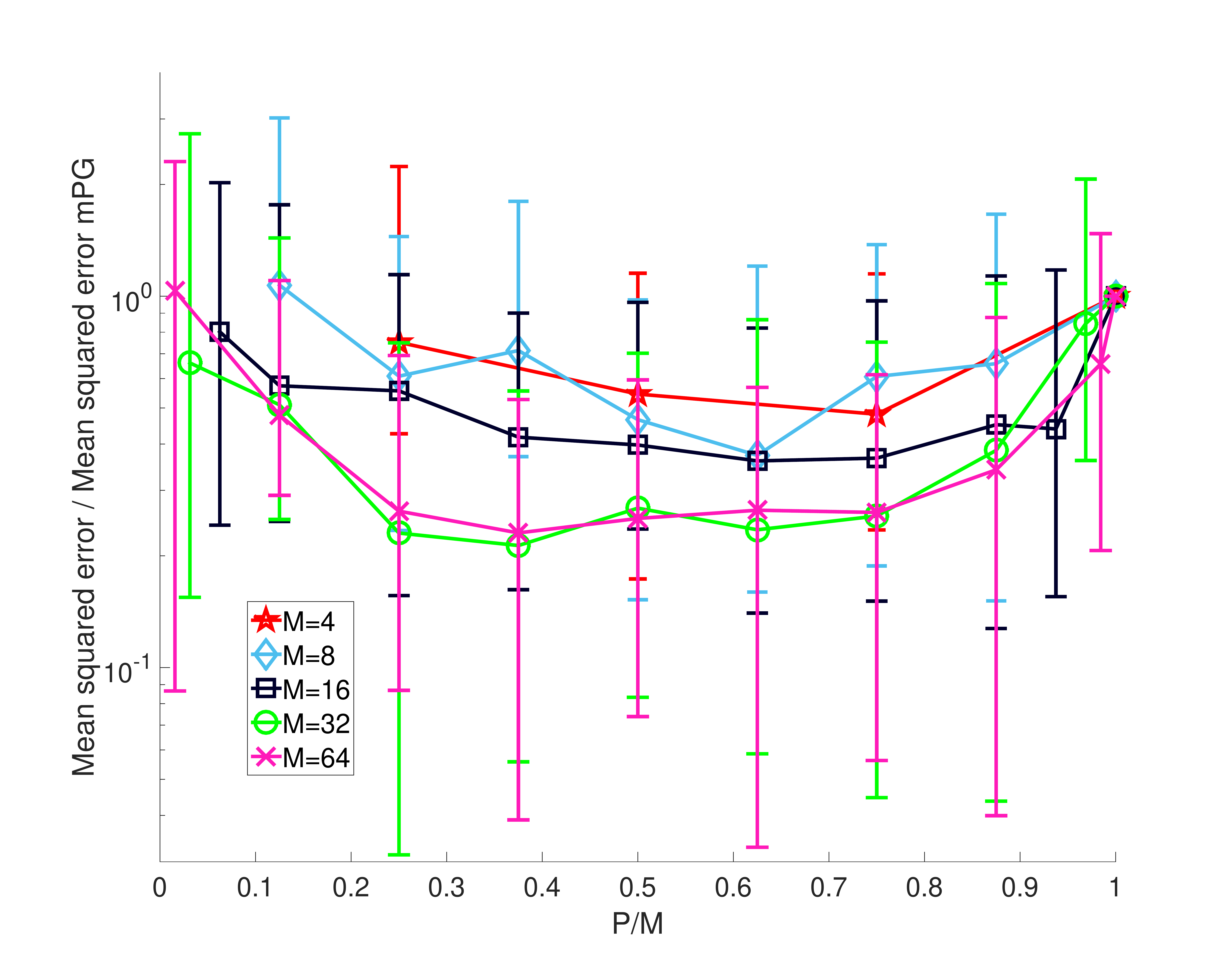}
		\caption{Kurtosis}
		\label{fig:supp-stdPos}
	\end{subfigure}	
	\caption{Median error in marginal moment estimates with different choices of P and M over 10 different synthetic datasets of the linear Gaussian state space model given in (10) after 1000 MCMC iterations. Errors are normalized by the error of a multi-start PG sampler which is a special case of iPMCMC for which $P=M$ (see Section 4).  Error bars show the lower and upper quartiles for the errors.  It can be seen that for all the moments then $P/M\approx1/2$ give the best performance. For the mean and standard deviation estimates, the accuracy relative to the non-interacting distribution case $P=M$ shows a clear increase with $M$. This effect is also seen for the skewness and excess kurtosis estimates except for the distinction between the $M=32$ and $M=64$ cases.  This may be because these metric are the same for the prior and the posterior such that good results for these metric might be achievable even when the samples give a poor match to the true posterior. \label{fig:supp-pSweepErrorBar}}
\end{figure*}

\begin{figure*}[p]
	\centering
	\begin{subfigure}[t]{0.49\textwidth}
		\makebox[\textwidth][r]{\includegraphics[width=1.1\textwidth,trim={4cm 0 5cm 0},clip]{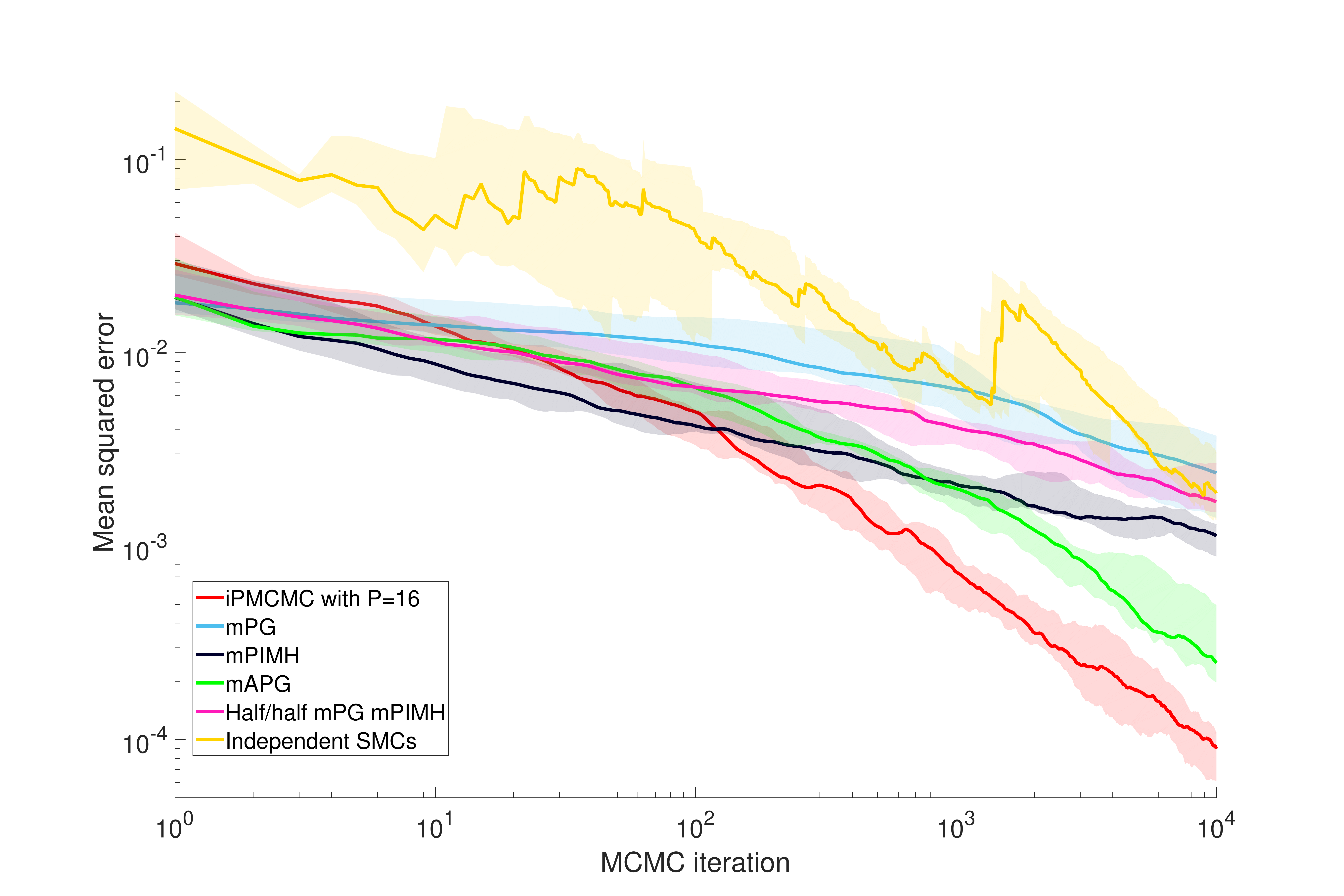}}
		\caption{Convergence in mean}
		\label{fig:supp-mean_dist_conv}
	\end{subfigure}
	~ %add desired spacing between images, e. g. ~, \quad, \qquad, \hfill etc. 
	%(or a blank line to force the subfigure onto a new line)
	\begin{subfigure}[t]{0.49\textwidth}
		\makebox[\textwidth][l]{\includegraphics[width=1.1\textwidth,trim={4cm 0 5cm 0},clip]{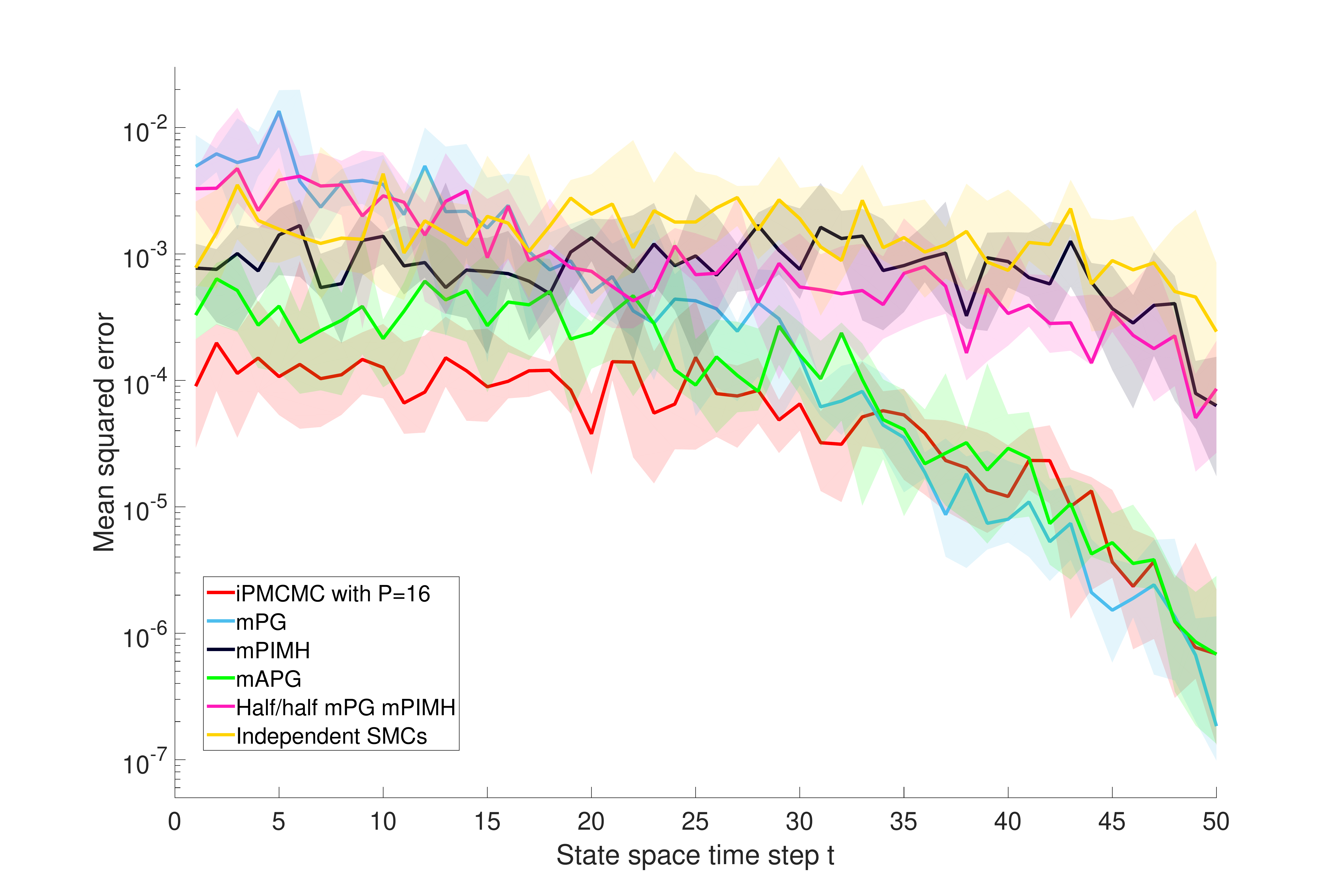}}
		\caption{Final error in mean}
		\label{fig:supp-mean_dist_pos}
	\end{subfigure}
	
	\begin{subfigure}[t]{0.49\textwidth}
		\makebox[\textwidth][r]{\includegraphics[width=1.1\textwidth,trim={4cm 0 5cm 0},clip]{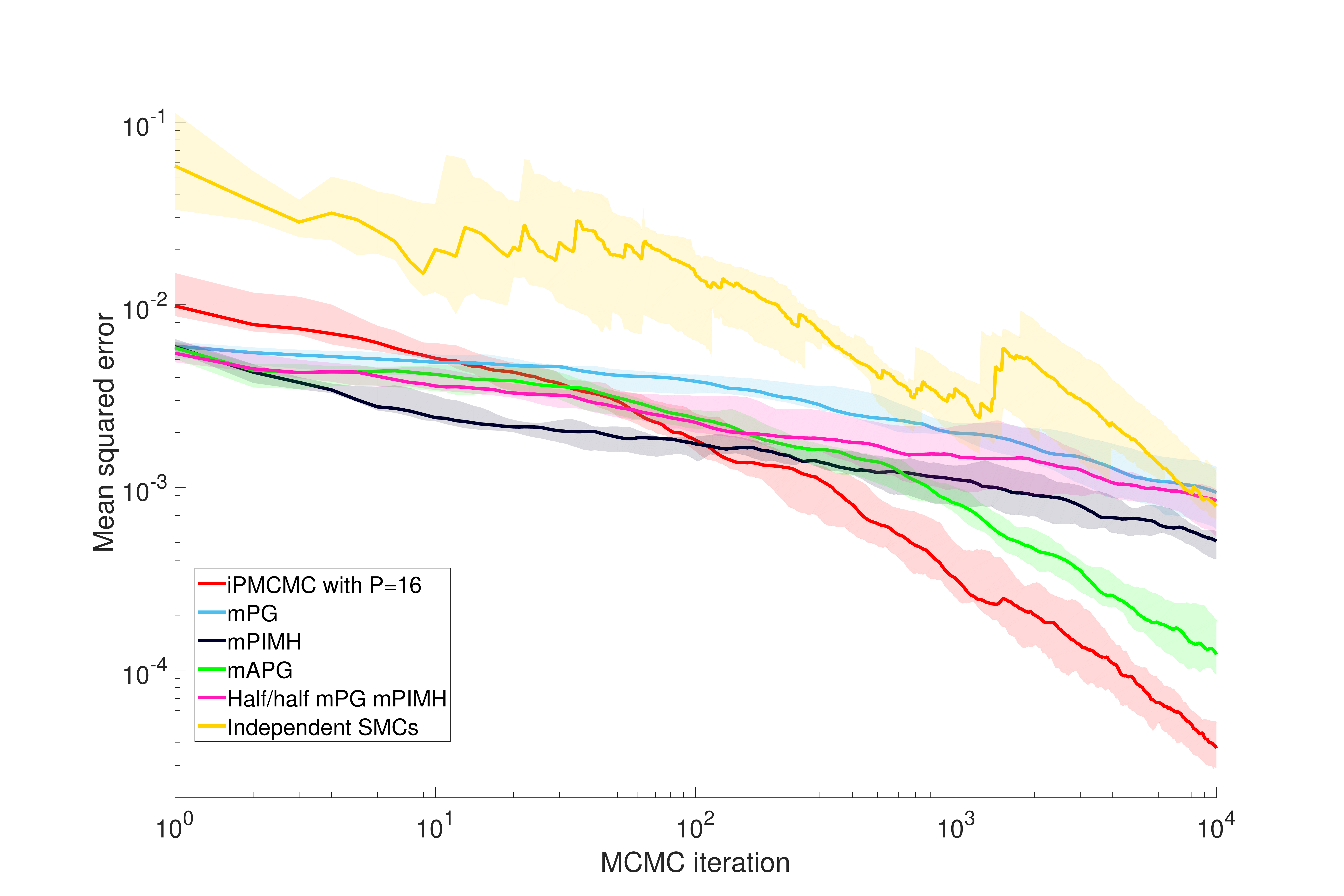}}
		\caption{Convergence in standard deviation}
		\label{fig:supp-std_dist_conv}
	\end{subfigure}
	~ %add desired spacing between images, e. g. ~, \quad, \qquad, \hfill etc. 
	%(or a blank line to force the subfigure onto a new line)
	\begin{subfigure}[t]{0.49\textwidth}
		\makebox[\textwidth][l]{\includegraphics[width=1.1\textwidth,trim={4cm 0 5cm 0},clip]{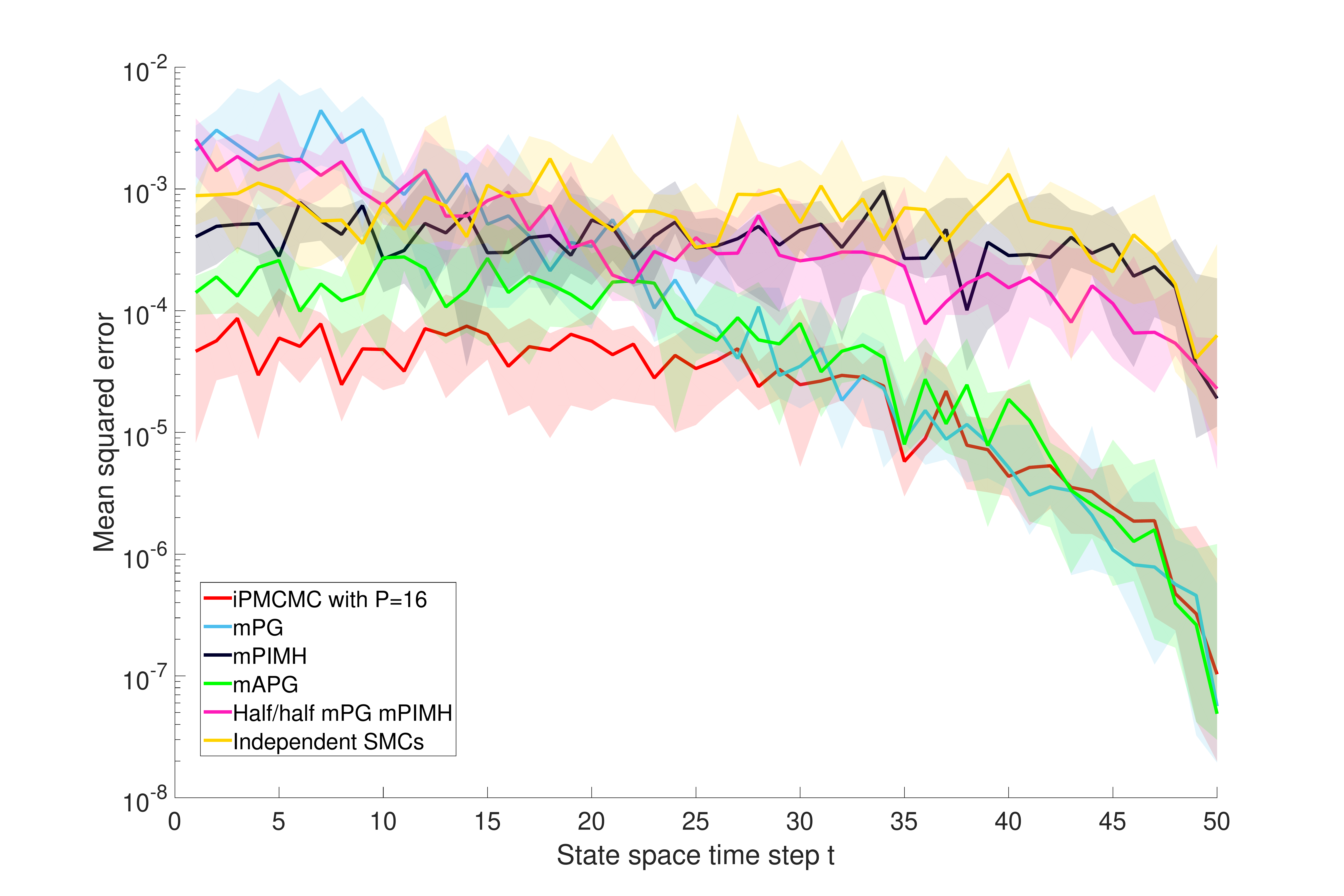}}
		\caption{Final error in standard deviation}
		\label{fig:supp-std_dist_pos}
	\end{subfigure}	
	\caption{Mean squared error in latent variable mean and standard deviation averaged over all dimensions of the LGSSM as a function of MCMC iteration (left) and position in the state sequence (right) for a selection of paraellelizable SMC and PMCMC methods.  See figure 3 in main paper for more details. \label{fig:supp-distributable_error_lss}}
\end{figure*}

\begin{figure*}[p]
	\centering
\begin{subfigure}[t]{0.49\textwidth}
	\makebox[\textwidth][r]{\includegraphics[width=1.1\textwidth,trim={4cm 0 5cm 0},clip]{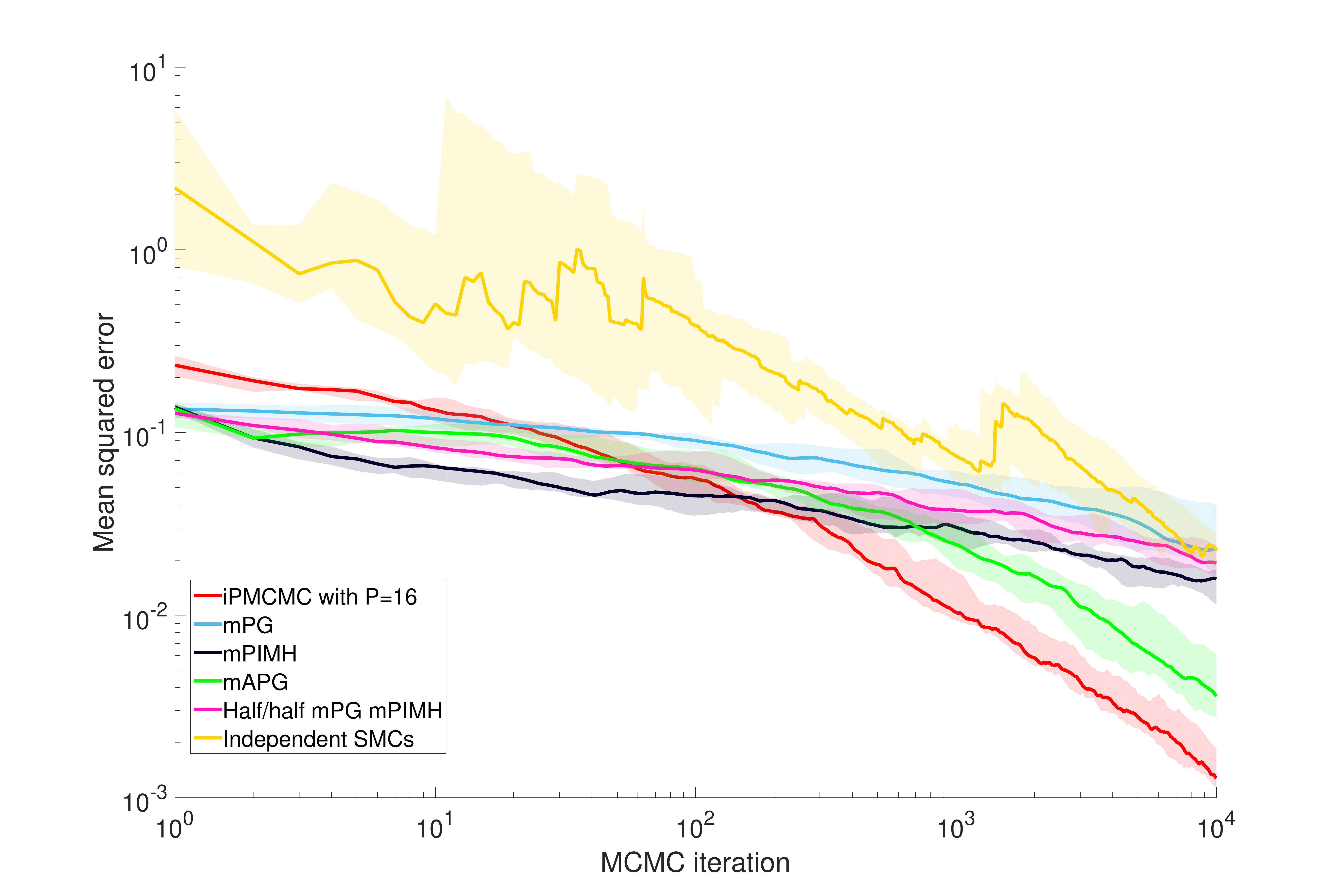}}
	\caption{Convergence in skewness}
	\label{fig:supp-ske_dist_conv}
\end{subfigure}
~ %add desired spacing between images, e. g. ~, \quad, \qquad, \hfill etc. 
%(or a blank line to force the subfigure onto a new line)
\begin{subfigure}[t]{0.49\textwidth}
	\makebox[\textwidth][l]{\includegraphics[width=1.1\textwidth,trim={4cm 0 5cm 0},clip]{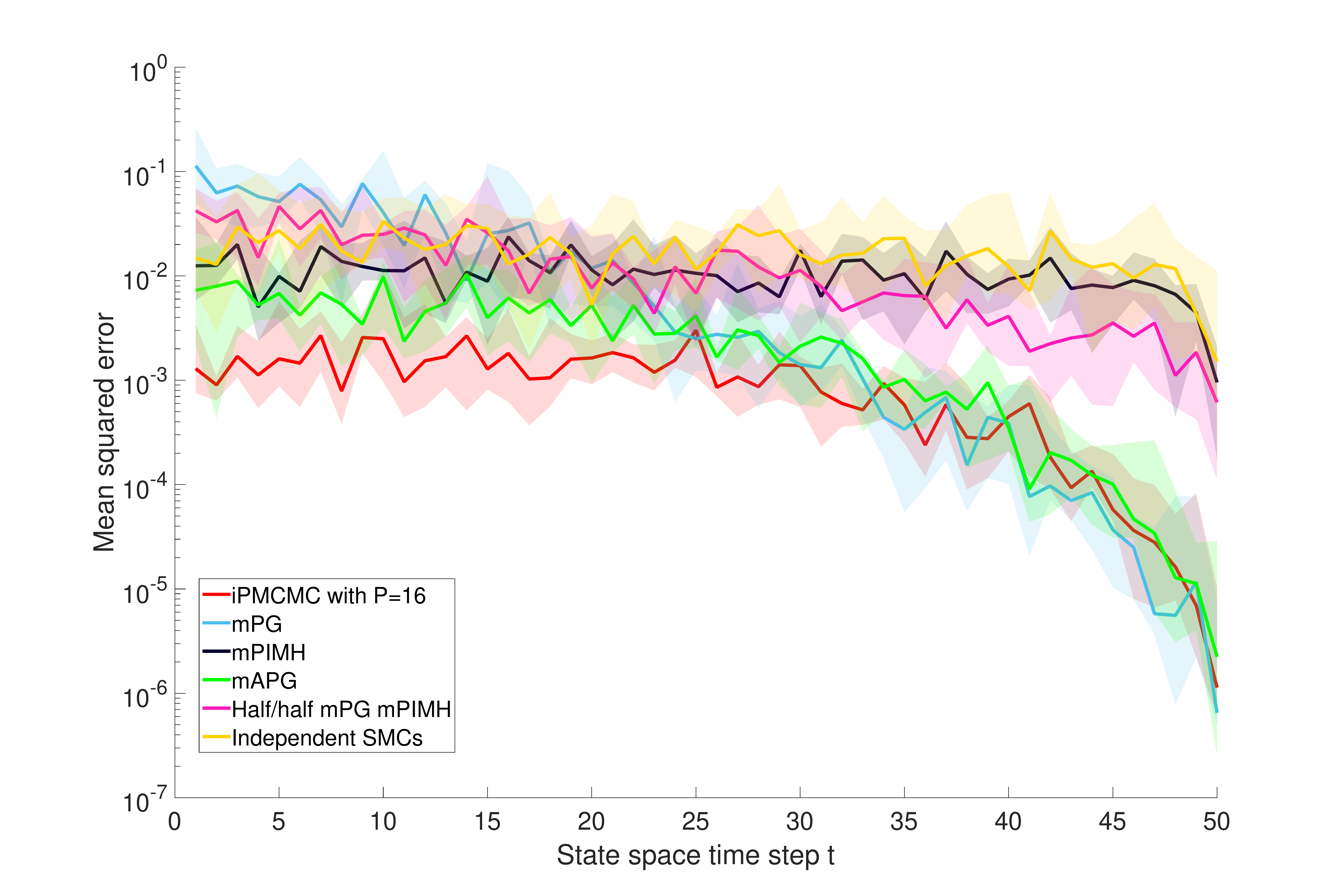}}
	\caption{Final error in skewness}
	\label{fig:supp-ske_dist_pos}
\end{subfigure}	

\begin{subfigure}[t]{0.49\textwidth}
	\makebox[\textwidth][r]{\includegraphics[width=1.1\textwidth,trim={4cm 0 5cm 0},clip]{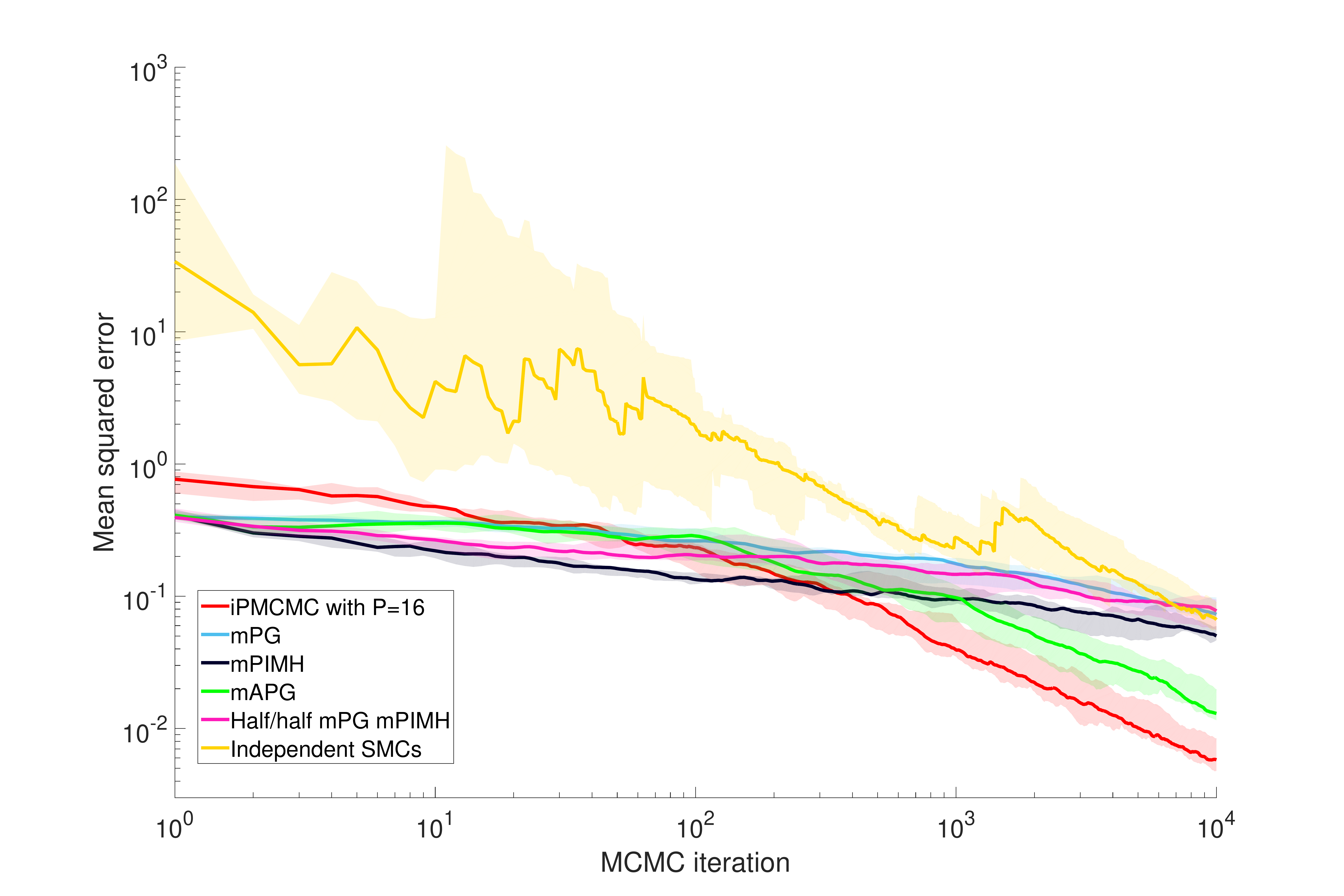}}
	\caption{Convergence in kurtosis}
	\label{fig:supp-kurt_dist_conv}
\end{subfigure}
~ %add desired spacing between images, e. g. ~, \quad, \qquad, \hfill etc. 
%(or a blank line to force the subfigure onto a new line)
\begin{subfigure}[t]{0.49\textwidth}
	\makebox[\textwidth][l]{\includegraphics[width=1.1\textwidth,trim={4cm 0 5cm 0},clip]{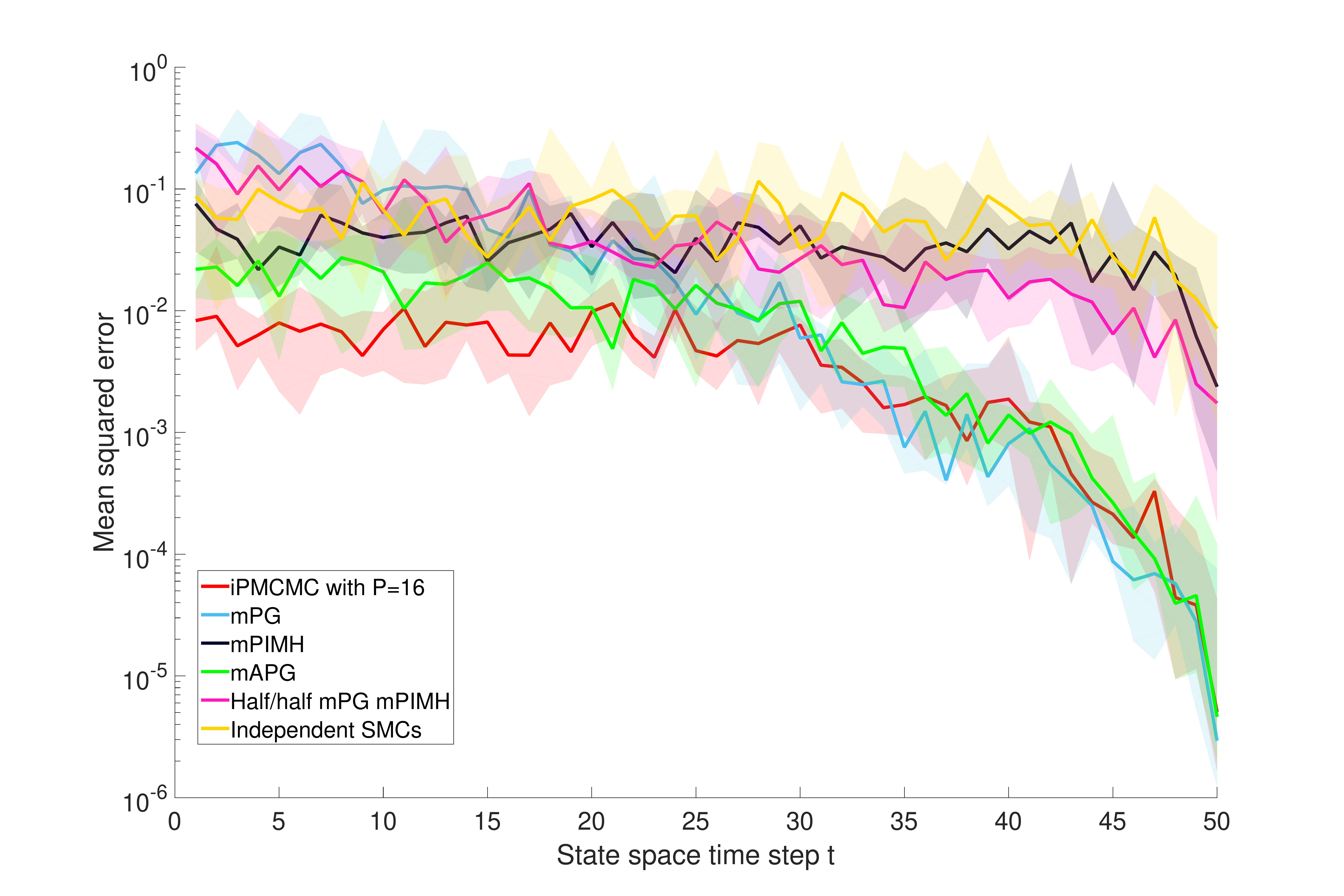}}
	\caption{Final error in kurtosis}
	\label{fig:supp-kurt_dist_pos}
\end{subfigure}	
	\caption{Mean squared error in latent variable skewness and kurtosis averaged over all dimensions of the LGSSM as a function of MCMC iteration (left) and position in the state sequence (right) for a selection of paraellelizable SMC and PMCMC methods.  See figure 3 in main paper for more details. \label{fig:supp-distributable_error_lss2}}
\end{figure*}

\begin{figure*}[p]
	\centering
	\begin{subfigure}[t]{0.49\textwidth}
		\makebox[\textwidth][r]{\includegraphics[width=1.1\textwidth,trim={2cm 0.7cm 2cm 0},clip]{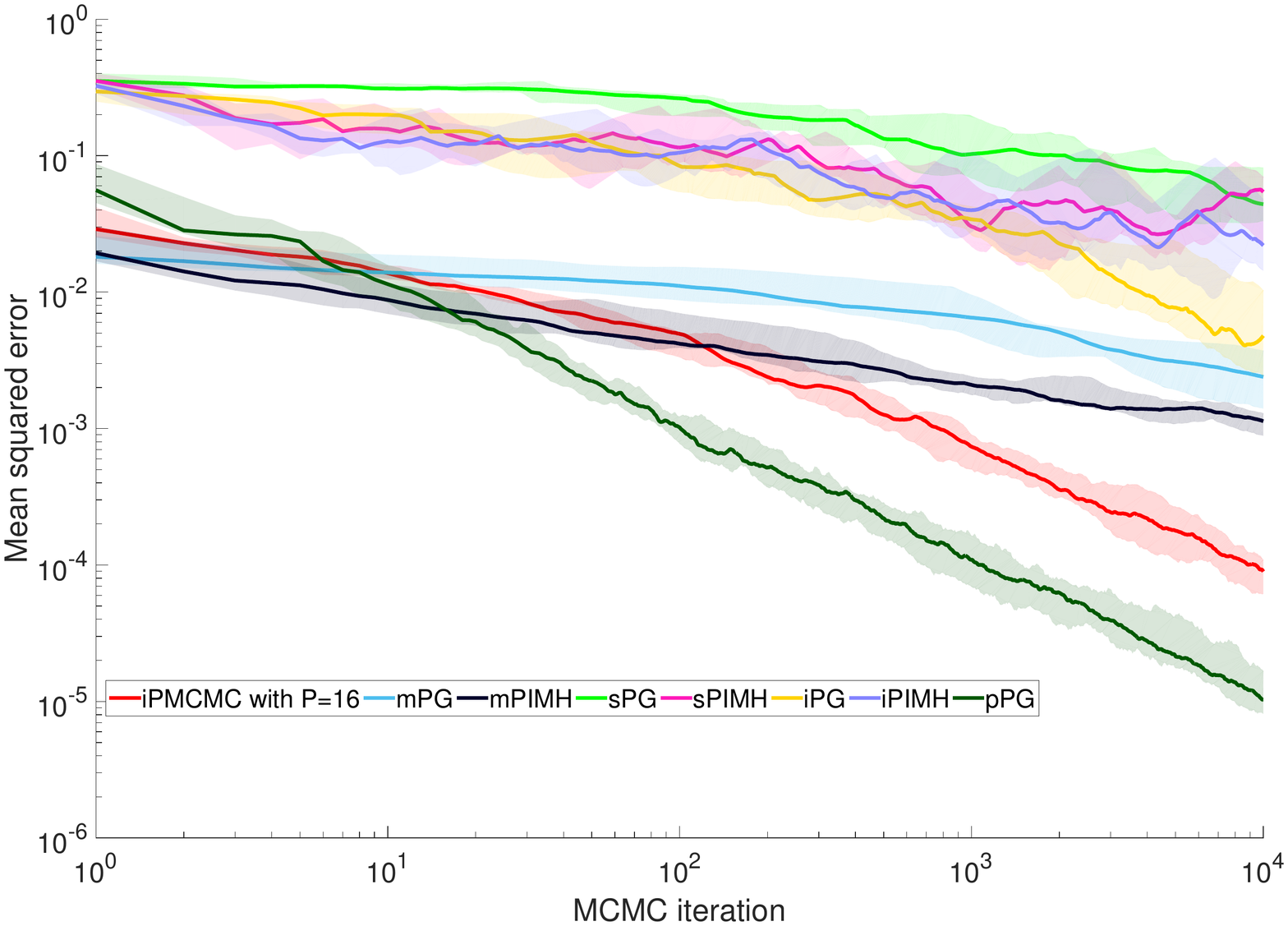}}
		\caption{Convergence in mean}
		\label{fig:supp-mean_series_conv}
	\end{subfigure}
	~ %add desired spacing between images, e. g. ~, \quad, \qquad, \hfill etc. 
	%(or a blank line to force the subfigure onto a new line)
	\begin{subfigure}[t]{0.49\textwidth}
		\makebox[\textwidth][l]{\includegraphics[width=1.1\textwidth,trim={4cm 0 5cm 0},clip]{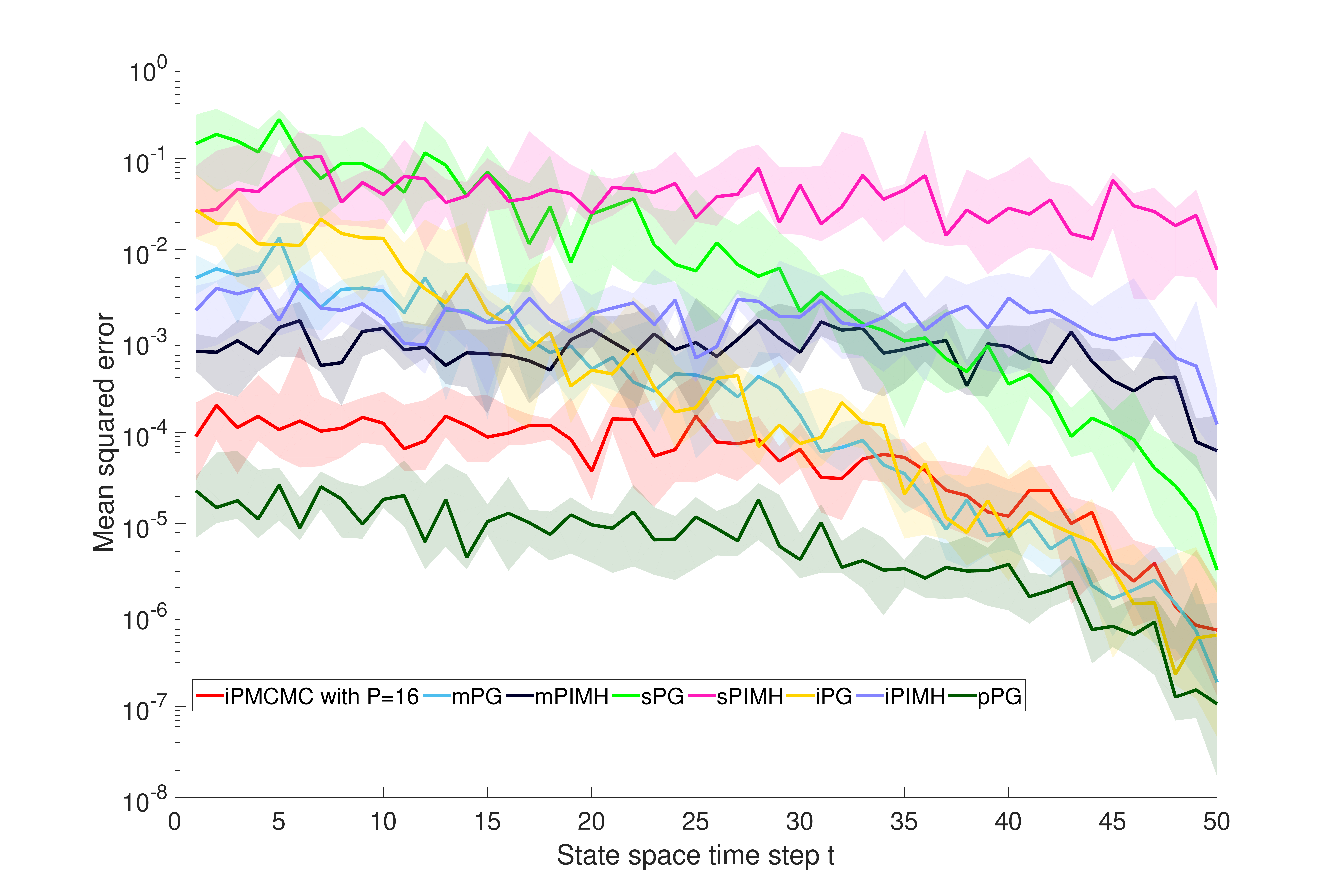}}
		\caption{Final error in mean}
		\label{fig:supp-mean_series_pos}
	\end{subfigure}
	
	\begin{subfigure}[t]{0.49\textwidth}
		\makebox[\textwidth][r]{\includegraphics[width=1.1\textwidth,trim={2cm 0.7cm 2cm 0},clip]{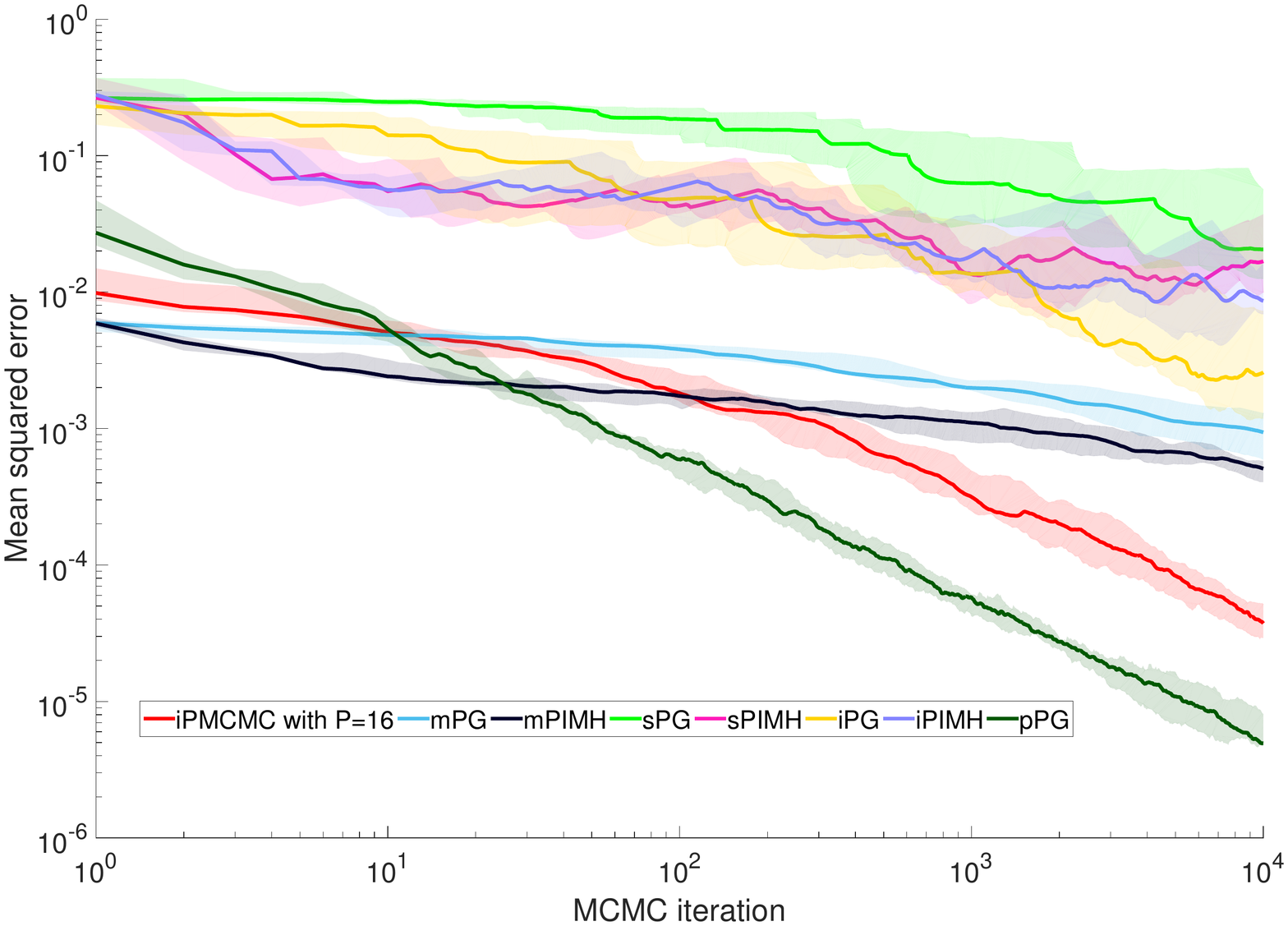}}
		\caption{Convergence in standard deviation}
		\label{fig:supp-std_series_conv}
	\end{subfigure}
	~ %add desired spacing between images, e. g. ~, \quad, \qquad, \hfill etc. 
	%(or a blank line to force the subfigure onto a new line)
	\begin{subfigure}[t]{0.49\textwidth}
		\makebox[\textwidth][l]{\includegraphics[width=1.1\textwidth,trim={4cm 0 5cm 0},clip]{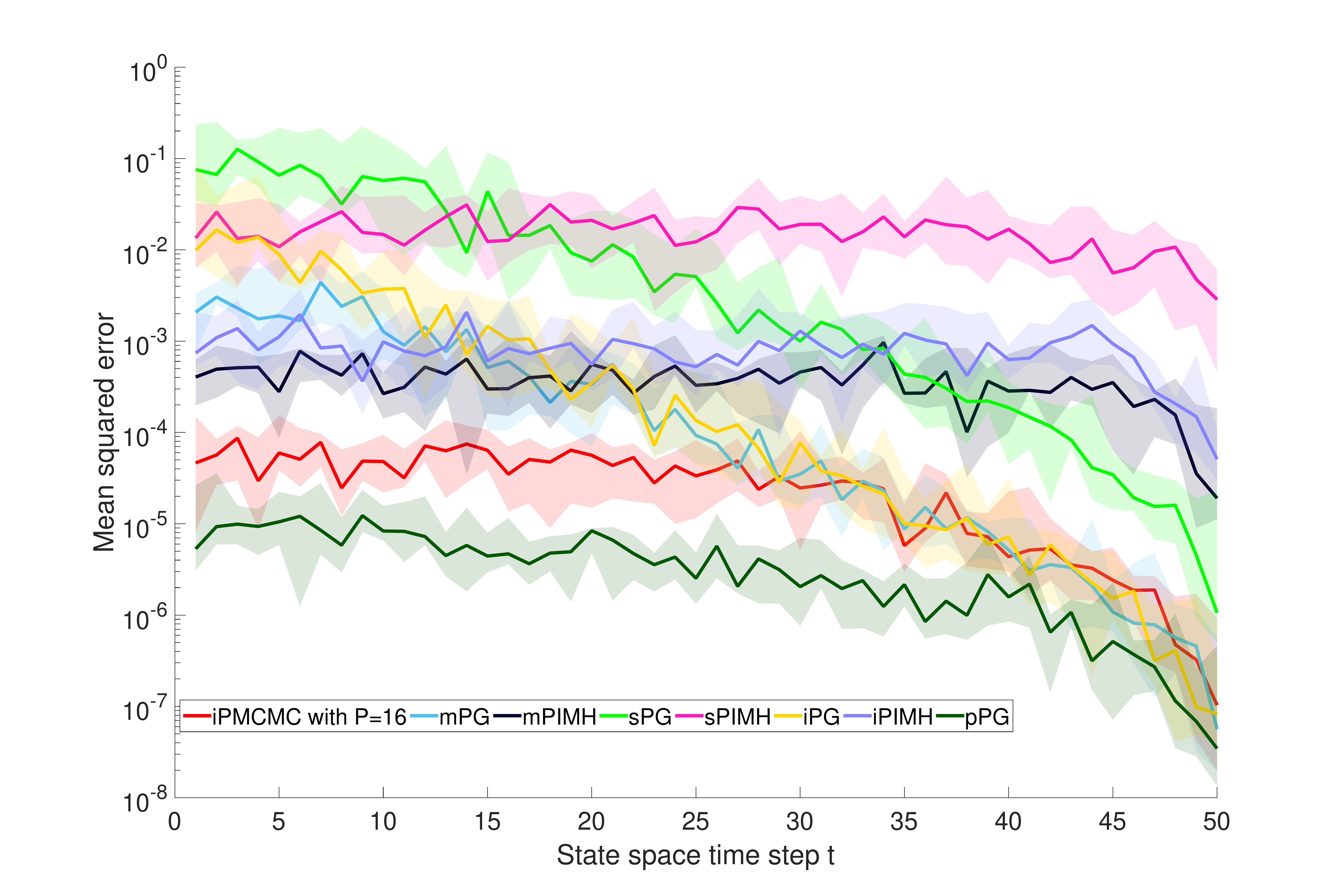}}
		\caption{Final error in standard deviation}
		\label{fig:supp-std_series_pos}
	\end{subfigure}	
	
	\caption{Mean squared error in latent variable mean and standard deviation averaged of all dimensions of the LGSSM as a function of MCMC iteration (left) and position in the state sequence (right) for iPMCMC, mPG, mPIMH and a number of serialized variants.  Key for legends: sPG = single PG chain, sPIMH = single PIMH chain, iPG = single PG chain run 32 times longer, iPIMH = single PIMH chain run 32 times longer and pPG = single PG with 32 times more particles.  For visualization purposes, the chains with extra iterations have had the number of MCMC iterations normalized by 32 so that the different methods represent equivalent total computational budget. \label{fig:supp-series_error_lss}}
\end{figure*}

\begin{figure*}[p]
	\centering
	\begin{subfigure}[t]{0.49\textwidth}
		\makebox[\textwidth][r]{\includegraphics[width=1.1\textwidth,trim={4cm 0 5cm 0},clip]{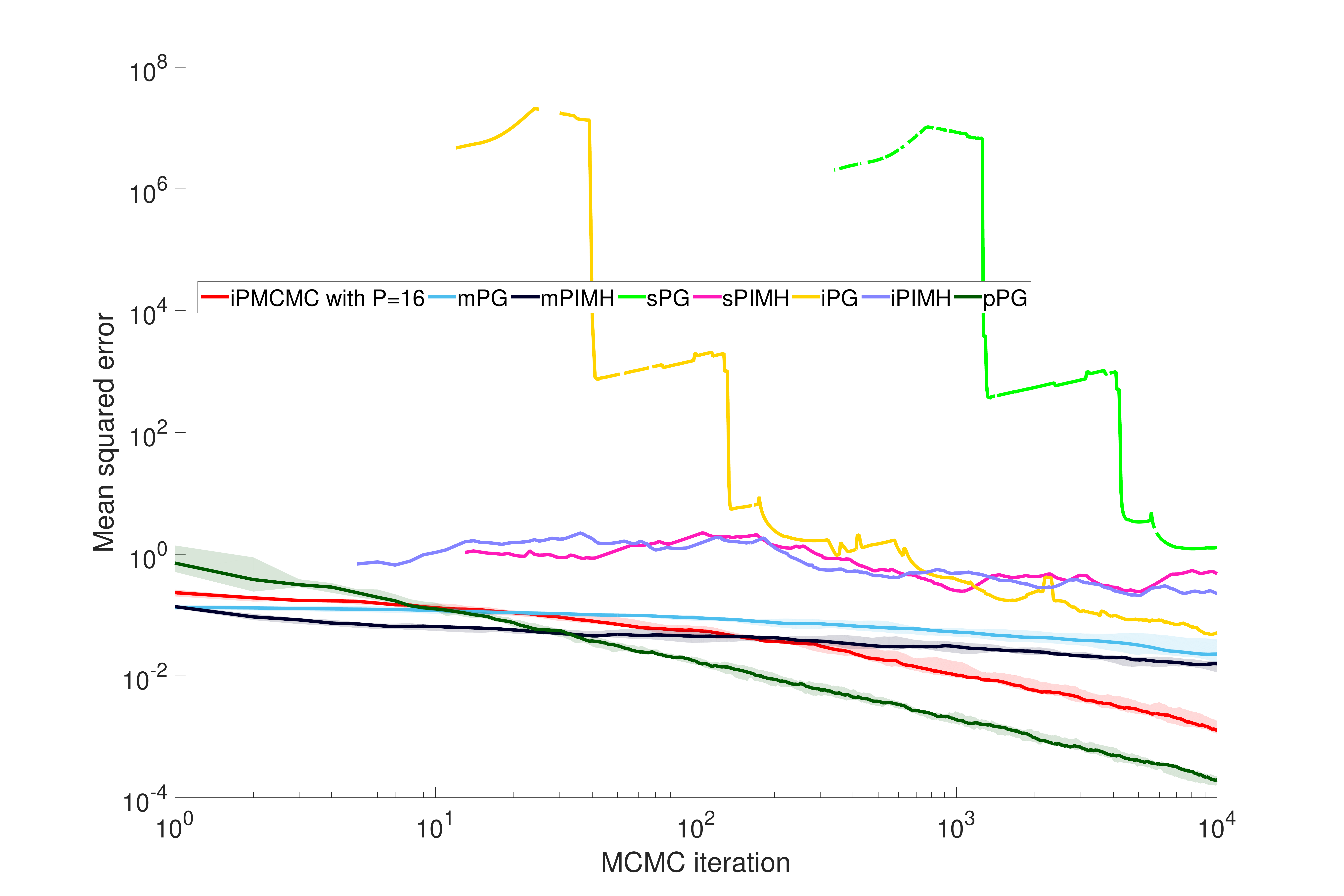}}
		\caption{Convergence in skewness}
		\label{fig:supp-ske_series_conv}
	\end{subfigure}
	~ %add desired spacing between images, e. g. ~, \quad, \qquad, \hfill etc. 
	%(or a blank line to force the subfigure onto a new line)
	\begin{subfigure}[t]{0.49\textwidth}
		\makebox[\textwidth][l]{\includegraphics[width=1.1\textwidth,trim={4cm 0 5cm 0},clip]{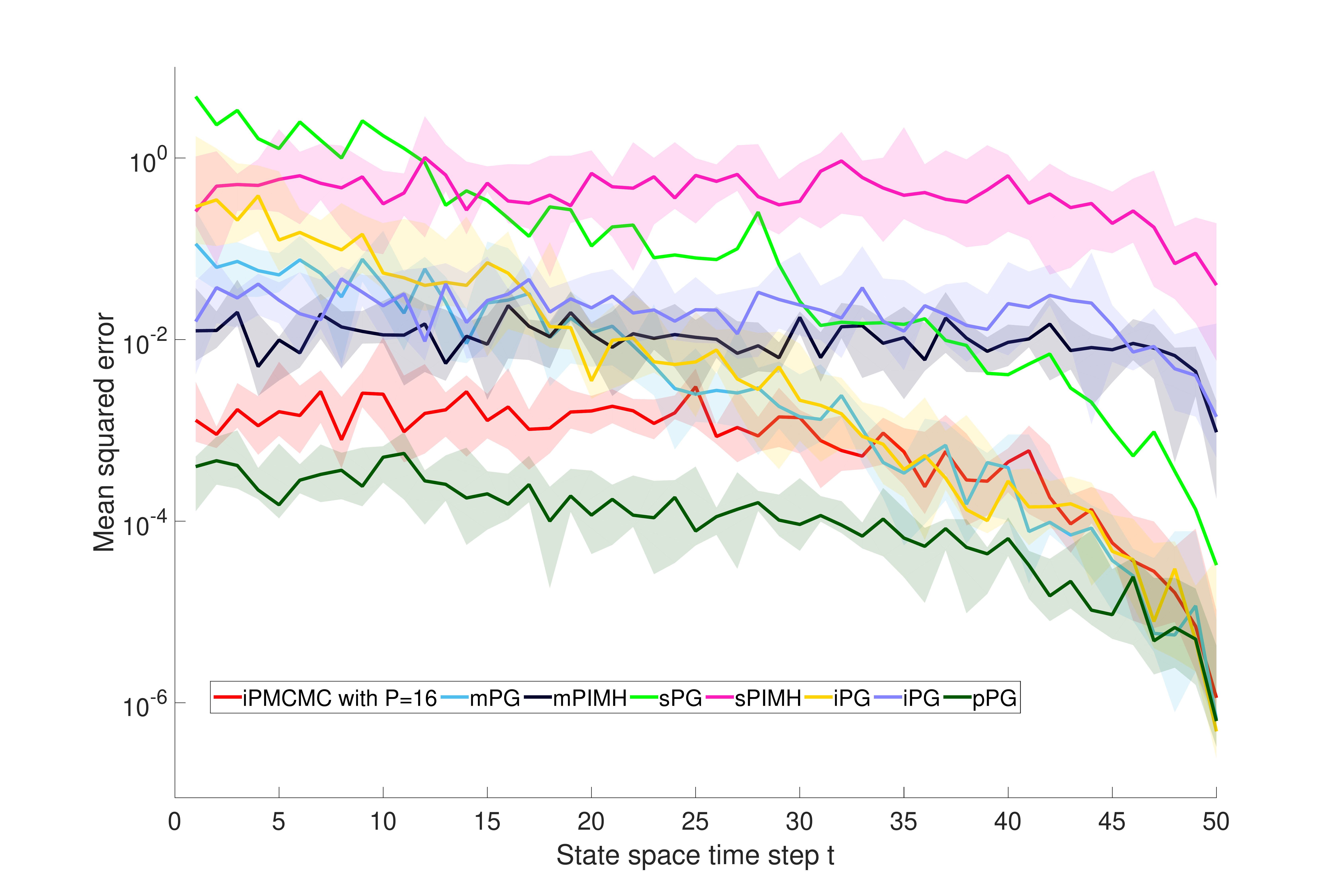}}
		\caption{Final error in skewness}
		\label{fig:supp-ske_series_pos}
	\end{subfigure}	
	
	\begin{subfigure}[t]{0.49\textwidth}
		\makebox[\textwidth][r]{\includegraphics[width=1.1\textwidth,trim={4cm 0 5cm 0},clip]{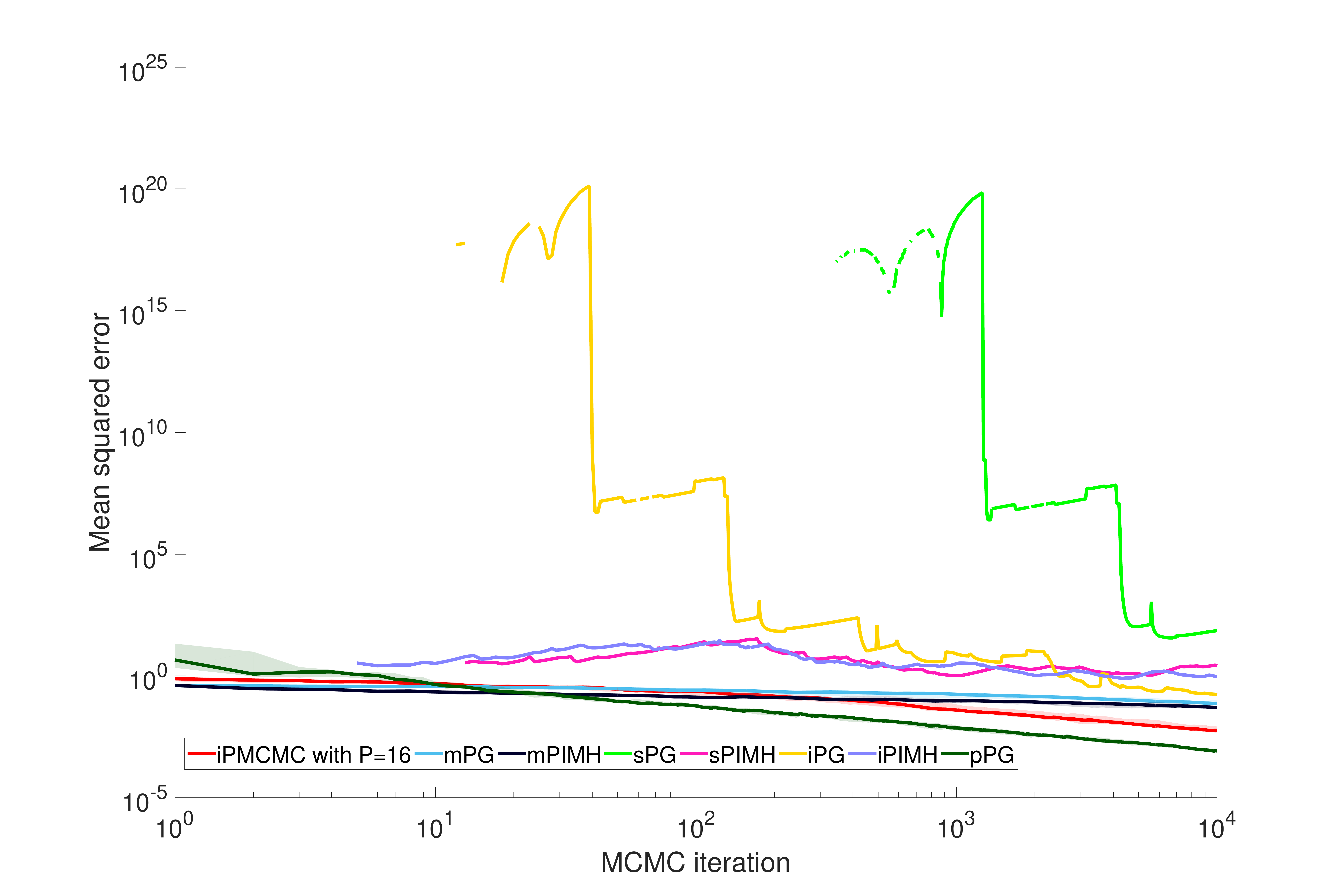}}
		\caption{Convergence in kurtosis}
		\label{fig:supp-kurt_series_conv}
	\end{subfigure}
	~ %add desired spacing between images, e. g. ~, \quad, \qquad, \hfill etc. 
	%(or a blank line to force the subfigure onto a new line)
	\begin{subfigure}[t]{0.49\textwidth}
		\makebox[\textwidth][l]{\includegraphics[width=1.1\textwidth,trim={4cm 0 5cm 0},clip]{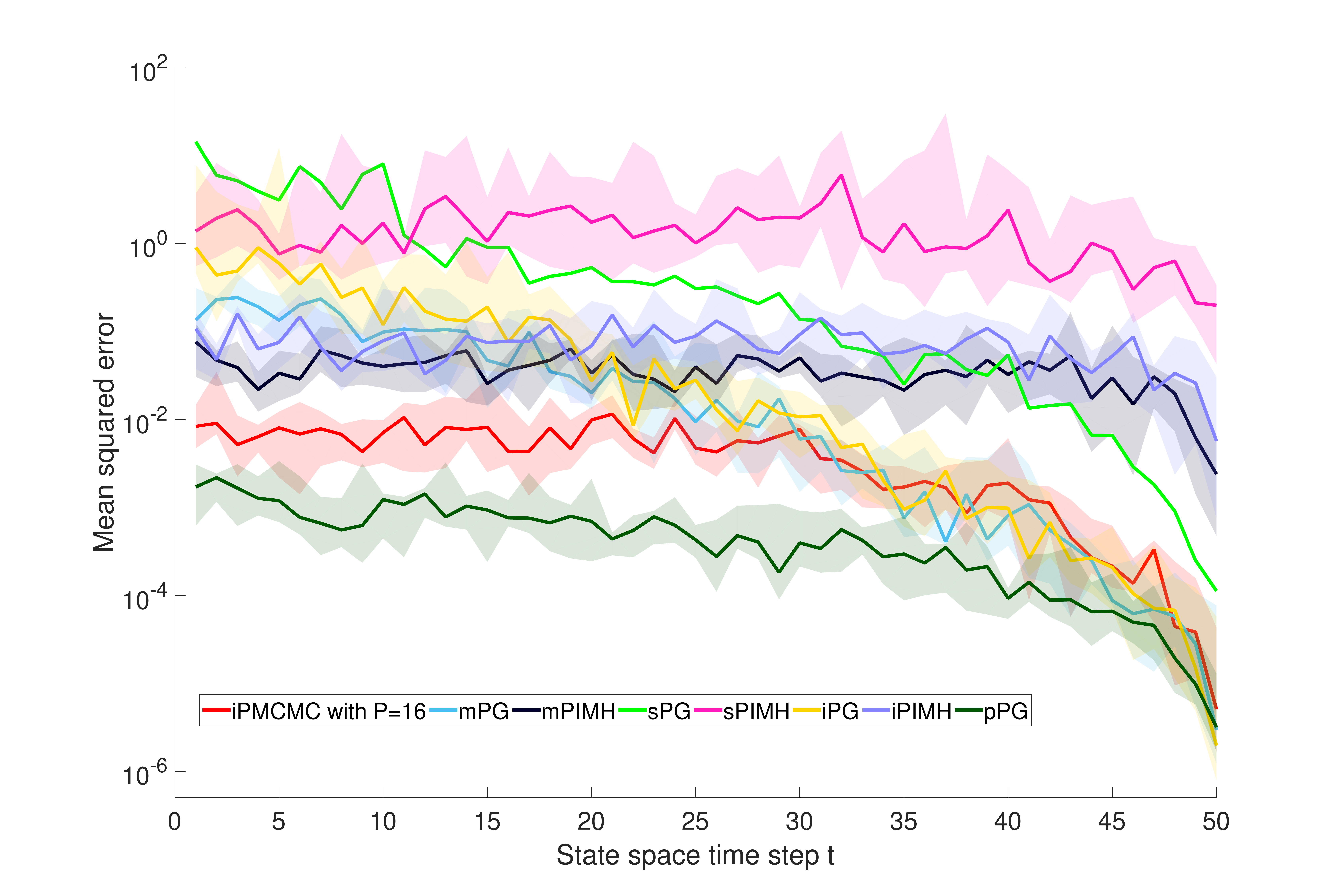}}
		\caption{Final error in kurtosis}
		\label{fig:supp-kurt_series_pos}
	\end{subfigure}

	\caption{Mean squared error in latent variable skewness and kurtosis averaged of all dimensions of the LGSSM as a function of MCMC iteration (left) and position in the state sequence (right) for iPMCMC, mPG, mPIMH and a number of serialized variants.  Key for legends: sPG = single PG chain, sPIMH = single PIMH chain, iPG = single PG chain run 32 times longer, iPIMH = single PIMH chain run 32 times longer and pPG = single PG with 32 times more particles.  For visualization purposes, the chains with extra iterations have had the number of MCMC iterations normalized by 32 so that the different methods represent equivalent total computational budget. \label{fig:supp-series_error_lss2}}
\end{figure*}

\begin{figure*}[p]
	\centering
	\begin{subfigure}[t]{0.49\textwidth}
		\makebox[\textwidth][r]{\includegraphics[width=1.1\textwidth,trim={4cm 0 5cm 0},clip]{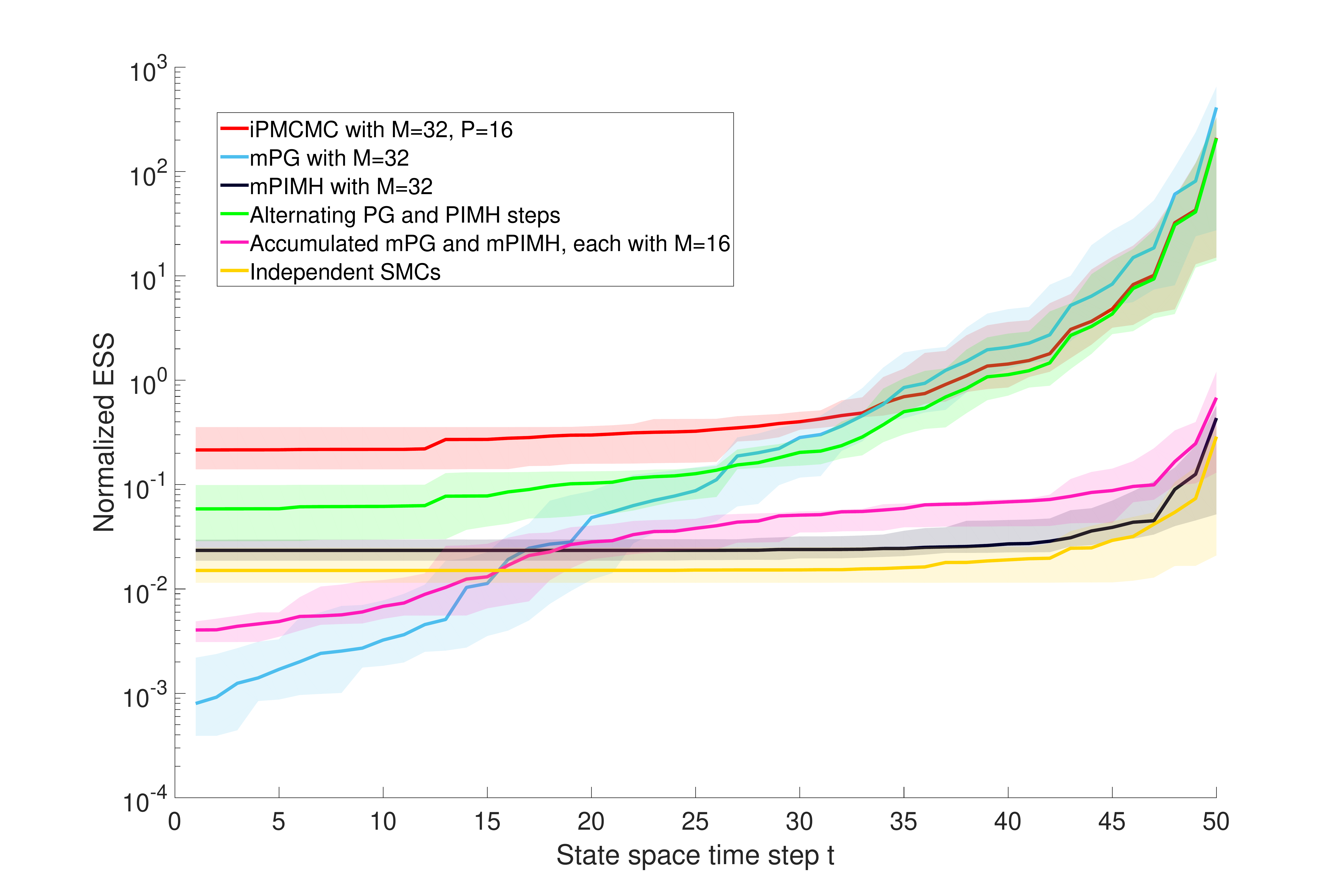}}
		\caption{ESS of distributed methods for LGSSM}
		\label{fig:supp-ESS-lss-dist}
	\end{subfigure}
	~ %add desired spacing between images, e. g. ~, \quad, \qquad, \hfill etc. 
	%(or a blank line to force the subfigure onto a new line)
	\begin{subfigure}[t]{0.49\textwidth}
		\makebox[\textwidth][l]{\includegraphics[width=1.1\textwidth,trim={4cm 0 5cm 0},clip]{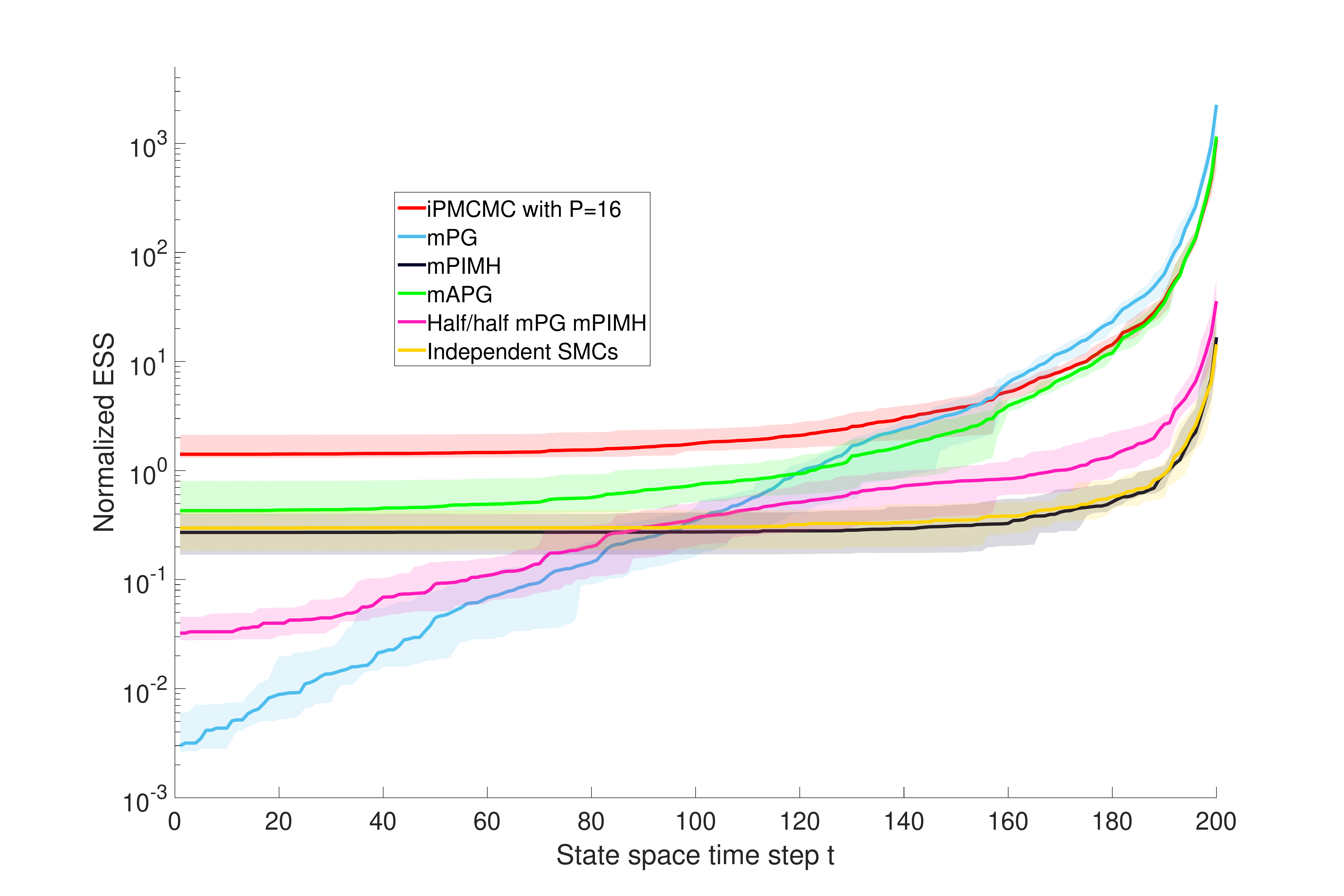}}
		\caption{ESS of distributed methods for NLSSM}
		\label{fig:supp-ESS-nlss-dist}
	\end{subfigure}
	
	\begin{subfigure}[t]{0.49\textwidth}
		\makebox[\textwidth][r]{\includegraphics[width=1.1\textwidth,trim={4cm 0 5cm 0},clip]{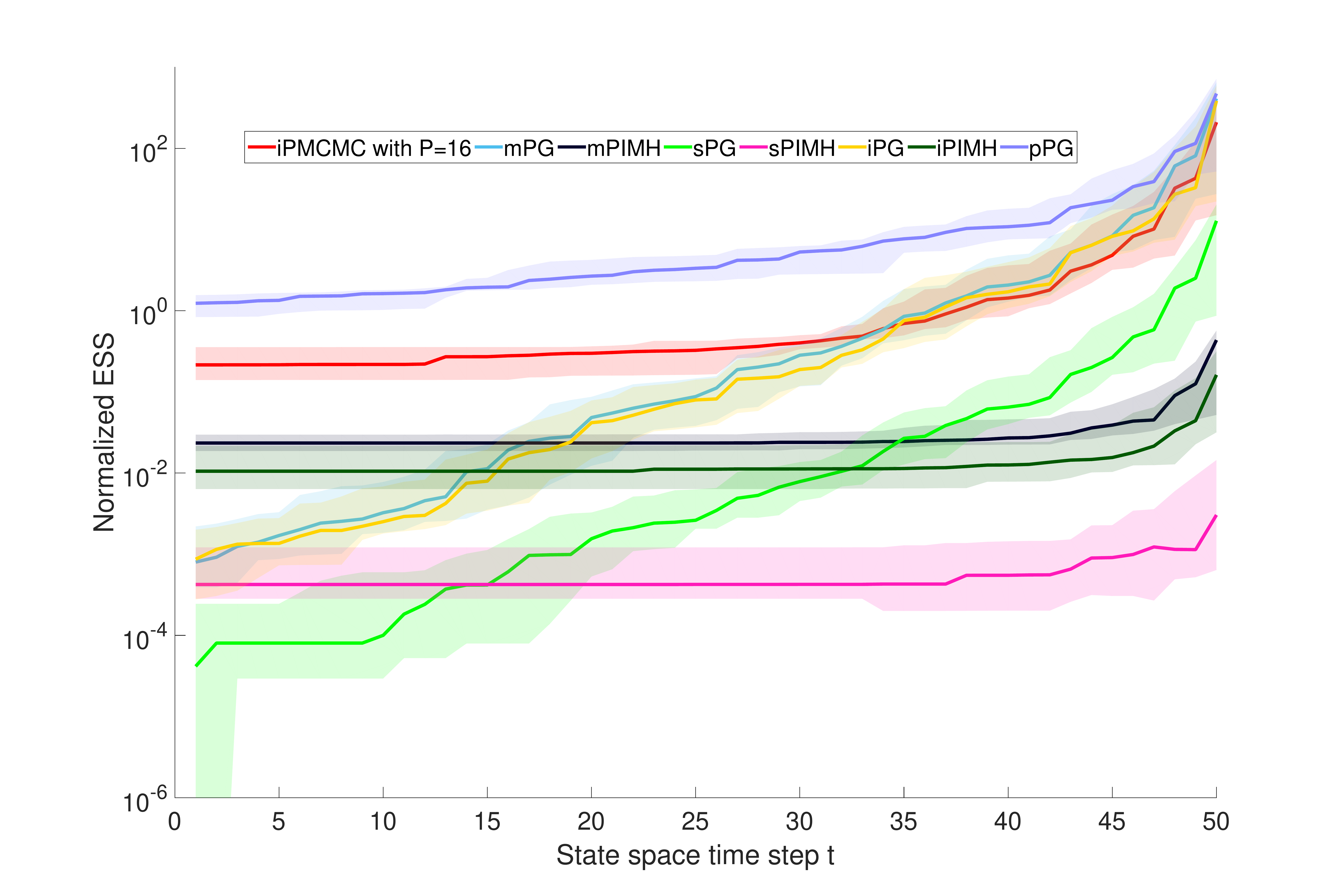}}
		\caption{ESS comparison to series equivalents for LGSSM}
		\label{fig:supp-ESS-lss-series}
	\end{subfigure}
	~ %add desired spacing between images, e. g. ~, \quad, \qquad, \hfill etc. 
	%(or a blank line to force the subfigure onto a new line)
	\begin{subfigure}[t]{0.49\textwidth}
		\makebox[\textwidth][l]{\includegraphics[width=1.1\textwidth,trim={4cm 0 5cm 0},clip]{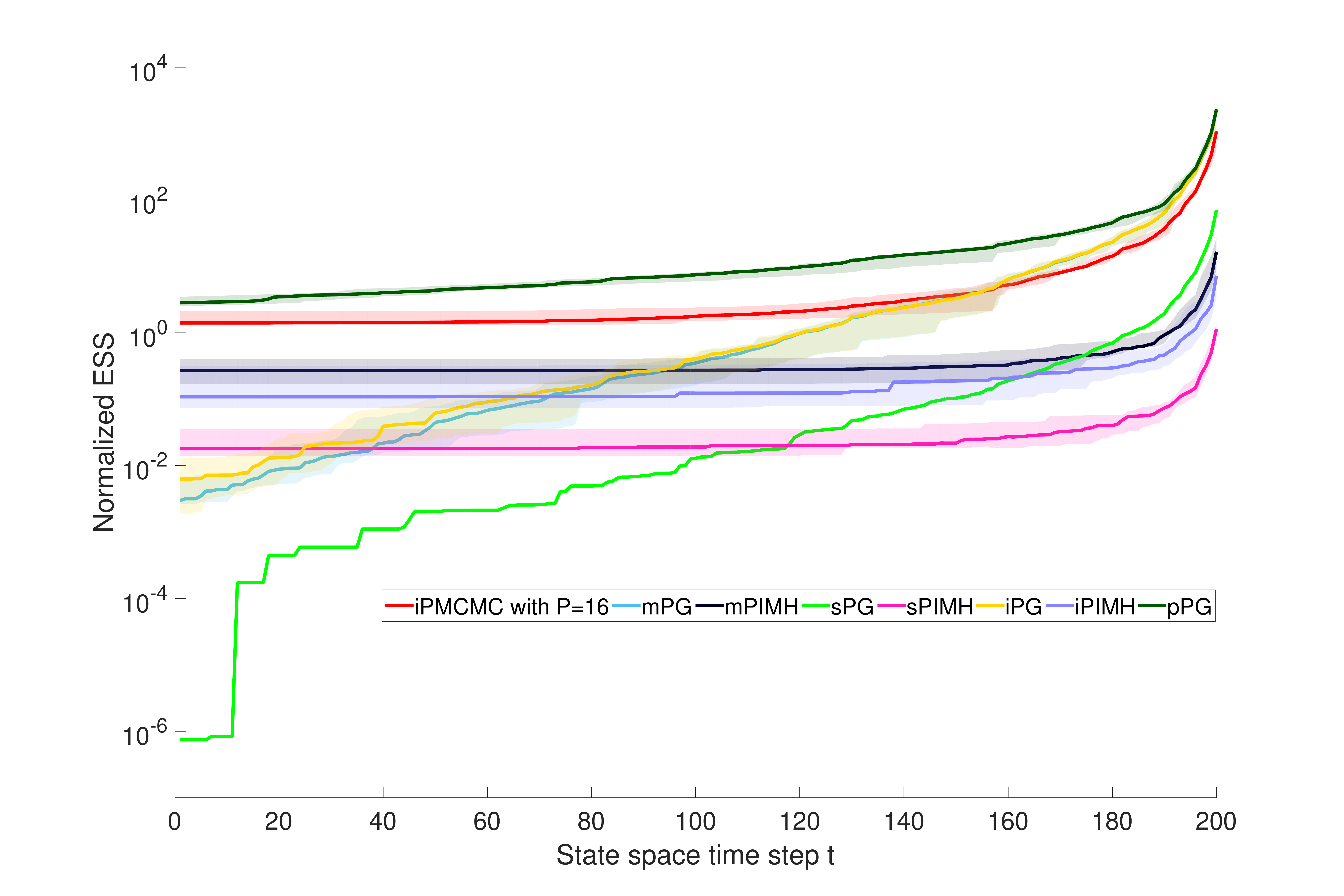}}
		\caption{ESS comparison to series equivalents for NLSSM}
		\label{fig:supp-ESS-nlss-series}
	\end{subfigure}	
	
	\caption{Normalized effective sample size for LGSSM (left) and NLSSM (right) for a number of distributed and series models.  Key for legends: sPG = single PG chain, sPIMH = single PIMH chain, iPG = single PG chain run 32 times longer, iPIMH = single PIMH chain run 32 times longer and pPG = single PG with 32 times more particles. \label{fig:supp-ess_extras}}
\end{figure*}